\documentclass[figure,letterpaper,onefignum]{mysiam}	
\usepackage{amssymb}
\usepackage{amsmath} 
\usepackage{float}
\usepackage{pifont}
\usepackage{graphicx}
\newcommand{\sm}[1]{\mathcal{#1}}
\newcommand{\bo}[1]{\mathbf{#1}}
\newcommand{\bx}{\mathbf{x}}

\newcommand{\floor}[1]{\left\lfloor#1\right\rfloor}
\newcommand{\ceil}[1]{\left\lceil#1\right\rceil}
\newcommand{\real}{\mathbb{R}}

\newcommand{\slugmaster}{%
\slugger{MMedia}{xxxx}{xx}{x}{x-x}}

\DeclareMathOperator{\prob}{Prob}
\DeclareMathOperator{\cov}{Cov}
\DeclareFontFamily{U}{msb}{}
\DeclareFontShape{U}{msb}{m}{n}{ <5> <6> <7> <8> <9> gen * msbm
  <10> <10.95> <12> <14.4> <17.28> <20.74> <24.88> msbm10}{} 
\DeclareSymbolFont{AMSb}{U}{msb}{m}{n}
\DeclareMathSymbol{\E}{\mathalpha}{AMSb}{"45}
\newcommand{\zig}{\Pisymbol{psy}{126}}
\title{A random walk on image patches
\thanks{This work was partially supported by
    National Science Foundation Grants DMS 0941476, ECS 0501578, and a
    contract from Sandia National Laboratories.}}
\author{Kye M. Taylor \thanks{Department of Applied Mathematics, 526
    UCB University of Colorado, Boulder, CO 80309-0526
    (\email{taylorkm@colorado.edu}).} \and Fran\c{c}ois G. Meyer
  \thanks{Department of Electrical Engineering, 425 UCB, University of
    Colorado, Boulder, CO 80309-0425 (\email{fmeyer@colorado.edu}).}}
\begin{document}
\maketitle
\begin{abstract}
In this paper we address the problem of understanding the success of
algorithms that organize patches according to graph-based metrics.
Algorithms that analyze patches extracted from images or time series
have led to state-of-the art techniques for classification, denoising,
and the study of nonlinear dynamics. The main contribution of this
work is to provide a theoretical explanation for the above
experimental observations. Our approach relies on a detailed analysis
of the commute time metric on prototypical graph models that epitomize
the geometry observed in general patch graphs. We prove that a
parametrization of the graph based on commute times shrinks the mutual
distances between patches that correspond to rapid local changes in
the signal, while the distances between patches that correspond to
slow local changes expand. In effect, our results explain why the
parametrization of the set of patches based on the eigenfunctions of
the Laplacian can concentrate patches that correspond to rapid local
changes, which would otherwise be shattered in the space of
patches. While our results are based on a large sample analysis,
numerical experimentations on synthetic and real data indicate that
the results hold for datasets that are very small in practice.
\end{abstract}
\pagestyle{myheadings}
\thispagestyle{plain}
\markboth{K.~M. TAYLOR AND F.~G. MEYER}{A RANDOM WALK ON IMAGE PATCHES}
\section{Introduction}
\paragraph{Problem statement and motivation} In this paper we address
the problem of understanding the success of algorithms that organize
patches according to graph-based metrics. Patches are local portions,
or snippets, of a signal or an image. The set of patches can be
organized by constructing a graph that connects patches that are
similar.  Indeed, it is reasonably straightforward to measure the
similarity between patches that are alike. The graph can then be used
to extend the notion of similarity to patches that are very
different. For instance, one can measure the distance between two visually
different patches by computing the number of edges of the shortest
path (geodesic) connecting them. In this work we explore a
distance defined by the commute time associated with a random walk
defined on the graph.

Algorithms that analyze patch data using graph-based metrics have led
to state-of-the art techniques for classification
\cite{Shen2008886,TaylorMeyer10}, denoising
\cite{bougleux09,Buades05,Gilboa08,Katkovnik10,peyre08,SingerBoaz,Szlam08},
and studying dynamics \cite{borges07,lac08,PhysRevLett.96.238701}.
The graph provides a new perspective from which to analyze the
similarities between patches, and consequently, the local signal or
image content they contain. For example, in \cite{borges07},
properties of the graph's geometry, such as the distribution of
clustering-coefficients and the average geodesic distance between two
vertices, are used to separate chaos and noise, or different types of
chaos. In \cite{Shen2008886,SingerBoaz,Szlam08,TaylorMeyer10}, the
geometry of the graph is analyzed by studying a random walk on
it. Specifically, the \textit{diffusion distance} \cite{Coifman06a}
(or \textit{spectral distance} \cite{Berard94}), and the
\textit{commute time distance} \cite{Bremaud99} (which is equivalent
to the \textit{resistance distance} \cite{citeulike:5942798}) are two
related graph metrics that are derived from the random walk, and that
can be used to parametrize the graph's geometry. These metrics can be
used to efficiently organize patches in a manner that reveals the
local behavior of the associated signal or image. In our previous work
\cite{Shen2008886,TaylorMeyer10}, we have noticed that metrics based
on a diffusion, or a random walk concentrate patches that contain
rapid changes in the signal or image data. These patches contain
changes associated with singularities (edges), rapid changes in
frequency (textures, oscillations), or energetic transients contained
in the underlying function. Furthermore, patches that contain only the
smooth parts of the image are more spread out according to such graph
metrics.

\paragraph{Outline of our approach and results} The main contribution
of this work is to provide a theoretical explanation for the above
experimental observations. Our approach relies on a detailed analysis
of the commute time metric on prototypical graph models that epitomize
the geometry observed in general patch-graphs. We assume that the set
of patches is composed of two broad classes: patches within which the
function varies smoothly and slowly, and patches where the function
exhibits anomalies: singularities, very rapid change in local
frequency, etc. We prove that a parametrization of the graph based on
commute times shrinks the mutual distances between vertices that
correspond to rapid local changes relative to the distances between
vertices that correspond to slow local changes. In effect, our results
explain why the parametrization of the set of patches based on the
eigenfunctions of the Laplacian \cite{SingerBoaz,Szlam08} can
concentrate anomalous patches, which would otherwise be shattered in
the space of patches. This concentration phenomenon can then be
exploited for further processing of the patches (e.g. denoising,
classification, etc). While our results are based on a large sample
analysis, numerical experimentations on synthetic and real data
indicate that the results hold for datasets that are very small in
practice.

\paragraph{Organization} This paper is organized as follows. In the
next section, we describe the patch-based representation of a signal,
and the associated patch-graph. We develop some intuition about the
graph of patches by studying several examples in section
\ref{ssec:firstlookpatchgraph}.  In section \ref{sec:patchgraphparam},
we describe the embedding of the graph of patches based on commute
time. The prototypical graph models that allow us to study the
parametrization are defined in section \ref{ssec:theory}. The main
theoretical result about the embedding of the graph models are
presented in section \ref{ss:ctestimates}. Numerical experiments
confirming our theoretical analysis are presented in section
\ref{sec:experiments}.  We finish with a discussion in section
\ref{sec:discuss}.
\section{Preliminaries and  Notation}
\label{sec:patchgraph}
For simplicity and without loss of generality, we assume that the
signal of interest is formed by a sequence of samples,
$\{x_n\}_{n = 1}^{N'}$. Because we want to extract $N= N'-(d-1)$
patches from this sequence, we need $d$ extra samples at the end
(hence the $N'$ samples). We first define the notion of a {\em patch}.\\

\begin{definition}
  We define a {\em patch} as a vector in $\real^d$ formed by a
  subsequence of $d$ contiguous samples extracted from the sequence $\{x_n\}_{n =1}^{N'}$,
  \begin{equation}
    \bx_n = 
    \begin{bmatrix}
      x_n & x_{n+1} &\ldots &x_{n+(d-1)}
    \end{bmatrix}      
    ^T,\hspace{.2in}\text{for }n = 1,2,\ldots, N.
    \label{eqn:apatch}
  \end{equation}
\end{definition} 
As we collect all the patches, we form the {\em patch-set} in $\real^d$.\\

\begin{definition}
  The  {\em patch-set} is defined as the set of patches extracted from
  the sequence $\{x_n\}_{n =1}^{N'}$,
  \begin{equation}
    \text{patch-set}=\{\bx_n, n = 1,2,\ldots, N\}.
    \label{eqn:patchset}
  \end{equation}
\end{definition} 
A main objective of this paper is to understand the organization of
the patch-set and relate this organization to the presence of local
changes in the signal or the image.  We note that the concept of patch is
related to the concept of time-delay embedding. Specifically, if the
sequence comprises measurements of a dynamical system, then Taken's
embedding theorem \cite{springerlink:10.1007/BF01053745,Takens81}
allows us to replace the unknown phase space of the dynamical system
with a topologically equivalent phase space formed by the patch-set
(\ref{eqn:patchset}). While in this work we do not assume that the
sequence $\{x_n\}$ is an observable of a dynamical system, we are
nevertheless interested in a similar goal: the organization of patches in $\real^d$.

Throughout this paper, we think about a patch, $\bo x_n$, in several
different ways. Originally, $\bx_n$ is simply a snippet of the time
series. Then, we think about $\bx_n$ as a point in $\real^d$. Later, we
also regard $\bx_n$ as a vertex of a graph.  Keeping these three
perspectives in mind is critical to our approach and understanding.

In order to study the discrete structure formed by the patch-set
(\ref{eqn:patchset}), we connect patches together (using
their nearest neighbors) and define a graph (or network) that we call
the {\em patch-graph}.\\

\begin{definition}
  The {\em patch-graph}, $\Gamma$, is a weighted graph defined as follows.
  \begin{remunerate}
  \item The vertices of $\Gamma$ are the patches $\{ \bx_n, n =1,\ldots, N\}$.
  \item Each vertex $\bx_n$ is connected to its $\nu$ nearest neighbors using
    the metric 
    \begin{equation}
      \rho(\bx_n,\bx_m) = \left\|\tfrac{\bx_n}{\|\bx_n\|}
        - \tfrac{\bx_m}{\|\bx_m\|}\right\|.
      \label{rho}
    \end{equation}
  \item The weight $w_{n,m}$ along the edge $\{\bx_n,\bx_m\}$ is given by
    \begin{equation}
      w_{n,m}= 
      \begin{cases}
        e^{\displaystyle -\rho^2(\bx_n,\bx_m)/\sigma^2} & \text{ if $\bx_n$ is connected to $\bx_m$,}\\
        0 &\text{ otherwise.}
      \end{cases}
      \label{eqn:weightdef}
    \end{equation}
  \end{remunerate}
\end{definition}
The edges of the patch-graph encode the similarities between its $N$
vertices.  We work with the metric $\rho$ (defined in (\ref{rho}))
because it is not sensitive to changes in the local energy of the
signal (measured by $\|\bx_n\|$). The metric $\rho$ allows us to
detect changes in the signal's local frequency content, or local
smoothness.  The parameter $\sigma$ controls the scaling of the
similarity $\rho(\bx_n,\bx_m)$ between $\bx_n$ and $\bx_m$ when
defining the edge weight $w_{n,m}$. In particular, $w_{n,m}$ will drop
rapidly to zero as $\rho(\bo{x}_n,\bo{x}_m)$ becomes larger than $\sigma$.

An important remark about the way we measure distances on the graph is
in order here. We use $\rho$ to define the graph topology defined by
the edges: which patch is connected to which patch. This is
appropriate since we can compare patches that are similar using $\rho$
(e.g. two patches containing the same edge, but at different
locations). On the other hand, as explained in section
\ref{sssec:commutetimes}, we use the commute time to analyze the
global geometry of the patch-graph.%
\begin{figure}[H]
\centerline{
    \includegraphics[width=.5\textwidth]{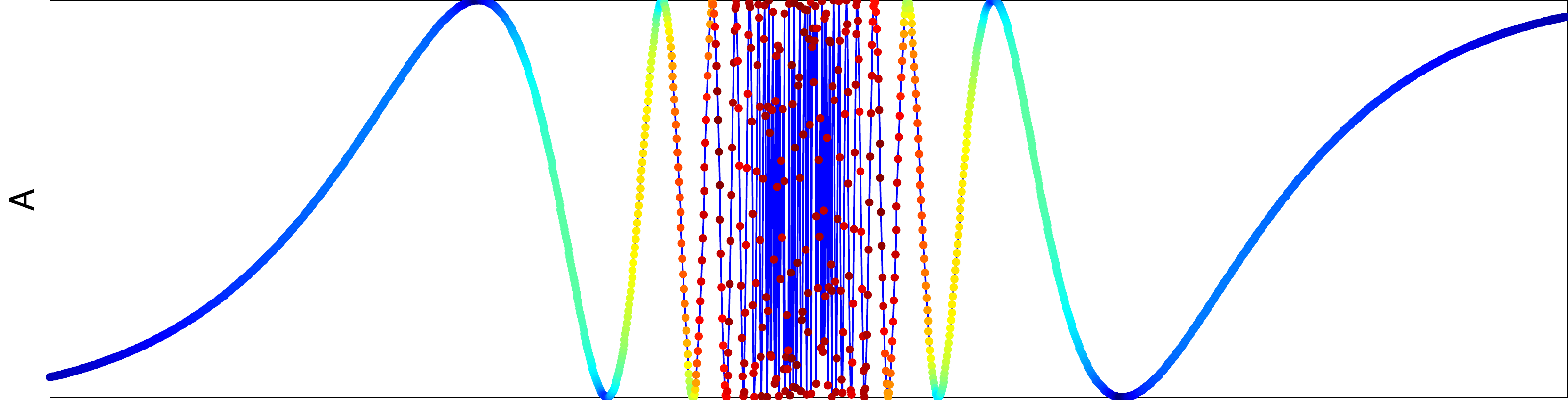} 
}
\centerline{
   \includegraphics[width=.5\textwidth]{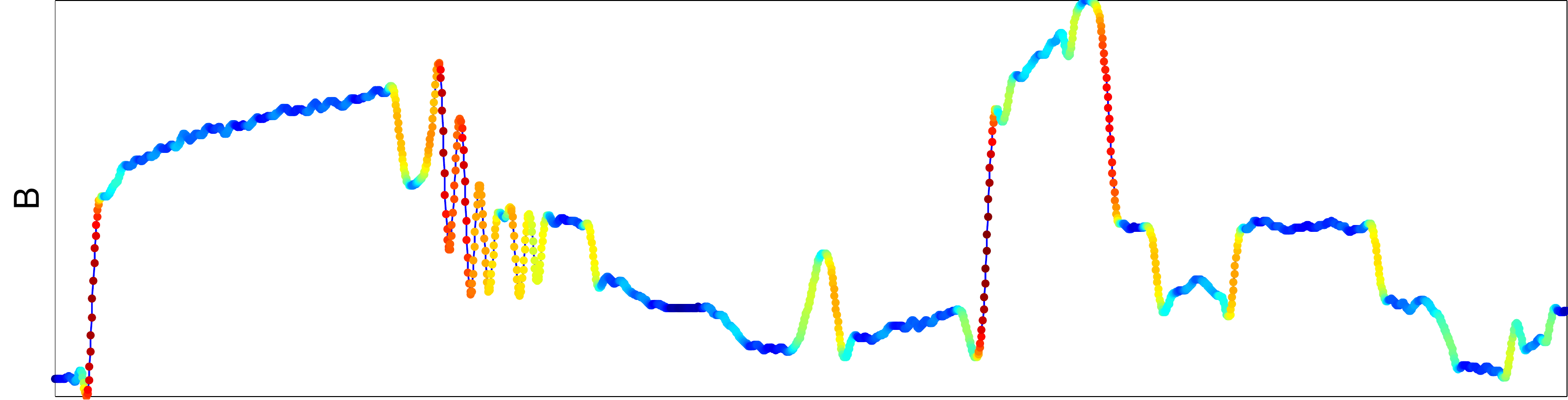} 
}\centerline{
    \includegraphics[width=.5\textwidth]{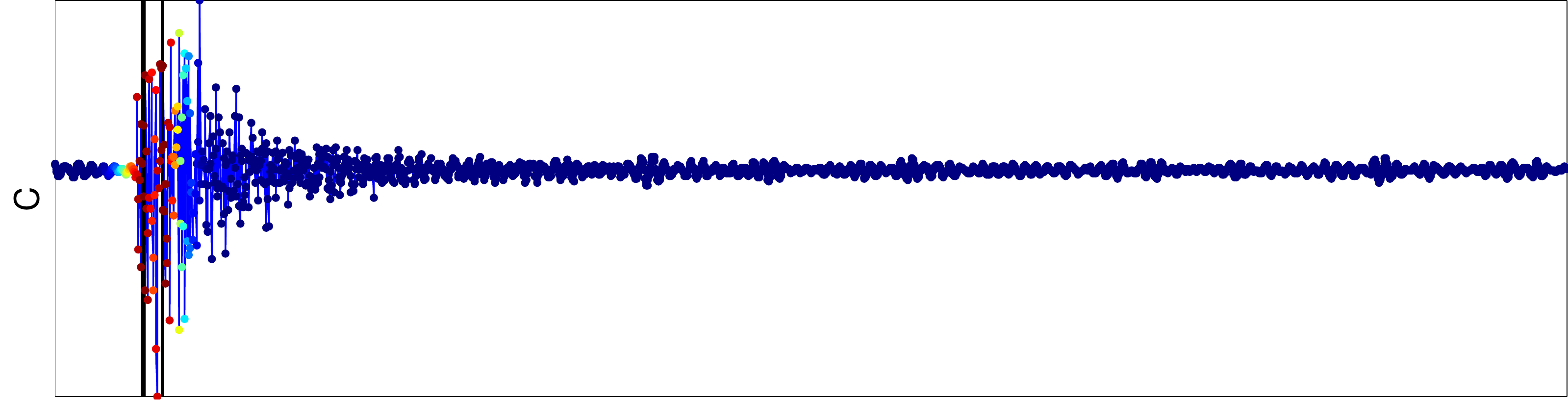} 
}
  \caption{A, B, C: time series composed of $N' = 2072$ samples.   The
    color of the signals A and B encodes the local variance (large = 
    red, low = blue). C: seismogram; the color indicates the
    temporal proximity to a seismic arrival, identified by vertical black lines. See text
    for more details. \label{fig:s2ts}}        
\end{figure}
\begin{figure}[H]
\centerline{
   \includegraphics[width=.27\textwidth
   ,height=.21\textheight]{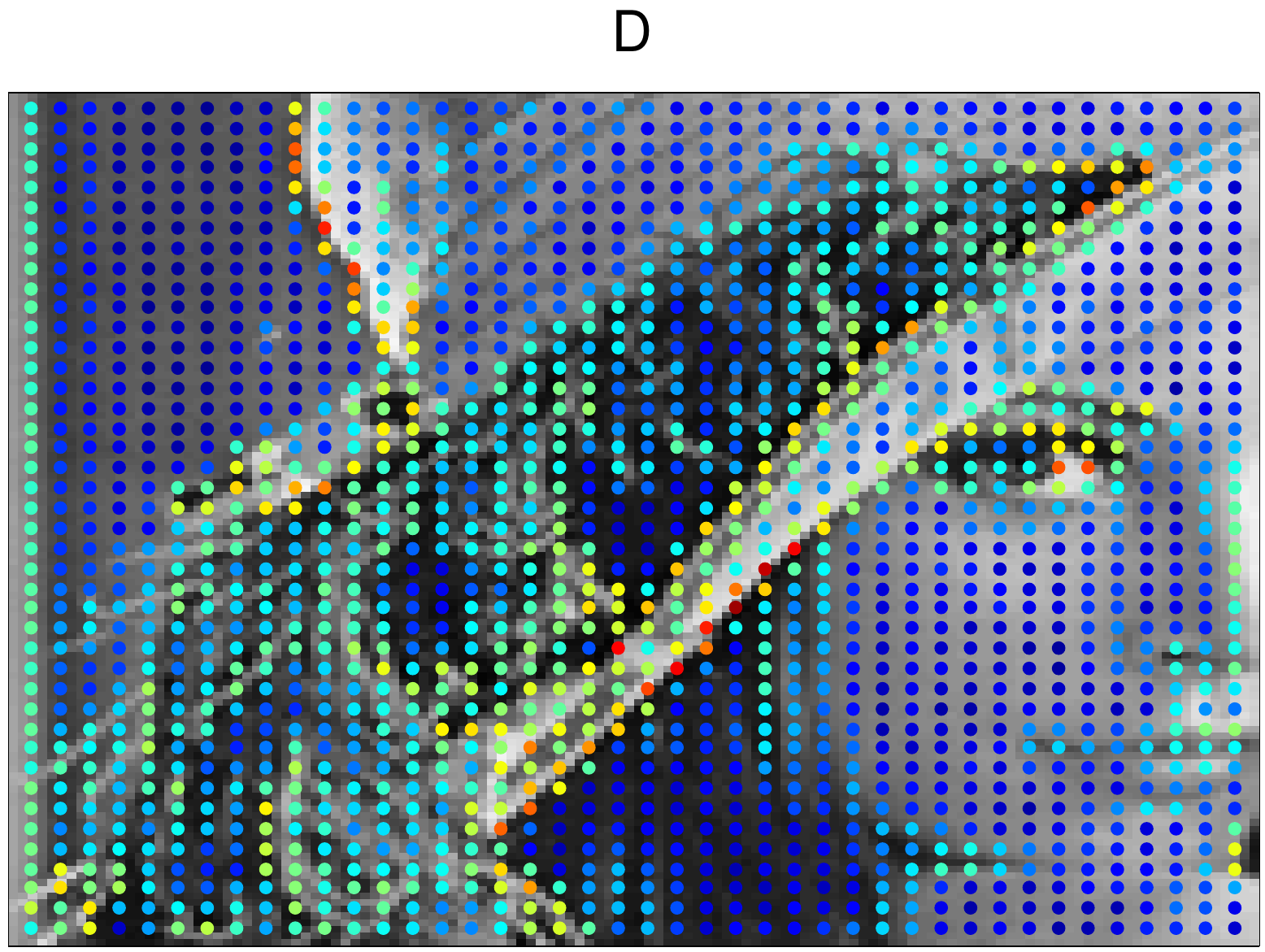} 
   \includegraphics[width=.27\textwidth
   ,height=.21\textheight]{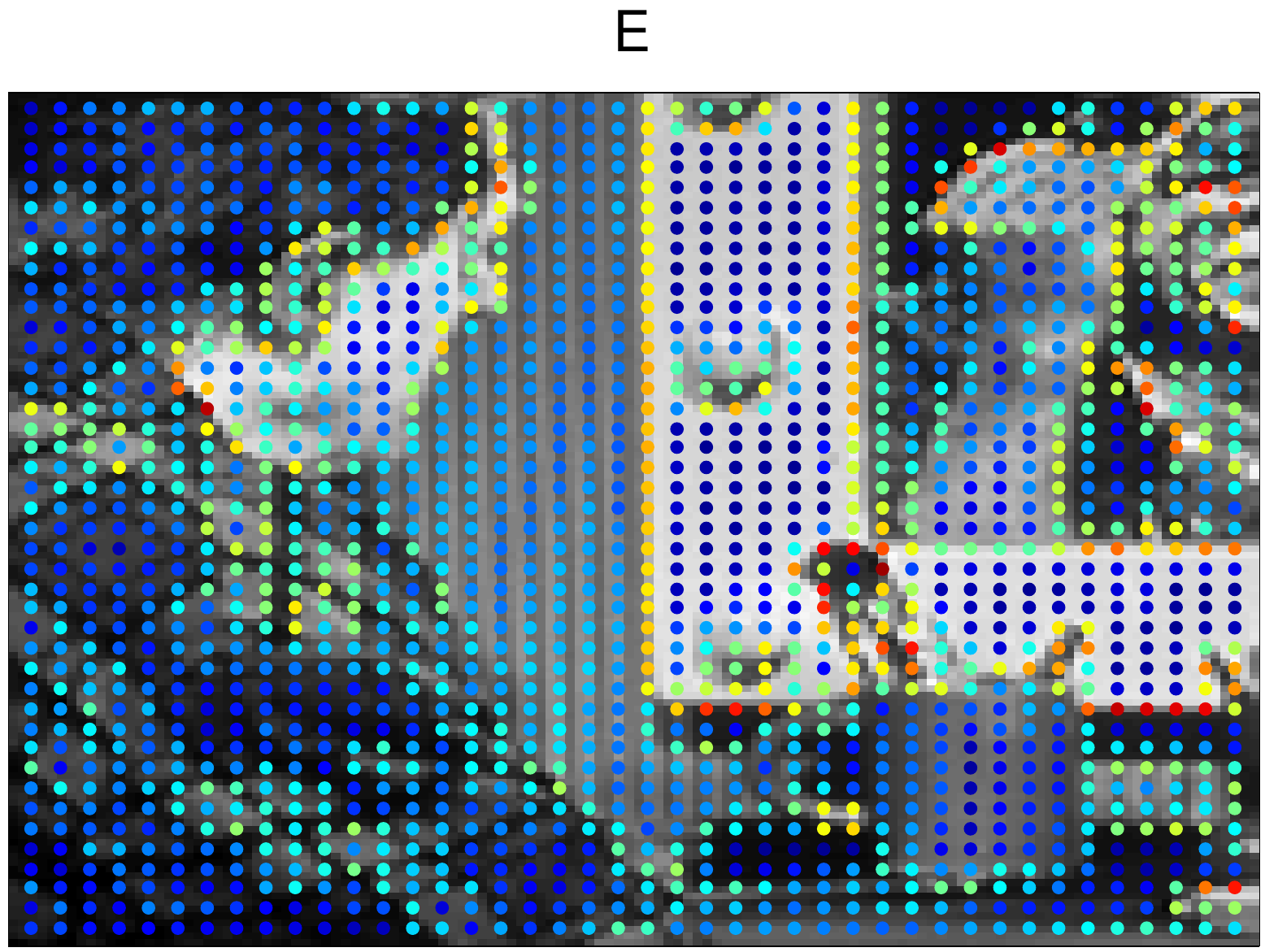} 
   \includegraphics[width=.27\textwidth,height=.21\textheight]{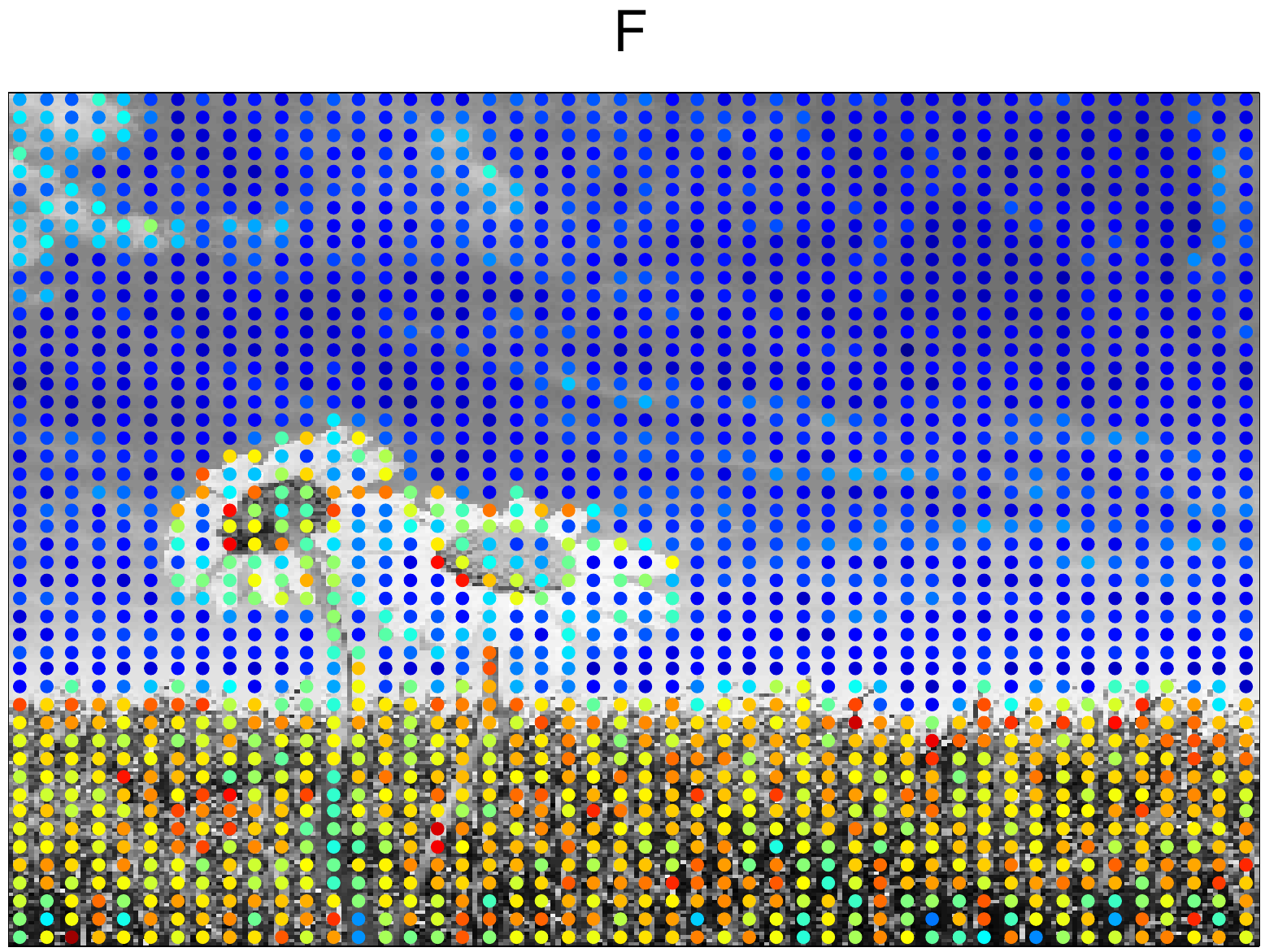}
}
  \caption{D, E, and F: image of size $128\times 128$, $128\times
    128$, and $240\times240$ pixels respectively. The color of the pixel
    at the center of each patch encodes the local variance of the
    image intensity.\label{fig:signalsec2im} }
\end{figure}
\noindent  Indeed, the distance defined by
$\rho$ becomes useless when we need to compare very different patches
(e.g. a patch of a uniform region vs a patch that contains an
edge). As explained in section \ref{sssec:commutetimes}, the
global organization of the patches can be discovered by studying the
speed at which a random walk propagates along the graph (via hitting
times). 

Finally, we note that the weighted graph is fully
characterized by its \textit{weight matrix}.\\

\begin{definition} The {\em weight matrix}
  $\bo{W}$ is the $N\times N$ matrix with entries $\bo{W}_{n,m} =
  w_{n,m}$. The {\em degree matrix} is the $N\times N$ diagonal matrix
  $\bo{D}$ with entries $\bo{D}_{n,n} = \sum_{l = 1}^N w_{n,l}$.
\end{definition}

\section{Warm up: A first look at the patch-set}
\label{ssec:firstlookpatchgraph}
The goal of this section is to provide the reader with some intuition
about the geometry of the patch-set and the associated
patch-graph. This will help us motivate our graph models and the
analysis of their geometry.  At the end of the section, we provide a
sketch of our plan of attack.
\subsection{Examples of signals and images}
\label{ssec:imsigexample}
We construct the patch-set associated with some examples of signals
and images. Because it is not practical to
visualize the patch-set in $\real^d$ when $d=25$, we display the
projection of the patch-set onto the three-dimensional space that
captures the largest variance in the patch-set (computed using principal
component analysis). Figure \ref{fig:s2ts} displays three signals
$\{x_n\}, n=1,\ldots,N'$, with $N'=2072$. Patches of size $d = 25$
samples are extracted around each time sample, which results in the
maximum overlap between patches. Signal A is a chirp, signal B is a
row of the image Lenna (shown in Fig. \ref{fig:signalsec2im}-D), and
signal C is a seismogram \cite{TaylorMeyer10}.

In order to quantify the local regularity of signals A and B, we
compute the variance over each patch, and color the curve according to
the magnitude of the local variance: hot (red) for large variance and
cold (blue) for low variance.  The color of signal C encodes the
temporal proximity to the arrival of a seismic wave associated with an
earthquake: hot color indicates close proximity, while cold
corresponds to baseline activity. Identifying arrival-times is
necessary for purposes such as locating an earthquake's epicenter.
This example illustrates the application of the present work to the
problem of detecting seismic waves \cite{TaylorMeyer10}.

Figure \ref{fig:signalsec2im} displays three images. We extract
patches of size $5 \times 5$. Here, the patches are not maximally
overlapping: we collect every third patch in the horizontal and
vertical directions for images D and E, while we collect every fifth
patch in each direction for image F. This results in patch-sets of
size $42\times 42$ for images D and E, and of size $48 \times 48$ for
image F. As before, the color of a pixel in the images encodes the local
variance within the patch centered at that pixel.
\subsection{Projections of the patch-sets}
Figure \ref{fig:pcasec2} shows the projections of each of the six
patch-sets.  Distances in Figure \ref{fig:pcasec2} correspond to the
normalized distance $\rho$. We observe that patches with high variance
(red-orange) appear to be scattered all over $\real^d$. These patches
correspond to regions where the image intensity varies
rapidly. Patches with low variance (blue-green), which correspond to
regions where the signal is smooth and varies very little, tend to be
concentrated along one-dimensional curves (for time series) and
two-dimensional surfaces (for images).  These visual observations can
be confirmed when computing the actual mutual distances between
patches (data not shown).

The organization of the patches in the patch-set can be explained
using simple arguments. Let us assume that the sequence $\{x_n\}$
corresponds to the sampling of an underlying differentiable function
$x(t)$, and assume that $x'(t)$, the derivative of $x(t)$, remains
small over the interval of interest. In this case, if two patches
$\bx_n$ and $\bx_m$ overlap significantly -- i.e. $|n -m|$ is small --
then they will be close to one another in $\real^d$. Indeed, the
values of the coordinates of patches $\bx_n$ and $\bx_m$ will be very
similar, since the signal $x(t)$ varies slowly. In principle, if the
sampling is fast enough, the patches should lie along a
one-dimensional curve in $\real^d$. By the same argument, when $x(t)$
exhibits rapid changes, the magnitude of the derivative, $|x'(t)|$,
can be very large, and therefore temporally neighboring patches are
not guaranteed to be spatial neighbors in $\real^d$. This argument
allows us to understand the  distribution of the patches in the signal
B, or the image F.

Instead of characterizing patches according to the local smoothness of
the underlying function, we can also analyze the distribution of the
patches according to the function's local frequency information. This
will help us understand the structure of the patch-set for signal~A.%
\begin{figure}[H]
\centerline{
  \includegraphics[width=.32\textwidth,height=.215\textheight]{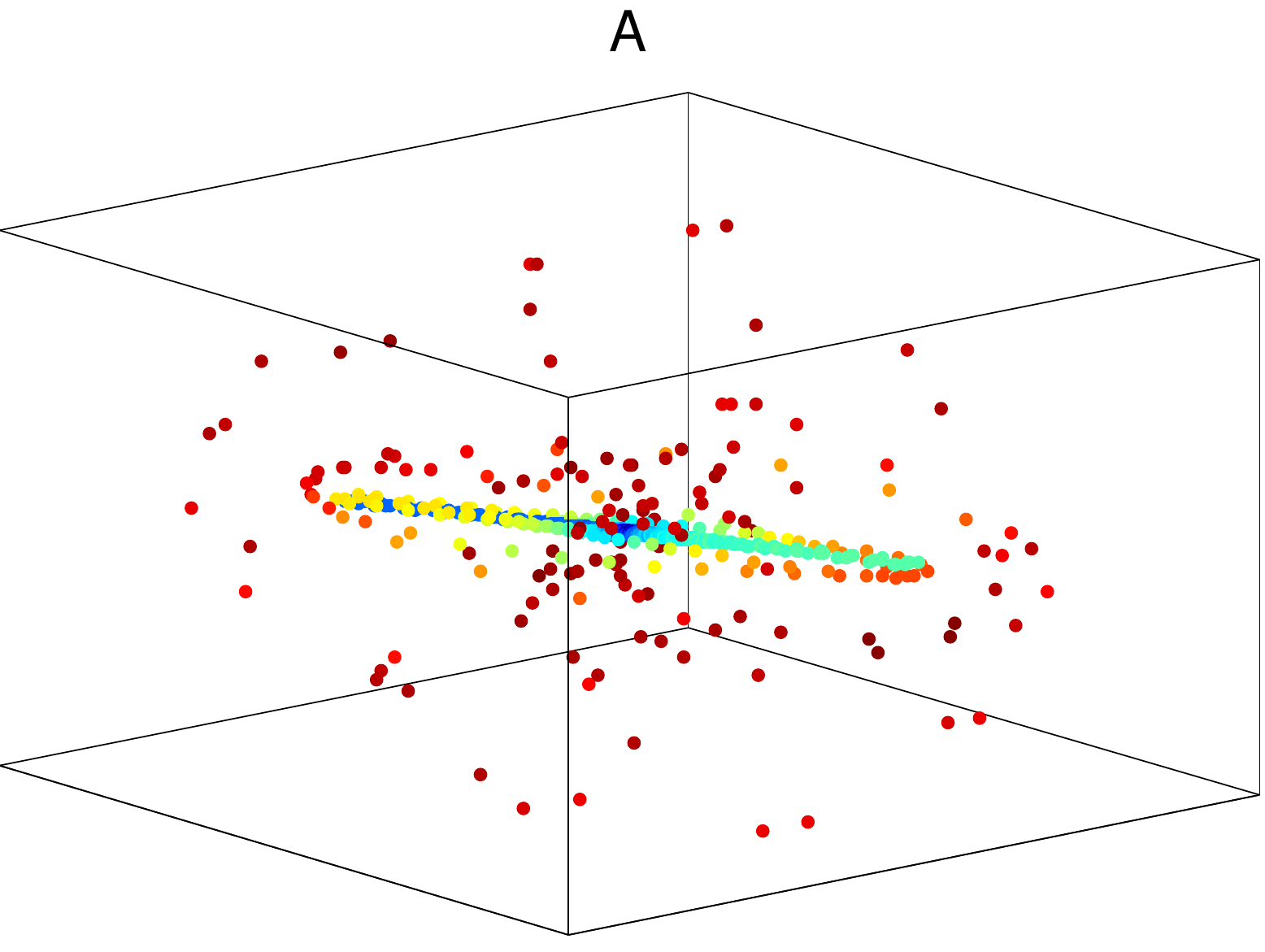}
  \includegraphics[width=.32\textwidth,height=.215\textheight]{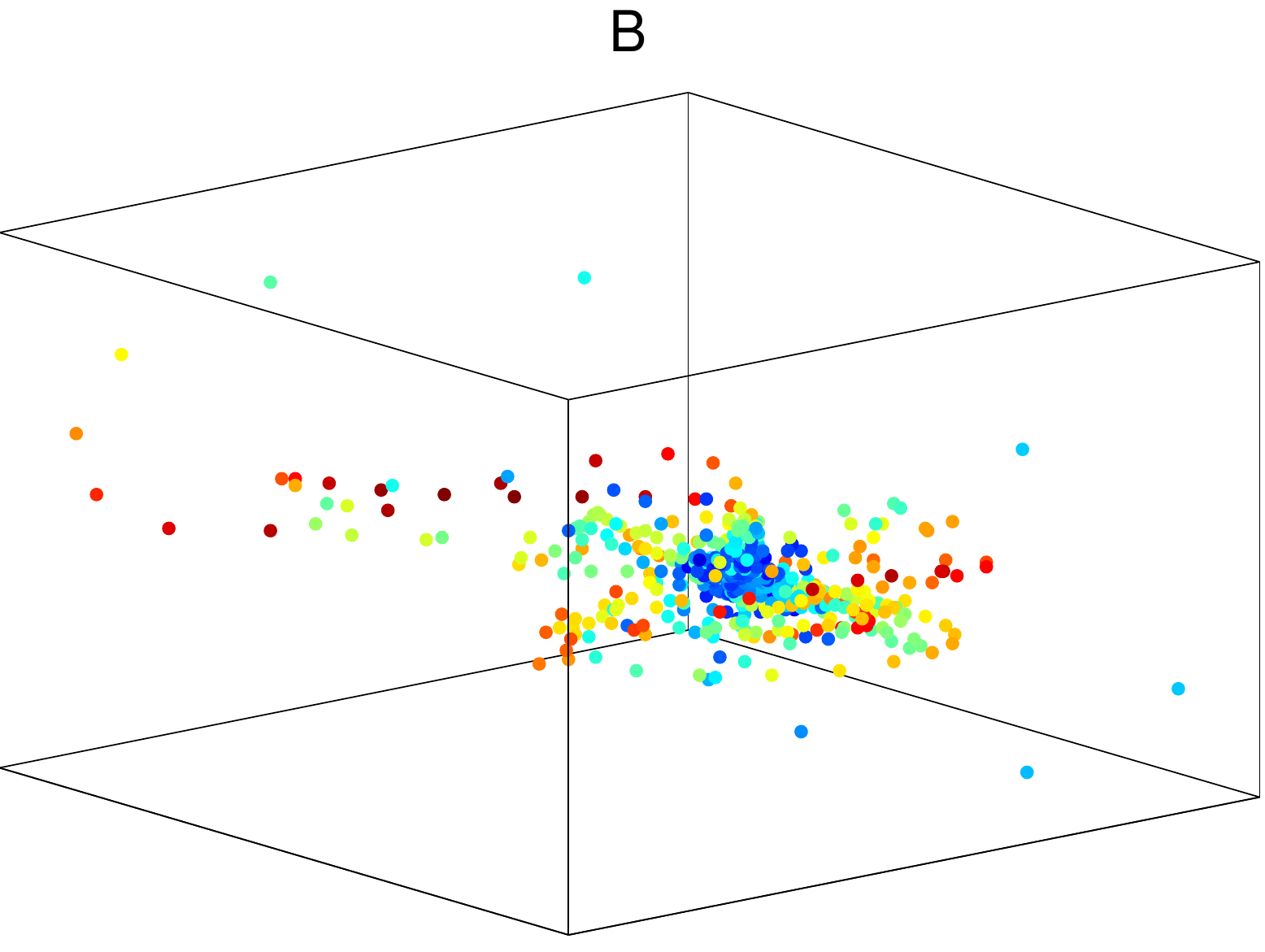}
  \includegraphics[width=.32\textwidth,height=.215\textheight]{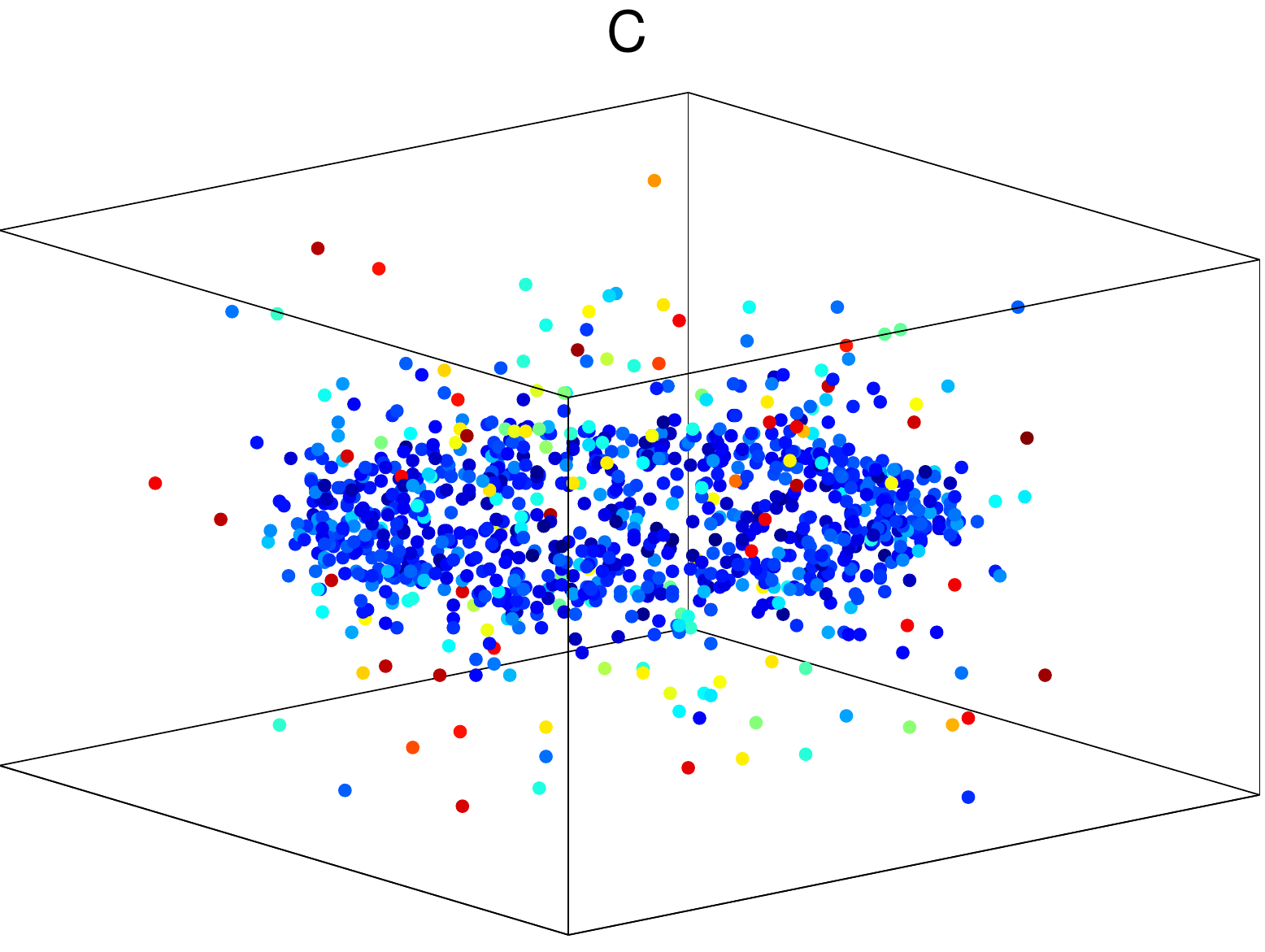}\\
}
\centerline{
    \includegraphics[width=.32\textwidth,height=.215\textheight]{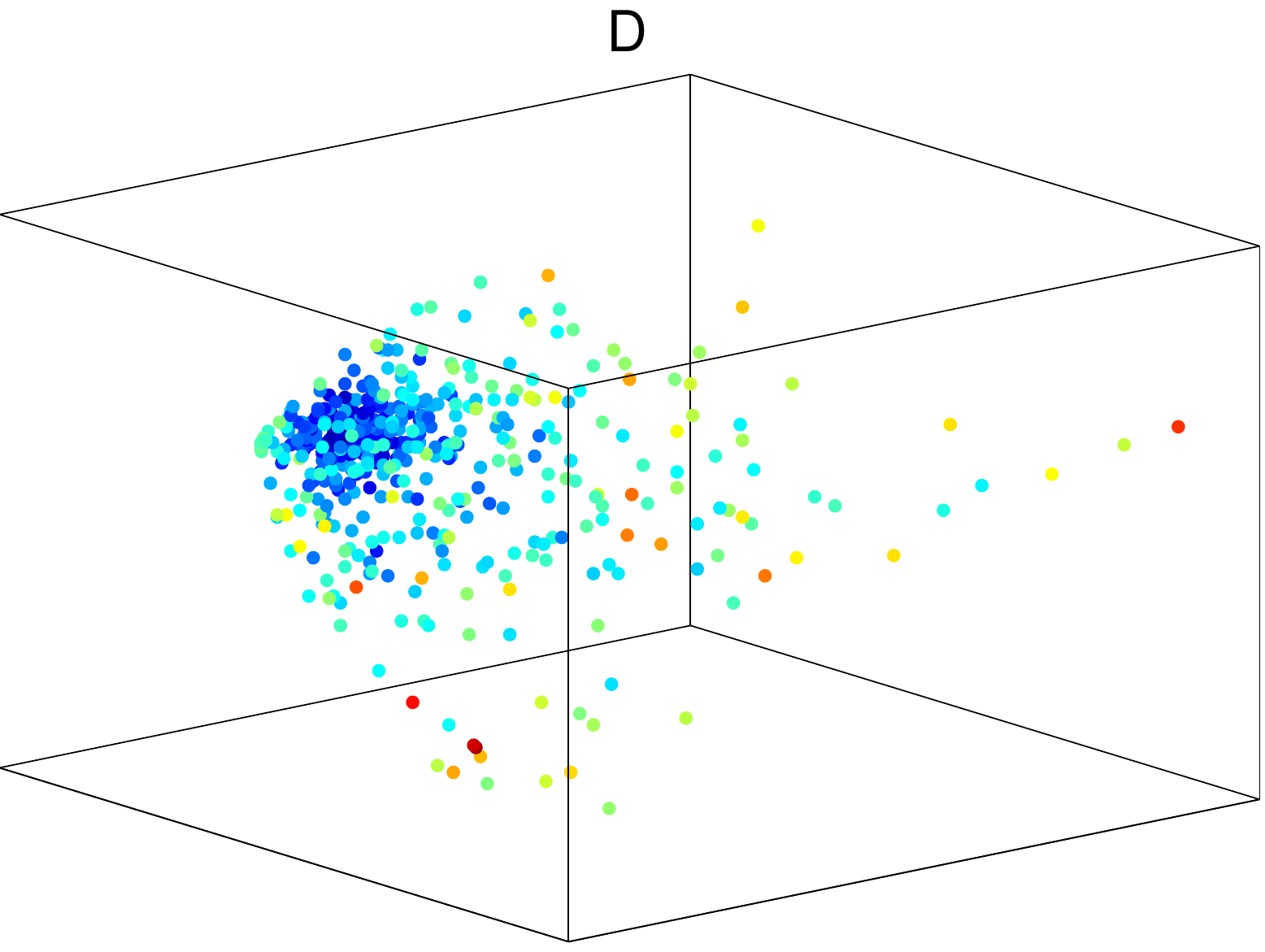}
    \includegraphics[width=.32\textwidth,height=.215\textheight]{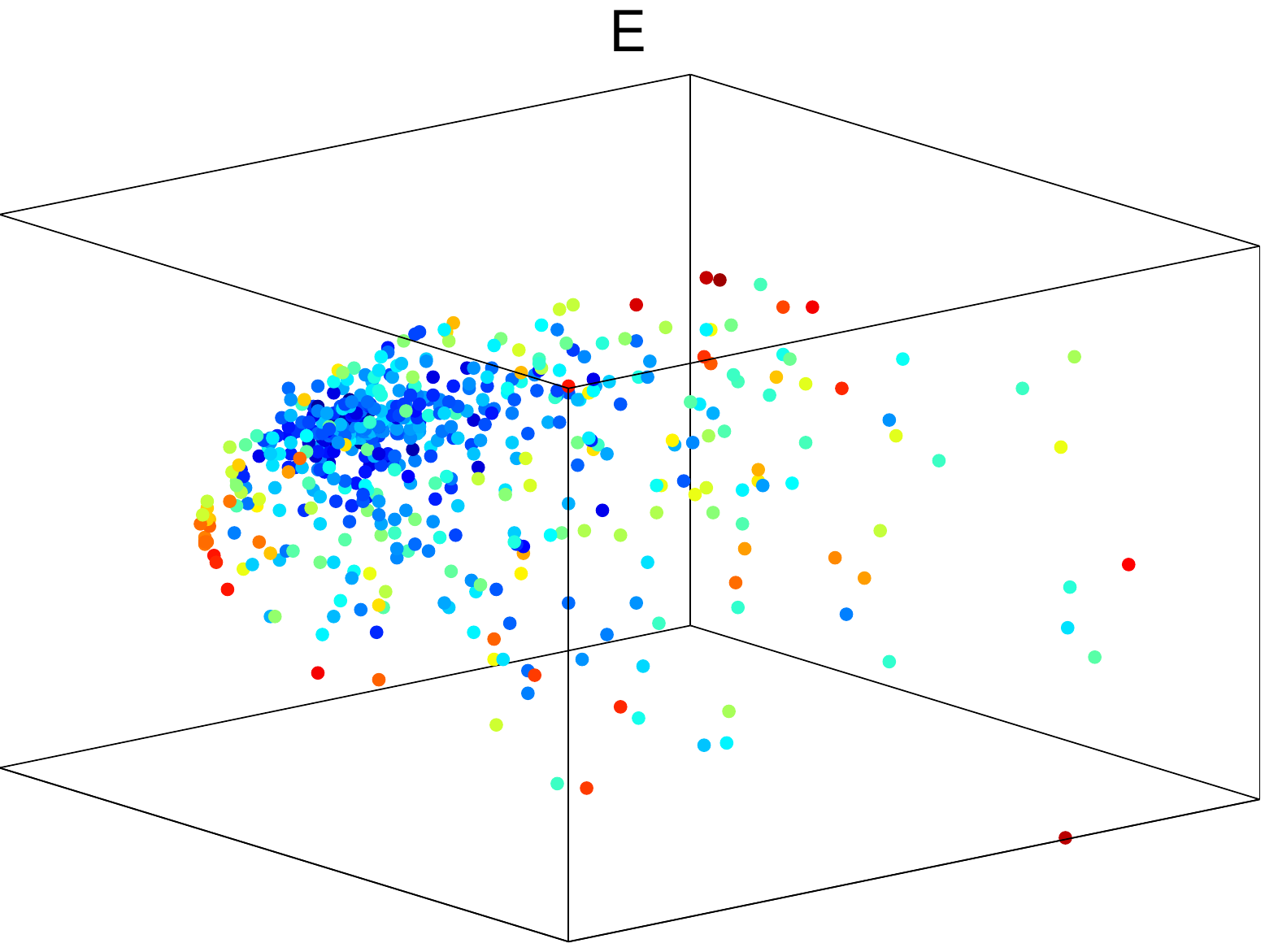}
    \includegraphics[width=.32\textwidth,height=.215\textheight]{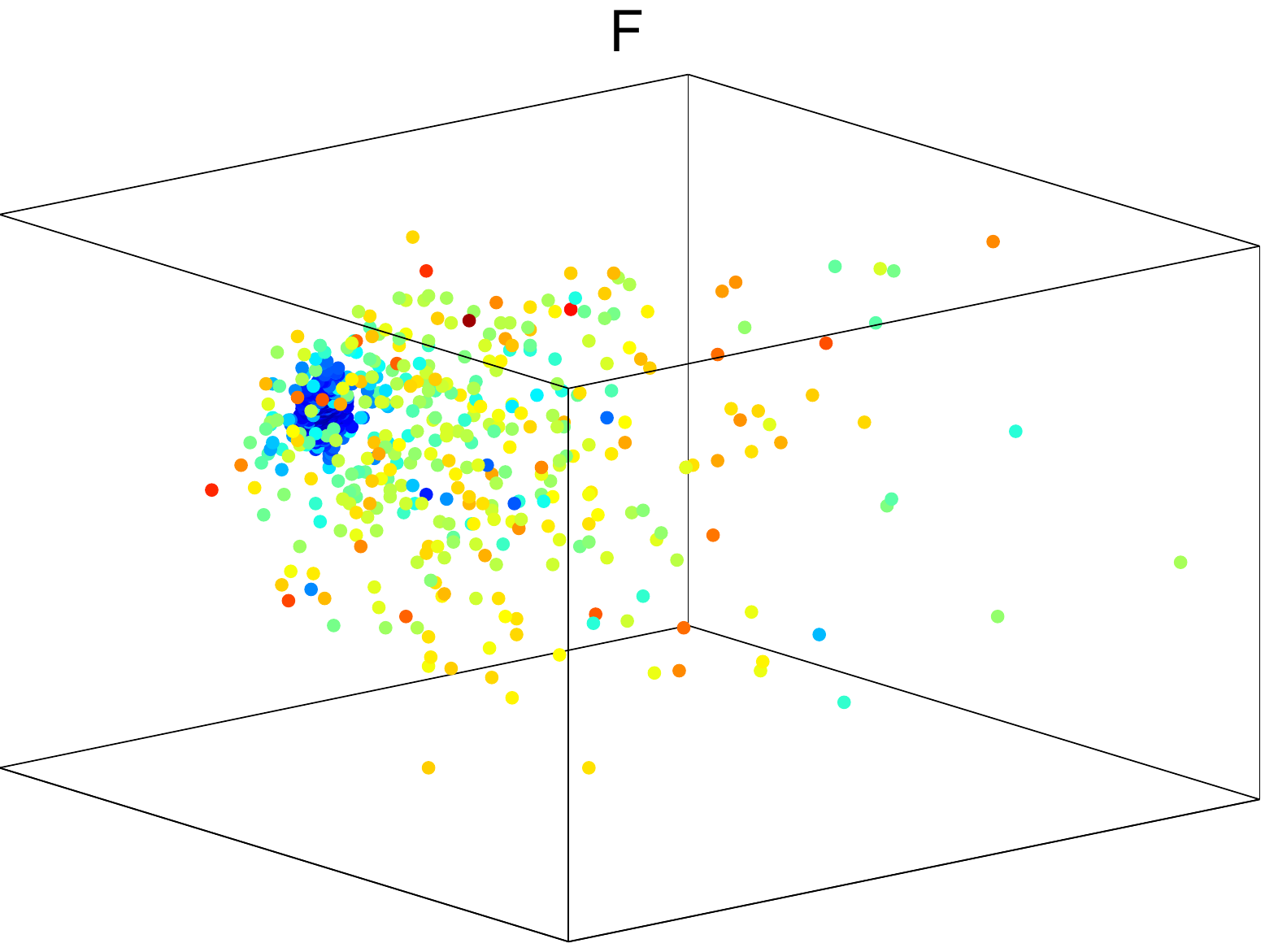}
}
  \caption{Principal component analysis of patch-sets associated with
    the time series A-C and the images D-F. Each point represents a
    patch; the  color encodes the variance within the patch (see Figures
    \ref{fig:s2ts} and \ref{fig:signalsec2im}.) \label{fig:pcasec2}}
\end{figure}
\noindent For this type of signal, it is appropriate to measure the
distance between the normalized patches, $\bx_n/\|\bx_n\|$ and
$\bx_m/\|\bx_m\|$ after computing the Fourier transforms (a simple
rotation of $\real^d$) of the respective patches. This process is akin
to the concept of time-frequency analysis. We expect that regions of
the signal with little local frequency changes will cluster in
$\real^d$: this is the case for the blue patches of the chirp A. On
the contrary, when the local frequency content changes rapidly (as in
the middle of the chirp A), the corresponding (red/orange) patches
will be at a large distance of one another in $\real^d$ (see Figure
\ref{fig:pcasec2}-A).

Finally, we can try to understand the organization of the patch-set for
the seismogram C. Let us assume that $\{x_n\}$ is obtained by sampling a
function of the form $x(t) = b(t) + w(t)$, where $w(t)$ represents a
seismic wave and $b(t)$ represents baseline activity. We can
expect that $w(t)$ is a rapidly oscillating transient with a rich
frequency content, while $b(t)$ is varying slowly. Now consider two
patches $\bx_n$ and $\bx_m$. It can be shown that if both patches
$\bx_n$ and $\bx_m$ are extracted from the baseline function, $b(t)$, and do
not contain any part of the energetic transient, then their mutual
distance is expected to be small. In addition, if $\bx_n$ contains
part of the energetic transient $w(t)$ and $\bx_m$ is extracted from the
baseline $b(t)$, then their mutual distance is expected to be large. Finally,
if $\bx_n$ and $\bx_m$ are composed of two different parts 
of $w(t)$,
then their mutual distance is also expected to be large (provided the
patches are sufficiently long and $w(t)$ oscillates sufficiently
fast). More generally, one can expect that two patches extracted from two
different energetic transients $w_1(t)$ and $w_2(t)$ will be at a
large distance from one another \cite{TaylorMeyer10}.%
\begin{figure}[H]
\centerline{
  \includegraphics[width=.3\textwidth,height=.235\textheight]{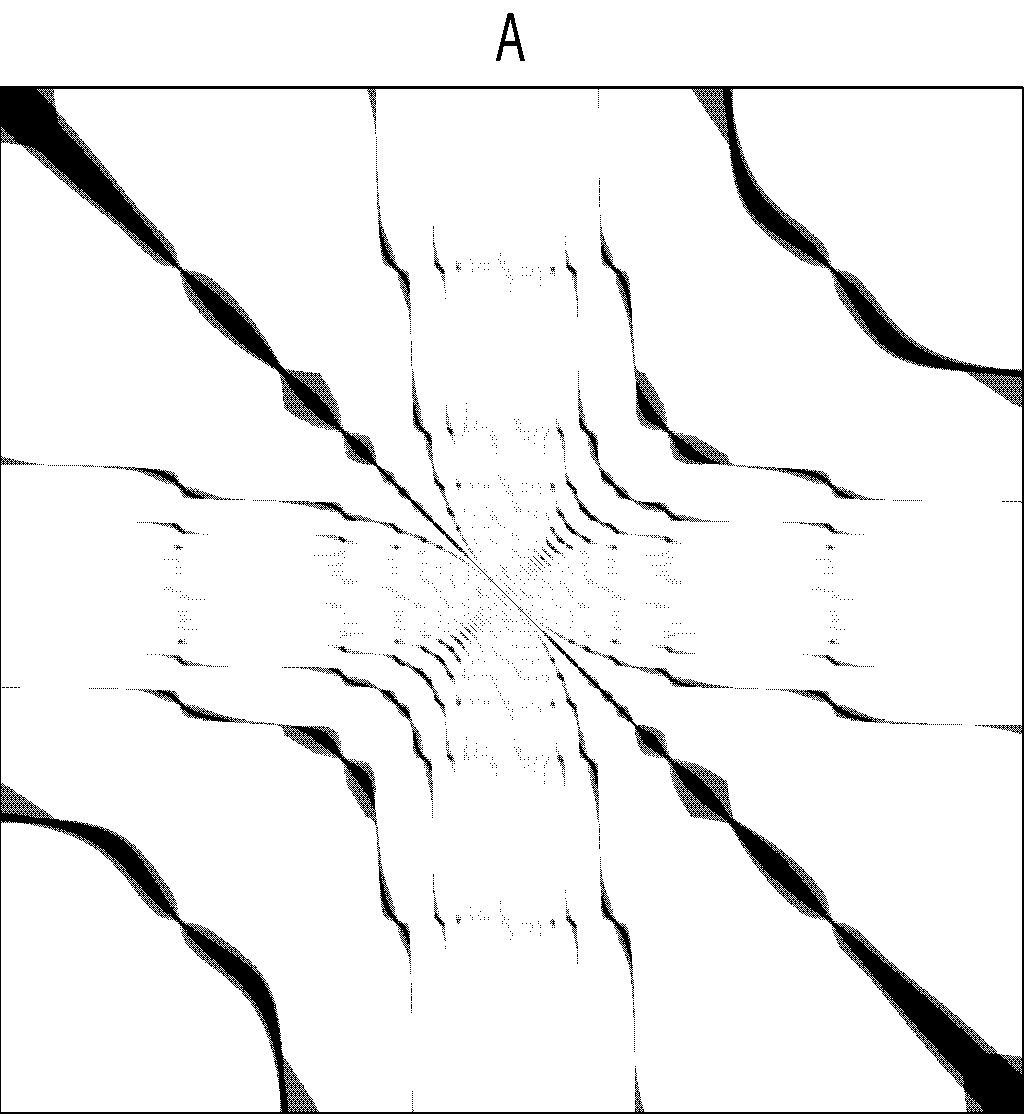}
  \includegraphics[width=.3\textwidth,height=.235\textheight]{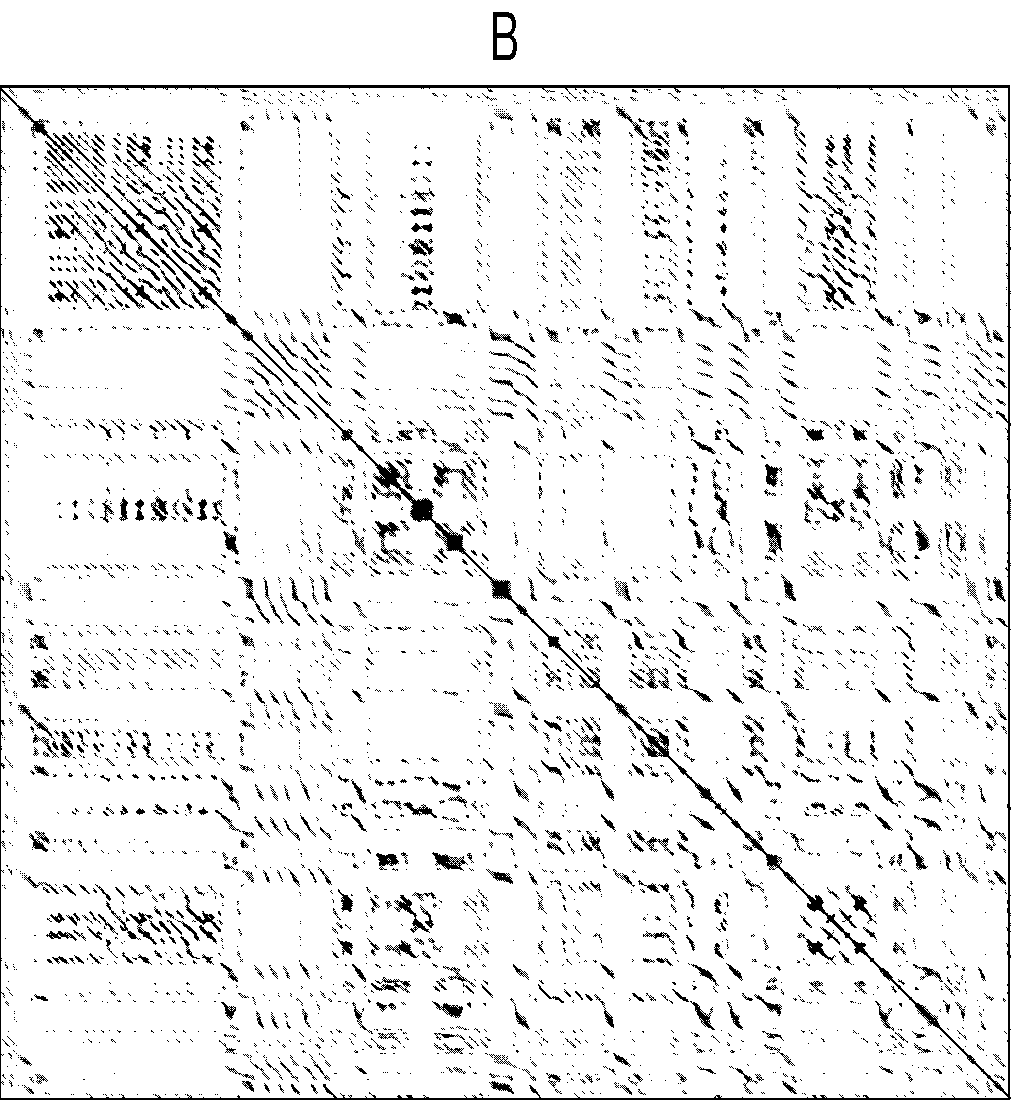}
  \includegraphics[width=.3\textwidth,height=.235\textheight]{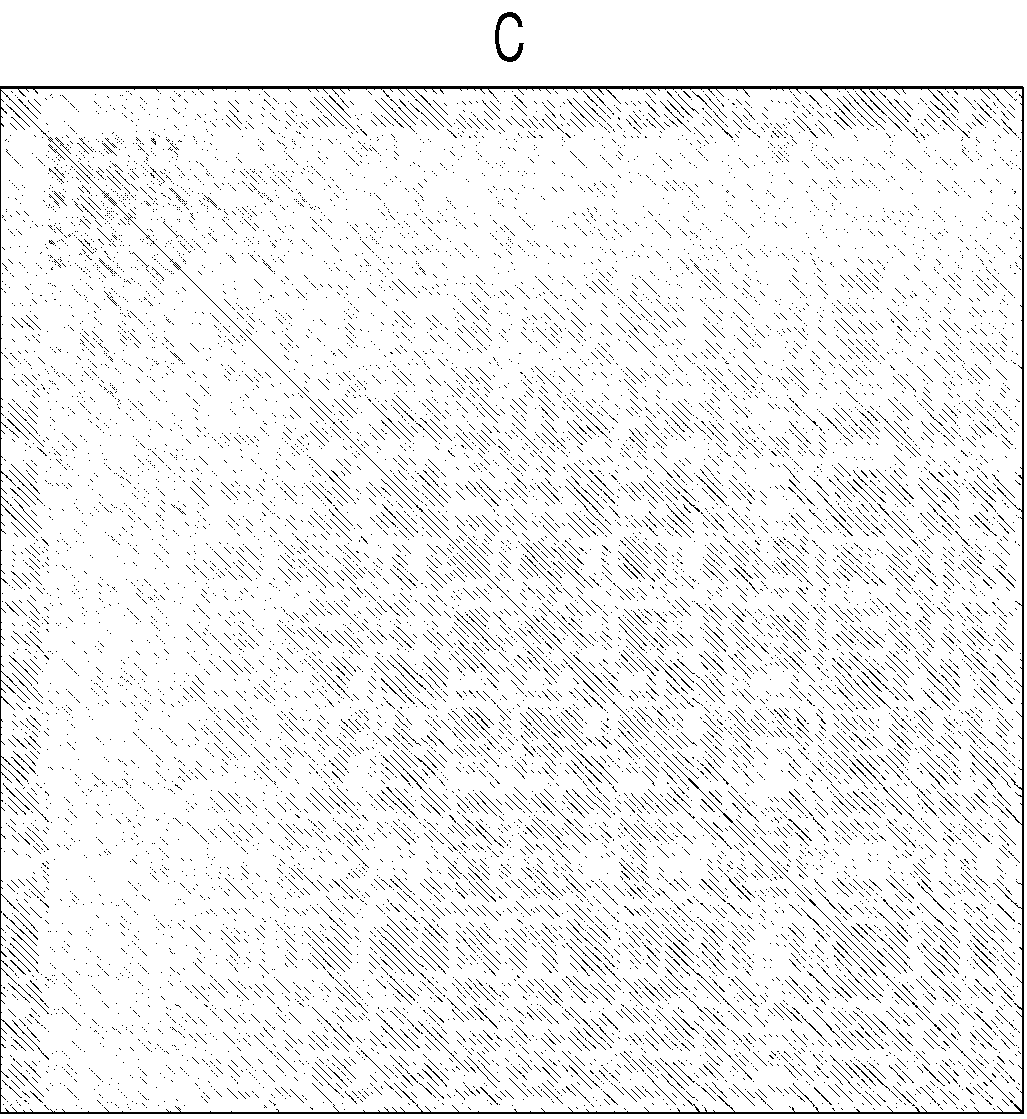}
}
\centerline{
    \includegraphics[width=.3\textwidth,height=.235\textheight]{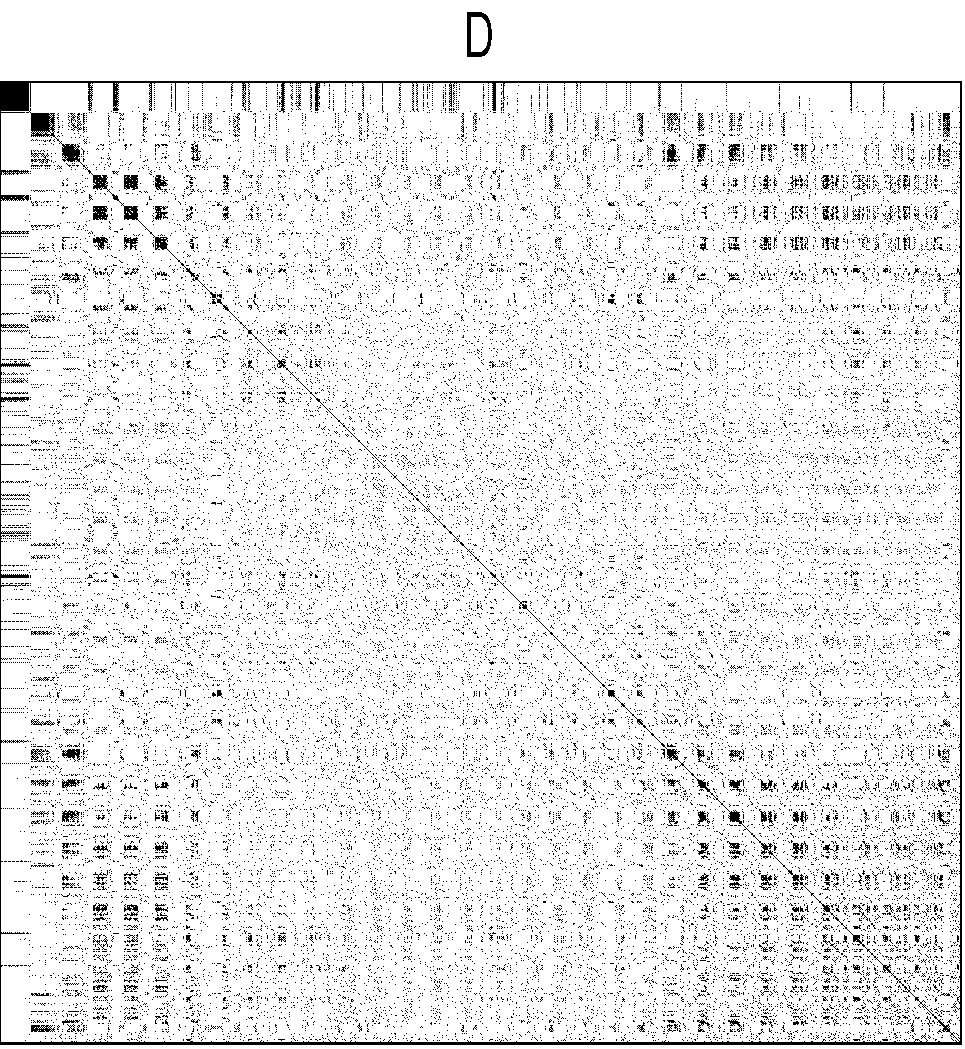}
    \includegraphics[width=.3\textwidth,height=.235\textheight]{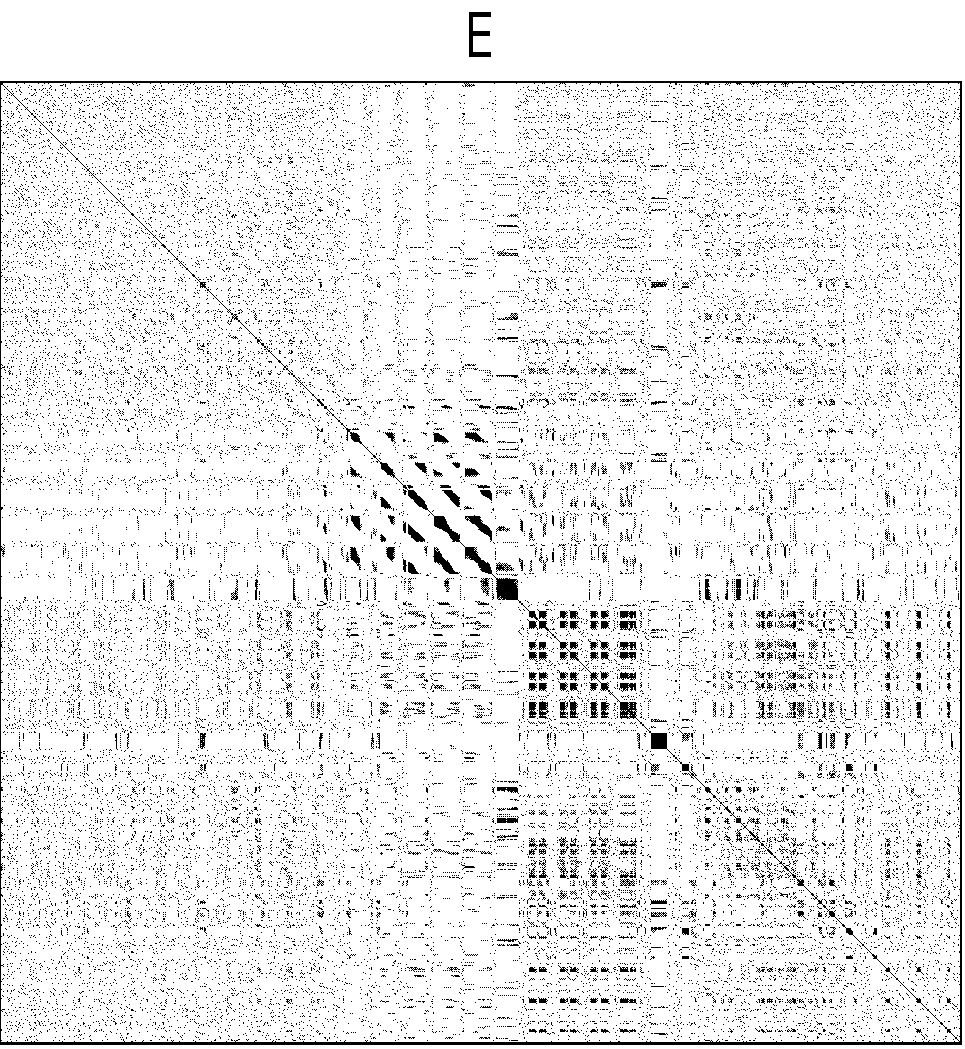}
    \includegraphics[width=.3\textwidth,height=.235\textheight]{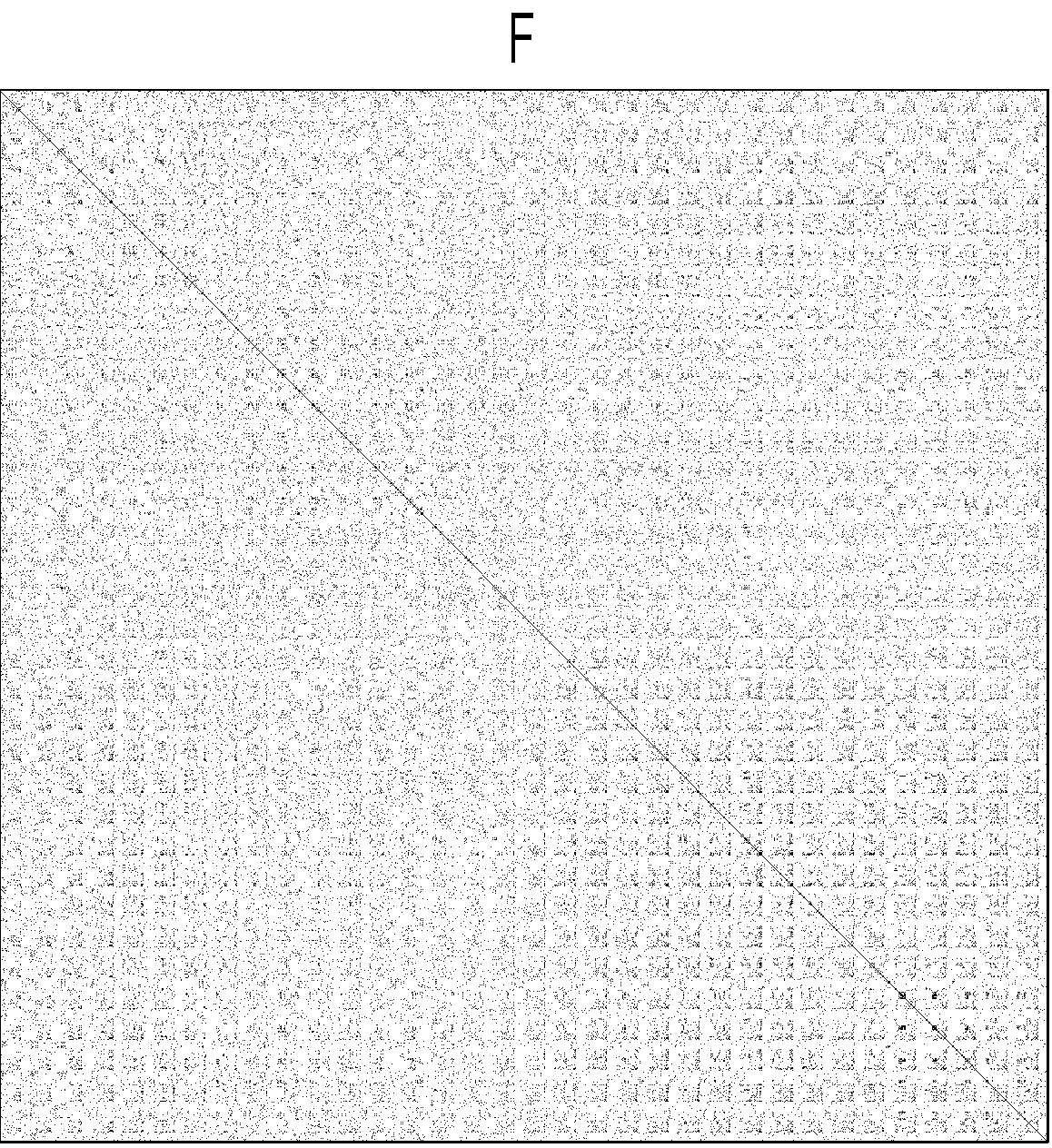}
}
  \caption{The weight matrices $\bo W$ associated with signals A-F are
    displayed as images: $w_{n,m}$ is encoded as a grayscale
    value: from white ($w_{n,m} =0$) to black $(w_{n,m} =1)$. Dark structures
    along the diagonal of the  $\bo W$ matrix associated with the
    time series A-C indicate that patches that are close in time are also
    close in $\real^d$.
    \label{fig:Wssec2}}
\end{figure}
\noindent
\subsection{From the patch-set to the patch-graph: the weight matrix $\bo W$}
\label{ssec:theweightmatrix}
Having gained some understanding about the organization of the
patch-set, we now move to the structure of the patch-graph and its
weight matrix $\bo W$.  Figure \ref{fig:Wssec2} displays the weight
matrices built from the patch-sets that correspond to the time series
A-C (top) and the images D-F (bottom). Note that when processing time
series A-C, the columns (or equivalently, the rows) of $\bo W$ can be
identified with temporally-ordered time-samples.  Therefore, a large
main diagonal in the weight matrix correspond to patches that are
close in time and also close in $\real^d$.  For instance, consider the
time series A and its associated weight matrix.  The dark bands near
the top-left and bottom-right of the diagonal correspond to the
slowly-varying oscillations near the beginning and end of the chirp
(see Figure \ref{fig:s2ts}).  Indeed, large entries in the diagonal of
$\bo W$ is a direct consequence of relatively little variation in the
time series.  On the other hand, the columns of $\bo W$ corresponding
to portions of the time series that exhibit rapid local changes
(center of Figure \ref{fig:Wssec2}-A) tend to lack such prominent
diagonal structures.  For such regions of the matrix $\bo W$, the
entries are no longer concentrated along the diagonal, and are
shattered across all rows and columns (see the center of $\bo W$ in
Figure \ref{fig:Wssec2}-A; the columns correspond to the fastest
oscillations at the center of the chirp).  The large distances between
these patches are also apparent in the lighter pixel intensities,
representing relatively smaller edge-weights. Note that the patches
extracted from the seismic data are very far apart, as indicated by
the much lighter shades of gray.  It is more difficult to relate the
ordering of an image's weight matrix to locations in the image
itself. For the weight matrices associated with images D-F, the
ordering of the columns is equivalent to the order in which the
patches were collected from the image plane: first left-to-right, then
top-to-bottom (similar to a raster scan, or how one would read pages
of a book). Hence, periodically repeating dark blocks in the weight
matrices associated with images D-F are indicative of image patches
that are close in $\real^d$ and close in the image-plane as a result
of relatively little change in the image's local content. For example,
the dark square-like structure that appears near the main diagonal of
$\bo W$ in Figure \ref{fig:Wssec2}-E, and which spans roughly one
fifth of the number of columns, corresponds to the mirror's smooth,
light border in image E.
\subsection{Summary of the experiments and our plan of attack}
The experiments in sections \ref{ssec:imsigexample} and
\ref{ssec:theweightmatrix} highlight the fact that regions of an
image, or of a signal, that contain anomalies (e.g. singularities,
edges, rapid changes in the frequency content, etc.) are scattered all
over the patch-set, making their detection and identification
extremely difficult (see Figures \ref{fig:pcasec2}-A and
\ref{fig:pcasec2}-F). In contrast, patches from smooth regions appear
to cluster along low dimensional curves or surfaces. Because the
anomalous patches are usually the most interesting ones, we need to
find a new parametrization of the patch-set that concentrates the
anomalies and separate them from the smooth baseline part of the
image. The structure of $\bo W$ in the ``rough regions'' suggest that
patches that contain anomalies appear to be very well connected (see
the center of Figure \ref{fig:Wssec2}-D, which corresponds to the boa
on the hat of Lenna). This concept can be quantified by studying how
fast a random walk would reach all patches in these rough regions of
$\bo W$, and suggest that we should consider studying the hitting
times associated with a random walk on the patch-graph. In the next
section we formalize this concept and propose a parametrization of the
patch-set in terms of commute time. A theoretical analysis of this
approach is provided in section \ref{ssec:theory}
\section{Parametrizing the patch-graph}
\label{sec:patchgraphparam}
\subsection{The  fast and slow patches}
We first introduce the concept of {\em fast} and {\em slow} patches.
We have noticed that patches that contain anomalies (discontinuities,
edges, fast changes in frequency, etc.) in the original signal lead to
regions of the matrix $\bo W$ where the nonzero entries are scattered
all around. We call such patches {\em fast patches} because, as we
will see in the following, a random walk will diffuse extremely fast
in such regions of the patch-graph. Conversely smooth regions of the
signal lead to {\em slow patches} that are associated with a small
number of large entries in $\bo W$, which are concentrated
along the diagonal. A random walk initialized in the slow patch region
of the patch-graph will diffuse very slowly.
\subsection{A better metric on the graph: the commute time}
\label{sssec:commutetimes}
As explained previously, we propose to replace the Euclidean distance,
which leads to the scattering of the fast patches seen in Figure
\ref{fig:pcasec2} by a notion of distance that quantifies the speed at
which a random walk diffuses on the patch-graph. We propose to use the
commute time. Parametrizing the graph using its commute time distance
is closely related to parametrizing the graph using its diffusion
distance \cite{Coifman06a,Jones08} (see Section \ref{sss:sdcommute}).
Although the works \cite{Buades05,SingerBoaz,Szlam08} do not
explicitly embed vertices of the patch-graph based on the diffusion
distance, they also study a random walk on the patch-graph, and define
the diffusion distance in terms of this walk.  In these studies, noise
is removed by evolving the diffusion process for a small time. A
detailed comparison of our approach with the seminal work of
\cite{SingerBoaz} is provided in section \ref{ssec:relatedwork}.  We
note that the notion of first-passage time associated with a diffusion
(which is equivalent to the hitting time associated with a random
walk) has been used extensively to characterize the geometry of
complex networks, and random media (e.g. \cite{Benichou10,Condamin07}
and references therein). It is therefore natural to analyze the
patch-set with this distance.
\subsubsection{A random walk on the patch-graph}
In order to define the commute time between two vertices, we first
need to define a random walk on the graph. In our problem, the random
walk does not correspond to a physical process, but will lead to a
notion of global proximity between patches. We consider a first-order
homogeneous Markov process, $Z_k$, defined on the vertices of the
patch-graph, $\Gamma$, and evolving with the transition probability
matrix $\bo{P}$ given by
\begin{equation}
  \bo{P}_{n,m} = \prob(Z_{k+1} = \bo x_m|Z_{k} = \bo x_n)  \triangleq
  \frac{w_{n,m}}{\sum_{l}w_{n,l}} = \frac{\bo{W}_{n,m}}{\bo{D}_{n,n}}.
\end{equation}
Consider a slow patch $\bx_n$ extracted from a regular/smooth part of
the signal. If the random walk starts at $\bx_n$, then it can only
travel along the low-dimensional structure that corresponds to the
temporal neighbors of $\bx_n$ (see e.g. Figure \ref{fig:pcasec2}-A.)
The existence of this narrow bottleneck is also visible in the
$\bo{W}$ matrix (see Figure \ref{fig:Wssec2}-A): a random walk
initialized within the fat diagonal of the upper left corner of
$\bo{W}$ (the low frequency part of the chirp) is trapped in this
region of the matrix, and can only travel along this fat diagonal. As
a result, it will take many steps for the random walk to reach another
slow patch $\bx_m$ if $| n-m |$ is large. This notion can be
quantified by computing the average \textit{hitting-time},
$h(\bx_n,\bx_m)$, which measures the expected minimum number of steps
that it takes for the random walk, started at vertex $\bx_n$, to reach
the vertex $\bx_m$ \cite{Bremaud99}
\begin{equation*}
  h(\bx_n,\bx_m) = \E_{n} \min\{j\geq 0 : Z_j = \bx_m\},
\end{equation*}
where the expectation $\E_{n}$ is computed when the random walk is
initialized at vertex $\bx_n$, i.e. when $Z_0 =\bx_n$.  The
commute time \cite{Bremaud99}: provides a symmetric version of $h$,
and is defined by
\begin{equation}
  \kappa(\bx_n, \bx_m) = h(\bx_n,\bx_m)  + h(\bx_m,\bx_n).
  \label{eqn:commutetimesdefined}
\end{equation}
\subsubsection{Spectral representation of the commute time}
\label{sss:sdcommute}
When the random walk is reversible and the graph is fully connected,
the commute time can be expressed using the eigenvectors
$\phi_1,\ldots,\phi_N$ of the symmetric matrix
\begin{equation*}
  \bo{D}^{-1/2}\bo{W}\bo{D}^{-1/2} = \bo{D}^{1/2}\bo{P}\bo{D}^{-1/2}.
\end{equation*}
The corresponding eigenvalues can be labeled such that
$-1<\lambda_N\leq\ldots\leq\lambda_2<\lambda_1=1$.  Each eigenvector
$\phi_k$ is a vector with $N$ components, one for each vertex of the
graph. Hence, we write
\begin{equation*}
  \phi_k = 
  \begin{bmatrix}
    \phi_k(\bx_1) &\phi_k(\bx_2) &\ldots &\phi_k(\bx_N)
  \end{bmatrix}
  ^T,
\end{equation*}
to emphasize the fact that we consider $\phi_k$ to be a function
sampled on the vertices of $\Gamma$. The commute time can be expressed as
\begin{equation}
  \kappa(\bx_n, \bx_m) = \sum_{k =
    2}^N\frac{1}{1-\lambda_k}\left(\frac{\phi_k(\bx_n)}{\sqrt{\pi_n}}
    - \frac{\phi_k(\bx_m)}{\sqrt{\pi_m}}\right)^2, 
  \label{eqn:commutetime}
\end{equation}
where $\pi_n = \sum_{m = 1}^Nw_{n,m}/\sum_{j,l = 1}^N w_{j,l}$ is the
stationary distribution associated with the transition probability
matrix $\bo{P}$ \cite{Lovasz93,Shen2008886}.  
\subsubsection{The relationship to diffusion maps}
\label{sssec:diffusionconnect} 
The {\em diffusion distance} \cite{Coifman06a} between vertices $\bo
x_m$ and $\bo x_n$, $D_t(\bo x_m,\bo x_n)$, measures the distance 
between the transition probability distributions -- computed at time
$t$ -- of two random walks initialized at $\bx_n$ and $\bx_m$,
$\sum_{l=1}^N |\bo{P}^{(t)}_{n,l} -\bo{P}^{(t)}_{m,l}|^2$. The diffusion distance can
also be decomposed in terms of the eigenvectors $\phi_k$
\cite{Coifman06a},
\begin{equation}
  D^2_t(\bo x_m,\bo x_n)=
  \frac{1}{V}\sum_{k =
    2}^N\lambda_k^{2t}\left(\frac{\phi_k(\bx_m)}{\sqrt{\pi_m}} -
    \frac{\phi_k(\bx_n)}{\sqrt{\pi_n}}\right)^2,
  \label{eqn:diffusiondist}
\end{equation}
where $V=\sum_{m',n'} w_{m',n'}$ is the volume of the graph.  It follows
that the commute time is a scaled sum of the squares of diffusion
distances computed at all times,
\begin{equation}
  \kappa(\bx_m, \bx_n) = V \sum_{t = 0}^\infty D^2_{t/2}(\bo x_m,\bo x_n).
  \label{alltimes}
\end{equation}
The significance of this equation is that the commute time
includes the short term evolution ($t \approx 0$) as well as the
asymptotic regime ($t \rightarrow \infty$) of
the random walk. We will come back to this analysis in section \ref{ssec:pgm}.
\subsection{Parametrizing the patch-graph}
\label{ssec:paramthegraph}
Equation (\ref{eqn:commutetime}) suggests the following
embedding $\Psi$ of the patch-graph $\Gamma$ into $\real^{N-1}$,
\begin{equation}
  \Psi : \bx_n  \longrightarrow\frac{1}{\sqrt{\pi_n}}
  \begin{bmatrix}
    \frac{\phi_2(\bx_n)}{\sqrt{1-\lambda_2}} &
    \frac{\phi_3(\bx_n)}{\sqrt{1-\lambda_3}} & \ldots &
    \frac{\phi_N(\bx_n)}{\sqrt{1-\lambda_N}}
  \end{bmatrix}
  ^T,\qquad n = 1,2,\ldots,N.
  \label{eqn:maptoN1}
\end{equation}
If we agree to measure the distance on the graph $\Gamma$ using the
square root of the commute time, then the mutual Euclidean distance
after embedding is equal to the original distance on the graph,
\begin{equation}
  \|\Psi (\bx_n) - \Psi(\bx_m)\| = \sqrt{\kappa (\bx_n,\bx_m)}.
  \label{isometry}
\end{equation}
The result is a direct consequence of (\ref{eqn:diffusiondist}) and
(\ref{alltimes}).  Similar ideas were first proposed in
\cite{Berard94} to embed manifolds and are the foundation of the
parametrizations given in \cite{Belkin03,Coifman06a}.  In practice, we
need not use all the $N-1$ coordinates in the embedding defined by
(\ref{eqn:maptoN1}). Indeed, since
$\lambda_N\leq\cdots\leq\lambda_2<\lambda_1$, we have that
$\tfrac{1}{\sqrt{1-\lambda_N}} \leq\cdots \leq\tfrac{1}{\sqrt{1-\lambda_3}}
\leq\tfrac{1}{\sqrt{1-\lambda_2}}$, and therefore, if we can accept some approximation error, we can use only the first $d'$ coordinates of $\Psi$.  As we will see in section
\ref{ssec:pgm}, this dimension reduction further improves the
separation between slow patches and fast
patches.  In the remaining of the paper we will work with the
embedding of $\Gamma$ into $\real^{d'}$ defined by
\begin{equation}
  \Phi:  \bx_n \longrightarrow
  \frac{1}{\sqrt{\pi_n}}
  \begin{bmatrix}
    \frac{\phi_2(\bx_n)}{\sqrt{1-\lambda_2}} &
    \ldots &
    \frac{\phi_{d'+1}(\bx_n)}{\sqrt{1-\lambda_{d'+1}}}
  \end{bmatrix}
  ^T.
  \label{eqn:lowdparam}
\end{equation}
We note that we can always choose $d'$ such that the embedding $\Phi$
almost preserves the commute time,
\begin{equation}
  \|\Phi (\bx_n) - \Phi(\bx_m)\|^2 \approx \kappa (\bx_n,\bx_m).
  \label{pseudo-isometry}
\end{equation}
In fact, our experiments indicate that this approximation holds for small
values of $d'$.%
\clearpage
\begin{figure}[H]
  \centerline{
    \includegraphics[width=.3\textwidth,height=.2\textheight]{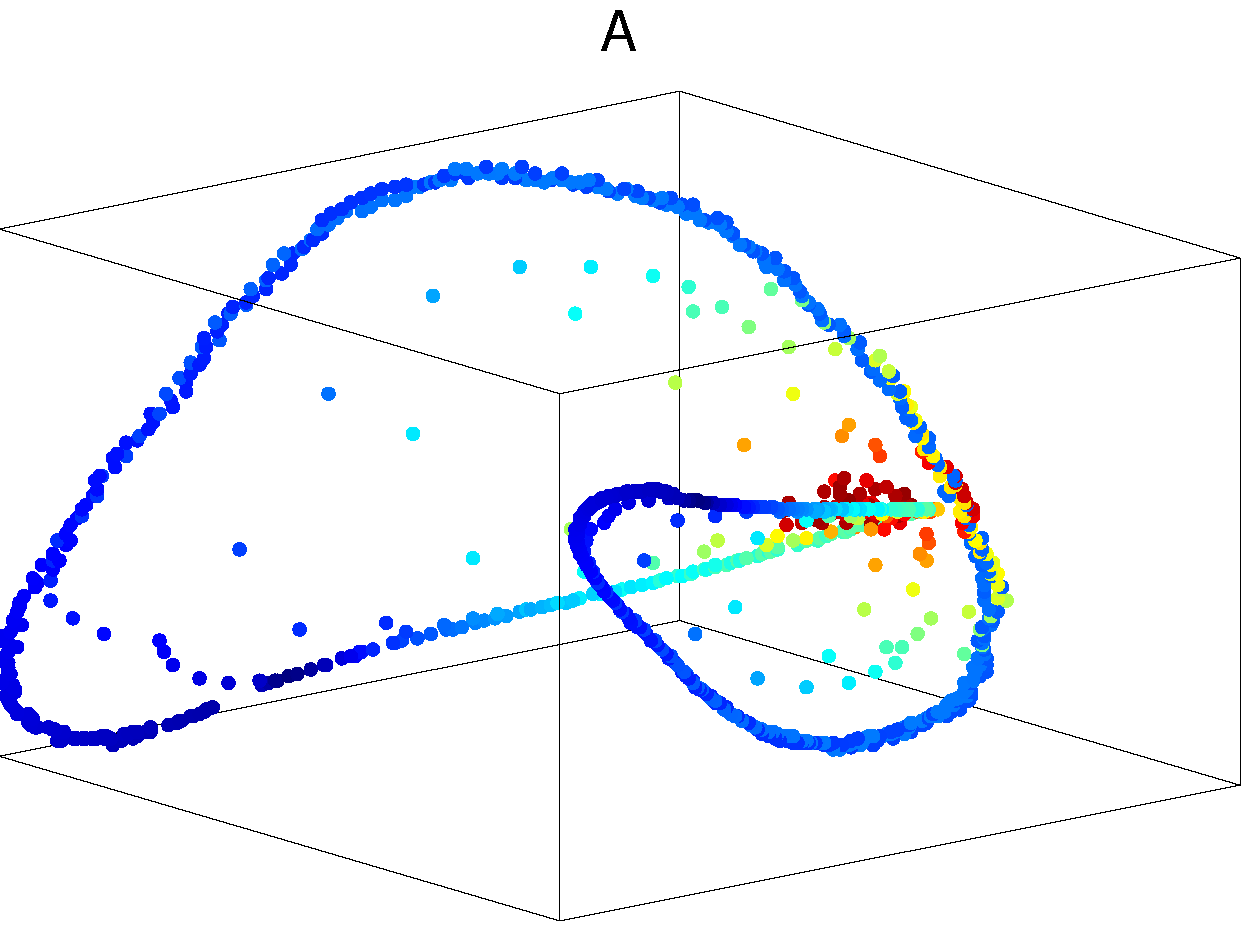} 
    \includegraphics[width=.3\textwidth,height=.2\textheight]{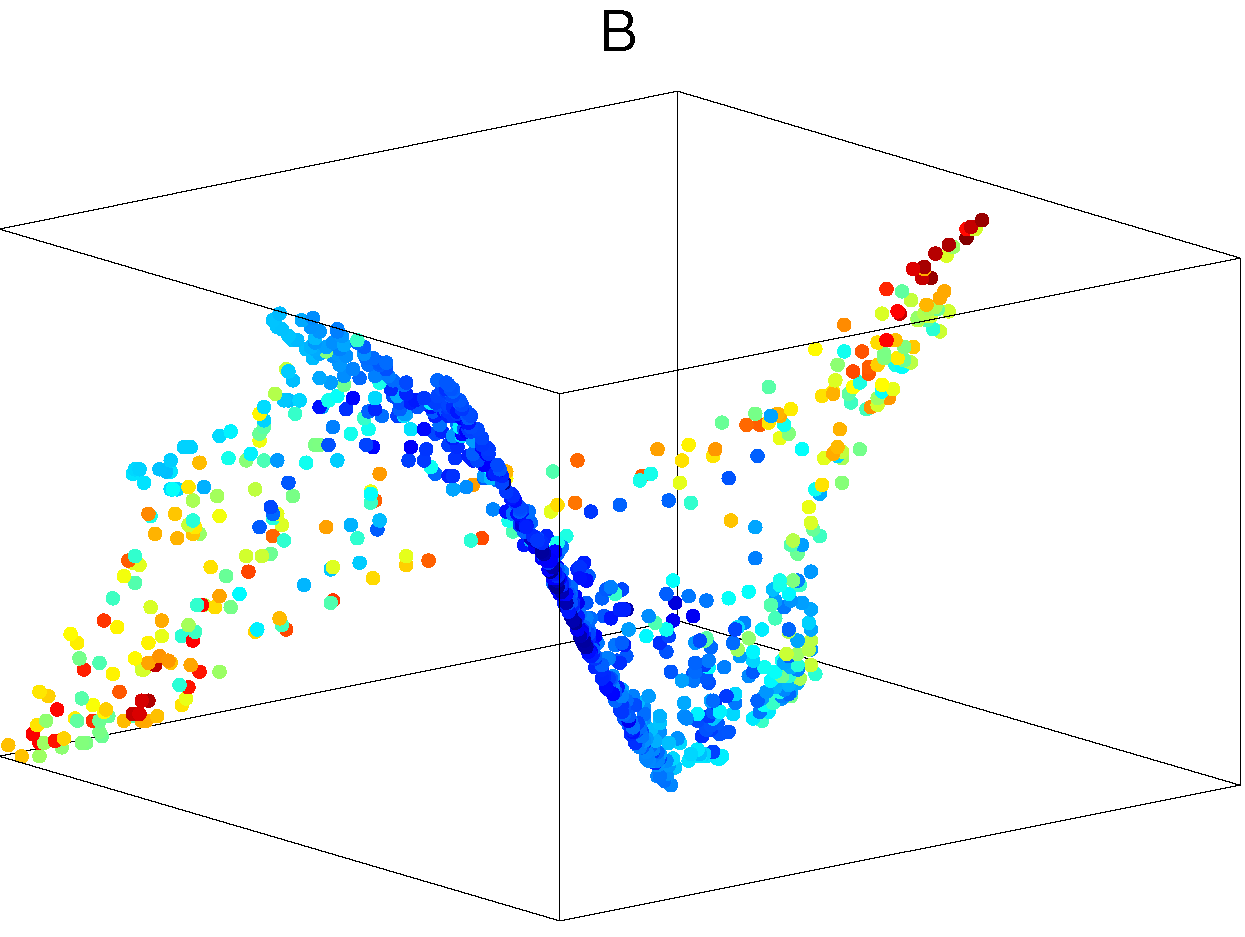} 
    \includegraphics[width=.3\textwidth,height=.2\textheight]{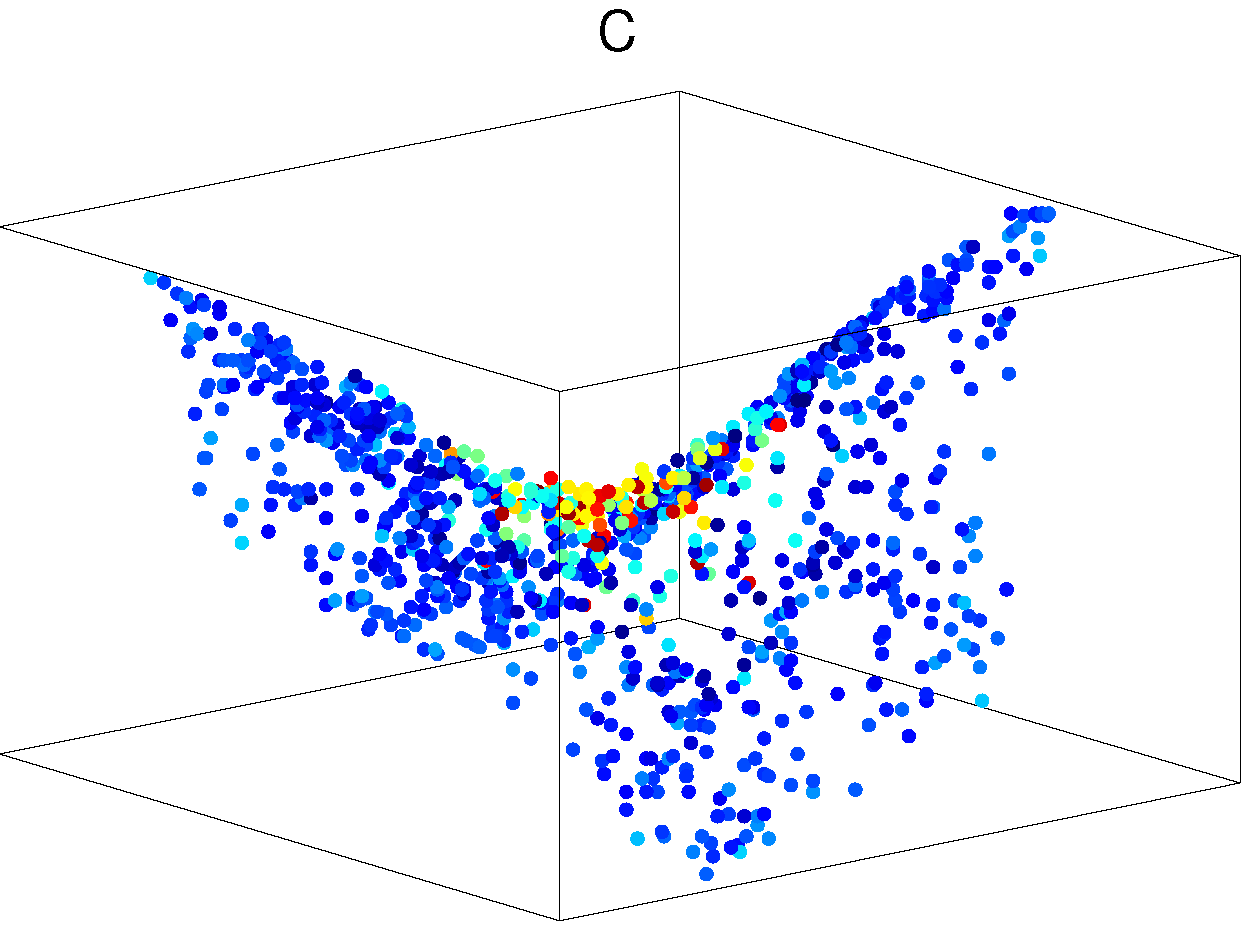} 
}
\centerline{
    \includegraphics[width=.3\textwidth,height=.2\textheight]{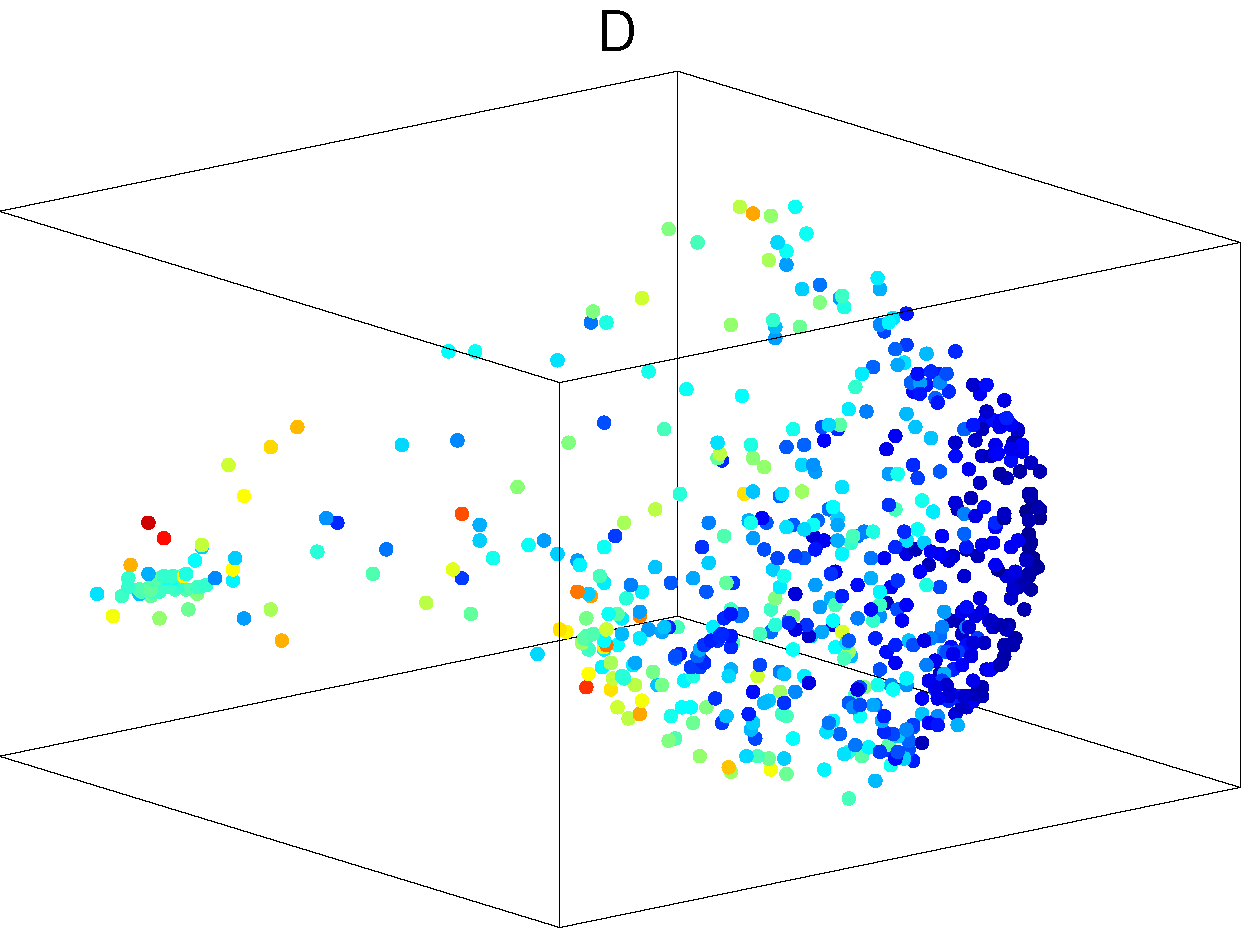}
    \includegraphics[width=.3\textwidth,height=.2\textheight]{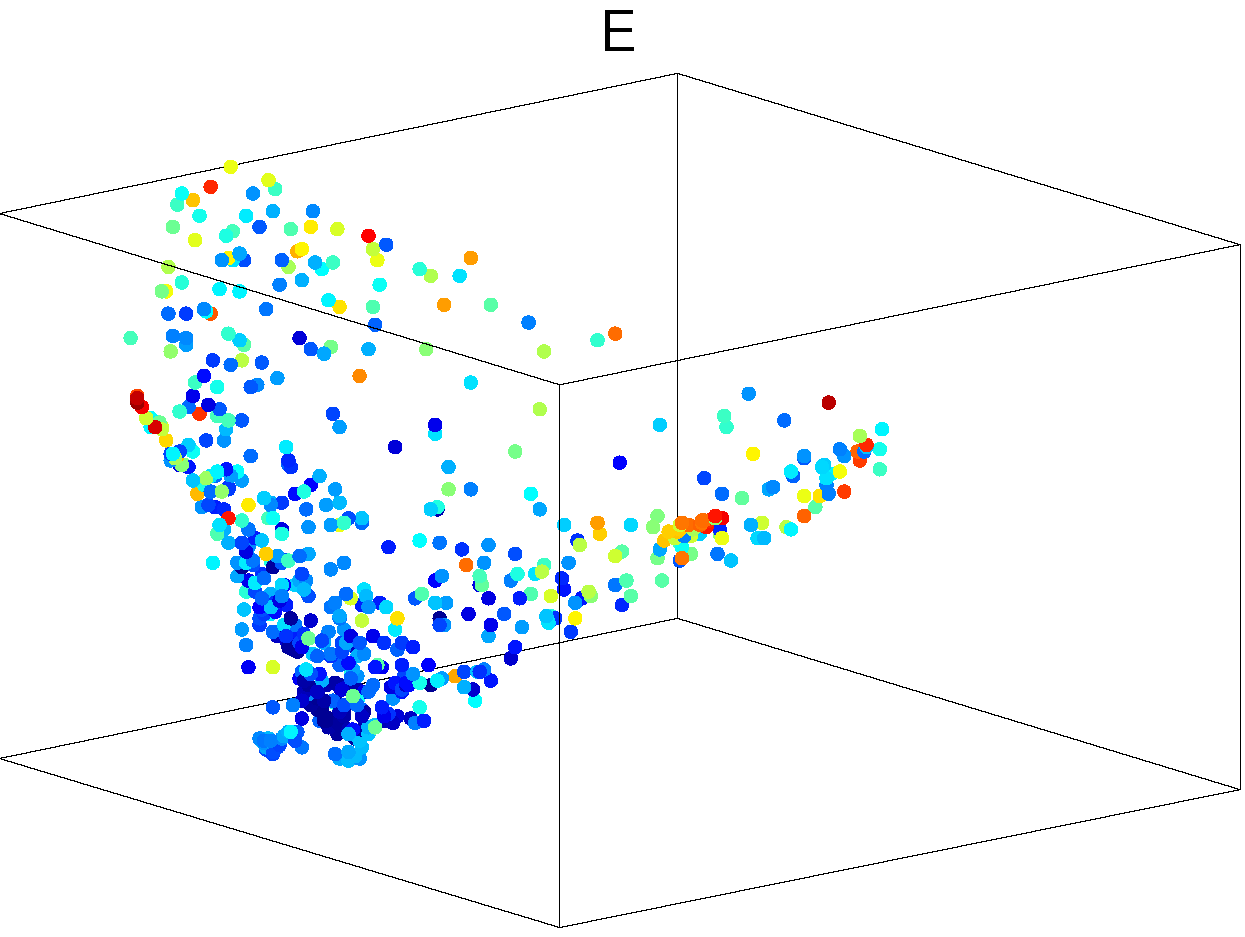}
    \includegraphics[width=.3\textwidth,height=.2\textheight]{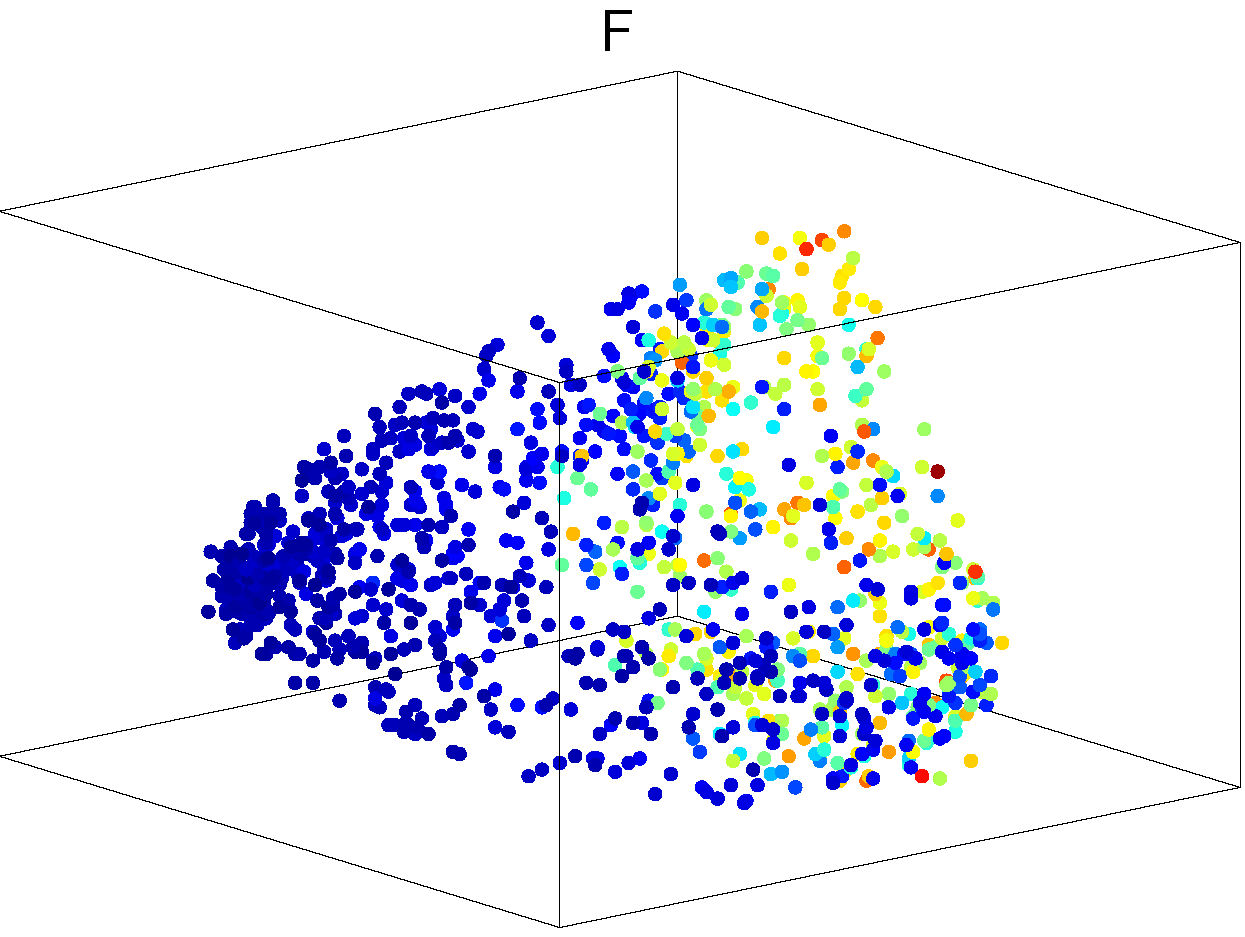}
}
    \caption{Scatter plot of the patch-set shown in Figure
      \ref{fig:pcasec2} after parametrizing using $\Phi$ in 
      (\ref{eqn:lowdparam}), with $d'=3$.  The fast patches (red and orange)
      are now concentrated and have been lumped together. The slow patches
      (blue-green) remain aligned along curves (for time-series) and surfaces (for 
      images). \label{fig:graphparams} } 
\end{figure}
\subsection{Examples (revisited)}
\label{ssec:exampleparams}
Figure \ref{fig:graphparams} displays the embedding of the patch-sets
associated with signals and images A-F using the map $\Phi$
(\ref{eqn:lowdparam}), where $d' = 3$.  The blue curve in Figure
\ref{fig:graphparams}-A corresponds to the slow patches (low
frequencies of the chirp) that are connected according to their
temporal proximity. On the other hand, red and orange patches
extracted from the high frequency part of the chirp are now
concentrated in a relatively small region (compare to Figure
\ref{fig:pcasec2}-A).  Similar features are seen in the parametrizations
of the patch-graphs associated with signals B-F.

\section{A model for  the patch-graph and the analysis of its
  embedding\label{ssec:theory}}
\subsection{Our approach}
The embedding of the patch-graph $\Gamma$ defined by $\Phi$, in
(\ref{eqn:lowdparam}), should lead to a representation of the
patch-set in $\real^{d'}$ where distances correspond to commute times
measured along the graph before embedding. Our goal is to explain the
concentration of the fast patches created by the embedding $\Phi$ (see
e.g. Figure \ref{fig:graphparams}). Our approach is based on a
theoretical analysis of a graph model that epitomizes the
characteristic features observed in patch-graphs composed of a mixture
of fast and slow patches. This model is composed of two subgraphs: a
subgraph of {\em slow patches}, which are extracted from the smooth
regions of the signal, and a subgraph of {\em fast patches}, which are
extracted from the regions of the signal that contain singularities,
changes in frequency, or energetic transients. We confirm our
theoretical analysis with numerical experimentations using synthetic
signals in section \ref{sec:experiments}, and we demonstrate that our
conclusions are in fact applicable to a larger class of
patch-graphs. The graph models are introduced in section
\ref{ssec:modelgraphs}. Our theoretical analysis of the embedding of
the graph models is given in section \ref{ss:ctestimates}. We evaluate
the performance of the embedding $\Phi$ when $d'$ is small in section
\ref{ssec:pgm}. 

\subsection{The prototypical graph models}
\label{ssec:modelgraphs}
We define the graph models in terms of the nonzero entries in the
associated weight matrix $\bo W$. Without loss of generality, we
assume that the number of vertices $N$ is even.

\paragraph{The slow graph model} The large entries in a weight matrix
$\bo W$ of a patch-graph composed only of slow patches will have large
entries when $|n-m|$ is small\footnote{We assume that the rows/columns
  of $\bo W$ are ordered according to increasing index $n$ of the
  sequence $\{x_n\}$. This assumption does not affect the graph's
  parametrization nor our theoretical conclusions, but allows us to
  interpret the structure in $\bo W$.}: temporal/spatial
proximity implies proximity in patch-space (see e.g. Figure
\ref{fig:Wssec2}-A, top corner). We therefore define the {\em slow
  graph model} as follows.\\

\begin{definition}
  The slow graph $\sm S(N,L)$ is a weighted graph composed of $N$  vertices, 

  \noindent $\bx_1,\ldots,\bx_N$. The weight on the edge $\{\bx_n, \bx_m\}$ is defined by
  \begin{equation} w_{n,m} =
    \begin{cases}
      w_{\sm S} & \text{if $|n-m|\leq L$},\\
      0 & \text{otherwise,} \end{cases} \qquad\text{for}\quad1\leq
    n,m\leq N \quad\text{and}\quad2L+1\leq N. \label{eqn:flatgraph}
  \end{equation} 
  \label{slowgraph}
\end{definition}
The weight $w_s$ is a positive real number that models the distance
between two temporally adjacent patches. The parameter $L$
characterizes the thickness of the diagonal in $\bo W$. The slow graph
is fully connected and each vertex has at most $2L$ neighbors, not
including self-connections (see Figure \ref{fig:graphmodels}). Hence,
we require that $2L+1\leq N$. Finally, note that the slow graph is
distinct from a regular ring, since the first and last vertices are
not connected. We do not consider a regular ring since it would imply
that the underlying signal is periodic.

\paragraph{The fast graph model} We now consider the model for a
patch-graph built from a patch-set comprising only fast patches. As
demonstrated in section \ref{ssec:firstlookpatchgraph}, most of the
entries in $\bo W$ have similar sizes, and appear to be scattered
throughout the matrix: temporal/spatial proximity does not correlate
with proximity in patch space. In fact, fast patches are all far away
from one another. We therefore define the {\em fast graph
  model} as follows. \\

\begin{definition}
  \label{fastgraph-def}
  The fast graph $\sm F(N,p)$ is a random weighted graph composed of $N$
  vertices, $\bx_1,\ldots,\bx_N$. The weight on the edge $\{\bx_n,
  \bx_m\}$ is defined by
  \begin{eqnarray*}
    w_{n,m} &=& w_{m,n} =
    \begin{cases}
      w_{\sm F} & \text{with probability $p$}, \\
      0 & \text{with probability $1-p$} 
    \end{cases}
    \quad \text{if $1\leq n < m\leq N$},\\ 
    \text{and }\quad \quad \quad&\\
    w_{n,m} &=& 1 \quad \text{if}  \quad n = m.
  \end{eqnarray*}
\end{definition}
The weight $w_F$ is a positive real number that models the distance
between two fast patches. The fast graph model is equivalent to a
weighted version of the Erd\"os-Renyi graph model \cite{erdos1960erg},
except that $\sm F(N,p)$ contains self-connections and has edge
weights possibly less than one. The parameter $p$ controls the density
of the edges; $p=1$ corresponds to a fully connected graph (clique).
\begin{figure}[H]
  \centerline{
    \includegraphics[width=.55\textwidth]{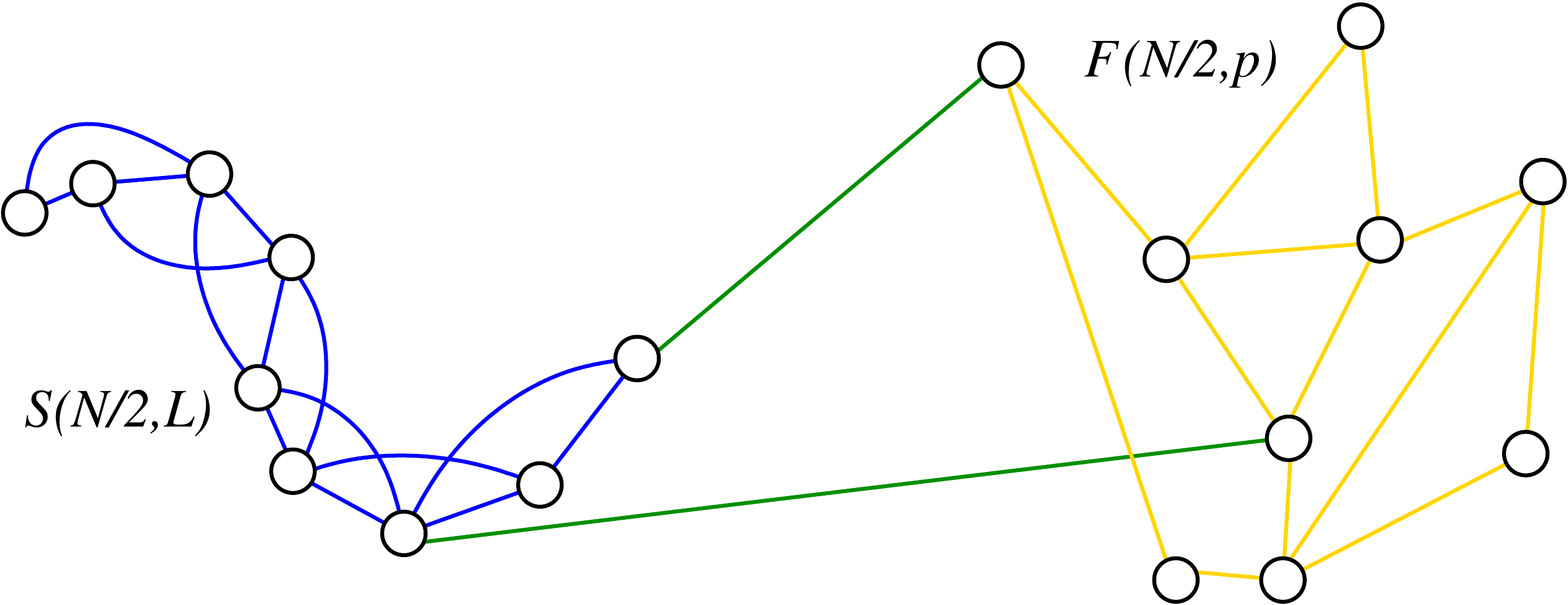} 
  }
  \caption{The fused graph model $\Gamma^\ast(N)$ is composed of a
    slow graph $\sm{S}(N/2,L)$ (blue) and a fast graph $\sm F(N/2,p)$
    (orange), connected by random edges (green).\label{fig:graphmodels}}
\end{figure}
\begin{figure}[H]
\centerline{
    \includegraphics[width=.3\textwidth,height=.3\textwidth]{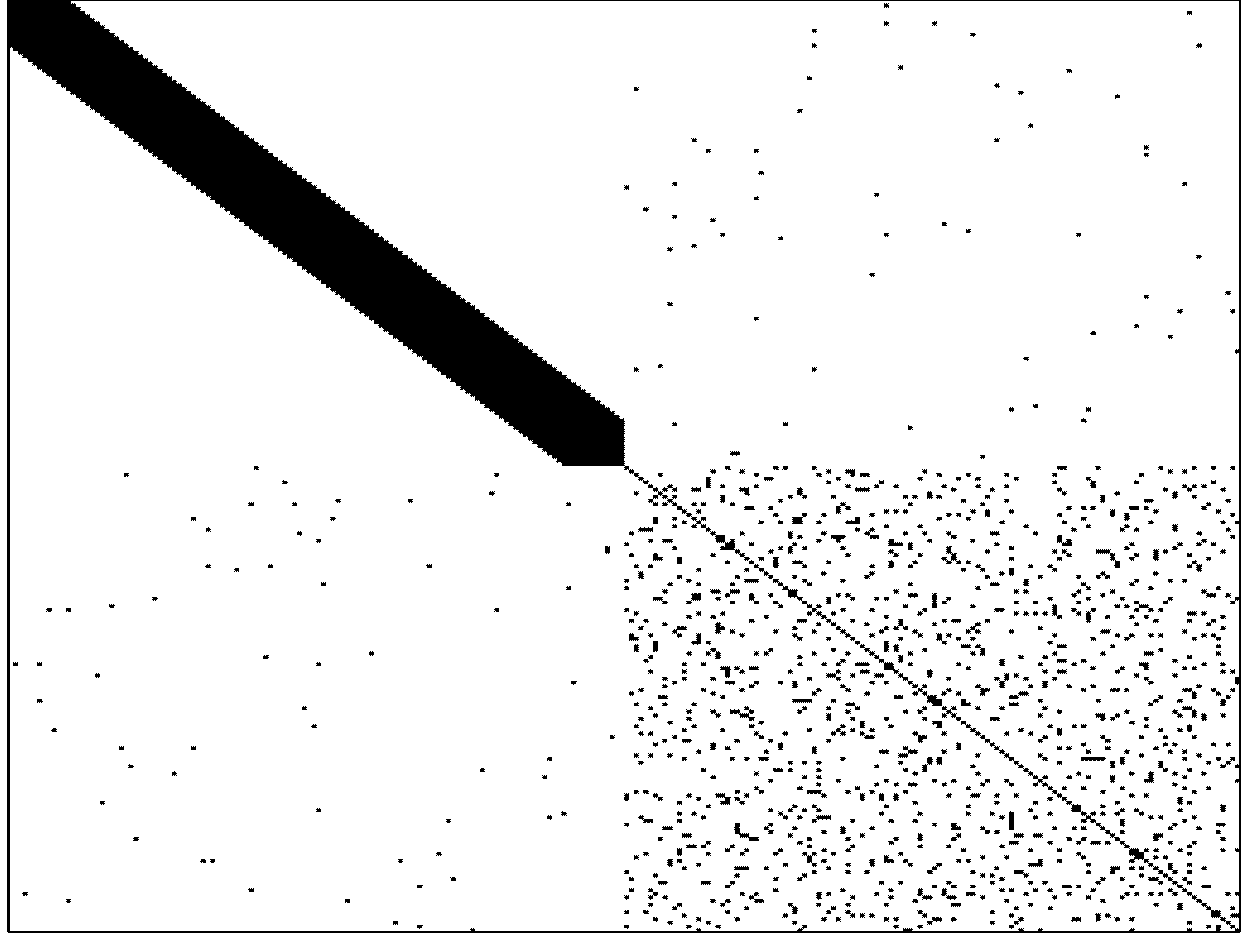}
}
  \caption{The weight matrix $\bo{W}$ of the fused graph model
    $\Gamma^\ast(256)$ is displayed as an image: $w_{n,m}$ is encoded
    as a grayscale value: from white ($w_{n,m} =0$) to black $(w_{n,m}
    =1)$. The entries of $\bo W$ associated with the slow graph appear
    in the upper-left quadrant of $\bo{W}$. Entries associated with
    the fast graph appear in the lower right quadrant. Random edges between
    the fast graph and slow graph appear in the upper right and lower
    left quadrants. \label{fig:randGWmatrix}}
\end{figure}
\paragraph{The fused graph model}  The {\em fused graph
  model} exemplifies the patch-set associated with a signal, or
an image, which exhibits regions of fast and slow changes. The fused
graph combines a slow and a fast subgraph of equal size (see Figure
\ref{fig:graphmodels}). \\

\begin{definition}
  The fused graph $\Gamma^\ast(N)$ is a weighted graph composed of a
  slow subgraph $\sm S(N/2,L)$ and a fast subgraph $\sm F(N/2,p)$. In
  addition, edges between $\sm S(N/2,L)$ and $\sm F(N/2,p)$ are
  created randomly and independently with probability $q$ and assigned
  the edge weight $w_c > 0$.
  \label{fused}
\end{definition}

Edges between $\sm S(N/2,L)$ and $\sm F(N/2,p)$ ensure that
$\Gamma^\ast(N)$ is connected (a requirement for the validity of the
parametrization (\ref{eqn:maptoN1})). These edges allow us to model
patches that are extracted from regions of the image that combine
edges/transients and smooth intensity. If $q$ is so small that no
edges are created between the two subgraphs, then an edge is placed at
random between the two subgraphs to ensure that the final fused graph
is connected.

The true patch-graph is always constructed using a $\nu$ nearest
neighbor rule (see section~\ref{sec:patchgraph}): each patch is
connected to $\nu$ other patches. In order to mimic a true
patch-graph, we adjust the thickness $L$ of the slow subgraph to the
density of the edge connection, $p$, in the fast subgraph, so that on
average, each vertex in the fused graph is connected to $2L$ vertices.
We know that the number of edges between distinct vertices in $\sm
F(N,p)$ is a binomial random variable with expectation
$\tfrac{N(N-1)}{2}p$. Since the total number of edges between distinct
vertices of $\sm S(N,L)$ is equal to\footnote{This is equivalent to
  the number of entries along the first $L$ upper diagonals of the matrix $\bo W$.}
\begin{equation}
  \sum_{j = 1}^L(N-j) = NL-\frac{L(L+1)}{2},
  \label{eqn:numoedgesGNP}
\end{equation}
we choose 
\begin{equation}
  p = \frac{2L}{N-1} - \frac{L(L+1)}{N(N-1)}.
  \label{eqn:pdefined}
\end{equation}
This choice of $p$ guarantees that the expected number of edges in
$\sm F(N,p)$ is equal to the number of edges in $\sm
S(N,L)$. Furthermore, provided that $L = \sm O(\ln(N))$, a short
computation shows that, for large values of $N$, this choice of $p$ 
also ensures that the expected degree of a vertex in $\sm F(N,p)$ is
equal to the average degree of a vertex in $\sm S(N,L)$.
Figure \ref{fig:randGWmatrix} shows the nonzero entries in the
weight matrix associated with one realization of the fused graph model
using parameters $N = 256$, $L = \ceil{2\ln N} = 12$ and $q =
\tfrac{1}{N}$.  Vertices $\bx_n$ with $n \leq 128$ are only connected
to other vertices $\bx_m$ if $|n -m| \leq L$. This connectivity mimics
the spatial (temporal) connectivity present in the smooth parts of an
image (signal).
\subsection{The main result}
\label{ss:ctestimates}
Our goal is to understand the effect of the embedding $\Phi$ defined
by (\ref{eqn:lowdparam}) on the fused graph. It turns out that
studying the embedding of each individual subgraph (slow and fast)
separately is much more tractable than considering the entire fused
graph. To complement our theoretical study of the fast and the slow
subgraphs, we provide numerical evidence in sections \ref{ssec:pgm}
that indicates that our understanding of the embedding of the
subgraphs can be used to analyze the embedding of the fused graph. In
section \ref{sec:experiments}, we confirm that our theoretical
analysis can be applied to true patch-graphs. Instead of studying
$\Phi$ directly, we take advantage of the fact that the embedding
$\Phi$ almost preserves the commute time (see
(\ref{pseudo-isometry})).  We can therefore understand the effect of
the embedding on the distribution of mutual distances $\|\Phi (\bx_n)
- \Phi(\bx_m)\|$ within a subgraph by studying the distribution of the
commute times $\kappa(\bx_n,\bx_m)$ on that subgraph. While it would
appear that it is a straightforward affair to compute the commute time
on the slow graph, the computation becomes rapidly intractable. For
this reason we provide lower and upper bounds for the average commute
time on the slow and fast subgraphs, respectively. This is sufficient
for our needs, since the two bounds rapidly separate even for low
values of $N$. To estimate these bounds, we rely on the connection
between commute times on a graph and effective resistance on the
corresponding electrical network
\cite{Chandra89theelectrical,2000math......1057D}.  Specifically, we
map a graph to an electrical circuit as follows: each edge with weight
$w_{n,m}$ becomes a resistor with resistance $1/w_{n,m}$. The vertices
of the graph are the connections in the circuit.  Given two vertices,
$\bx_n$ and $\bx_m$ in the circuit, one can compute the effective
resistance between these nodes, $R_{n,m}$. The key result
\cite{Chandra89theelectrical} is that $\kappa(\bx_n,\bx_m) = V
R_{n,m}$, where $V$ is the volume of the graph.

Before stating the main Lemma, let us take a moment to compute some
rough estimates of the commute times on the slow and fast graphs. To
get some quick answers, we consider the simplest versions of the two
graph models. When $L =1$, the slow graph $\sm S(N,1)$ is a {\em path}
with self-connections. On a path of $N$ vertices {\em without}
self-connections, the commute time between vertex $\bx_n$ and $\bx_m$
is equal to $2(N-1)|m-n|$. Therefore, the average commute time
(computed over all pairs of vertices) on a path of length $N$ is $\sm
O(N^2)$. While it would make sense that adding edges to a path should
decrease the commute time, this is usually not true
\cite{Lovasz93}. Nevertheless, the presence of edges that allow the
random walk to move forward by a distance $L$ at each time step lead
us to conjecture that the average commute time on $\sm S(N,L)$ should
be of the order $\frac{1}{L}\sm O(N^2)$.  In fact, as we will see in
Lemma \ref{prop:prop2}, the average commute time of the slow graph is
of the order $\frac{1}{L^2}\sm O(N^2)$. With regard to the fast graph,
we can analyze the case where the density of edges $p=1$. In this
case, the fast graph $\sm F(N,1)$ is a {\em complete graph}, or {\em
  clique}, and every vertex is connected to every other vertex.  In a
complete graph, the average commute time is $\sm O(N)$. Since the fast
graph can be regarded as a complete graph whose edges have been
removed with probability $1-p$, we expect the commute time to be
slightly larger than $\sm O(N)$, since removing edges restricts the
random walker's options to get from one vertex to another. Again, in
agreement with our intuition, Lemma \ref{prop:prop2} asserts that in
the fast graph, the commute time is of the order of $[L\ln(N)/\ln(L)]
\sm O(N)$.

We are now ready to state the main lemma. Our results will be stated
in terms of the ``average behavior'' of the commute time on each
graph, a concept that we need to define properly. In the case of the
slow graph, which is deterministic, we consider the average commute
time computed over all pairs of vertices.\\

\begin{definition}
  Let $\kappa_{\sm{S}}$ be the average commute time between vertices in the
  slow graph $\sm S(N,L)$ 
  \begin{equation} 
    \kappa_{\sm{S}} \triangleq  
    \frac{2}{N(N-1)}\sum_{1\leq m<n\leq N}\kappa(\bo{x}_n,\bo{x}_{m}).
  \end{equation}
\end{definition}

In the case of the fast graph, the ``average behavior'' of the commute
time needs to be defined more carefully. Indeed, each fast graph is a
realization of a stochastic process, and therefore we need to consider
the {\em expectation} of the commute time. More precisely, given a
realization, $\sm F$, of a fast graph, we compute the expected commute
time $\E_{\bx_n,\bx_m}\left[\kappa| {\sm F}\right]$ as the expectation
of $\kappa(\bx_m,\bx_n)$ over all possible random assignment of the
vertices $\bx_n$ and $\bx_m$. We then need to consider how
$\E_{\bx_n,\bx_m}\left[\kappa| {\sm F}\right]$ varies as a function of
$\cal F$. Therefore, we compute a second expectation over all possible
random graphs $\sm F$.

\begin{definition}
  The expected commute time $\kappa_{\sm{F}}$ on a fast graph $\sm F$
  generated according to (\ref{fastgraph-def}) is defined by
  \begin{equation}
    \kappa_{\sm{F}} \triangleq \E_{\sm{F}}
    \left[\E_{\bx_n,\bx_m}\left[\kappa| {\sm F}\right]\right].
  \end{equation}
  where the inner expectation is computed over all random assignments of
  the vertices $\bx_n,\bx_m$ given a realization $\sm F$ of a fast graph
  geometry, and the outer expectation is computed over all possible
  realizations $\sm F$ of the fast graph.
\end{definition}\\

\begin{lemma}
  We have
  \begin{equation} 
    \left(N(2L+1)-L(L+1)\right)\frac{2 \left(N+1\right)}{3L^2(L+1)}
    \leq 
    \kappa_{\sm{S}}.
    \label{eqn:kappafc}
  \end{equation}
  We also have
  \begin{equation}
    \kappa_{\sm{F}}\leq
    \left(N(2L+1)-L(L+1)\right)\left(\frac{\ln N}{\ln
        \left(2L-\frac{L(L+1)}{N}+1\right)}+\frac{1}{2}\right),
    \label{eqn:kappadc}
  \end{equation}
  provided that, for all assignments of the vertices $\bx_m$ and $\bx_n$, and for all fast
  graphs $\sm F$, the covariance $\cov (M,R_{m,n})$ between the number of edges, $M$, and
  the effective resistance, $R_{m,n}$, of the associated electrical circuit is
  nonpositive.
  \label{prop:prop2}
\end{lemma}

\begin{proof}
  The proofs are given in appendix \ref{sec:ctproofs}. 
\end{proof}\\

Because $L$ needs to grow logarithmically with $N$ (to ensure that
$\sm F(N,p)$ is connected with high probability; see appendix
\ref{sssec:connectedR}), we choose $L = c\ln N$ for some $c>1$, and the
upper bound on $\kappa_{\sm{F}}$ is negligible relative to the
lower bound on $\kappa_{\sm{S}}$, when $N$ is large.\\

\begin{corollary}
  Assume that $L = c\ln N$ for some constant $c>1$.  It follows that,
  as $N\rightarrow \infty$, the lower bound on $\kappa_{\sm{S}}$ grows
  like $\left(\tfrac{N}{\ln N}\right)^2$, and the upper bound on
  $\kappa_{\sm{F}}$ grows like $\tfrac{N(\ln N)^2}{\ln\ln N}$.
  Furthermore, the lower bound on $\kappa_{\sm{S}}$ grows faster than
  the upper bound on $\kappa_{\sm{F}}$, and so with a probability that
  approaches one as $N\rightarrow\infty$,
  \begin{equation*}
    \lim_{N\rightarrow\infty}\frac{\kappa_\sm{F}}{\kappa_\sm{S}} = 0.
  \end{equation*}
  \label{prop:prop21} 
\end{corollary} 

\begin{proof} 
  Notice that $\kappa_\sm{S}$ is bounded away from zero. Because the
  choice of $L$  guarantees that the fast graph is connected with a
  probability approaching one, $\kappa_\sm{F}$ is finite with
  probability approaching one. Therefore the ratio
  $\kappa_\sm{F}/\kappa_\sm{S}$ is bounded below by zero and from
  above by a ratio of the bounds from Lemma \ref{prop:prop2}. The
  ratio of bounds goes to zero, which follows from a simple, but
  lengthy, limit calculation.
\end{proof}\\

We can translate the corollary in terms of the mutual distance between
vertices of the subgraphs after the embedding $\Phi$: 
$\Phi({\sm F}(N,p))$ will be more concentrated than $\Phi({\sm S}(N,L))$.

\subsection{Spectral decomposition of commute times on the graph
  models\label{ssec:pgm}} 
The results of section \ref{ss:ctestimates} apply to the exact commute
times on the graph models. However, as mentioned in section
\ref{ssec:paramthegraph}, it is more practical to use a truncated
version of the spectral expansion of the commute time, defined by
Equation (\ref{eqn:commutetime}). We also noticed that the commute
time encompasses the short term evolution ($t \approx 0$) as well as
the asymptotic regime ($t \rightarrow \infty$) of the behavior of the
random walk. Neglecting eigenvalues $\phi_k$ for large $k$ emphasizes
the long term behavior of the random walk, and we expect that it
should further increase the difference between the slow and fast
graphs. In this section, we confirm experimentally that approximating
the commute times by truncating the expansion (\ref{eqn:commutetime})
actually emphasizes the separation between the fast subgraph and the
slow subgraph in the fused graph model. In all the numerical
experiments in this section, unless otherwise stated, we fix $N =
1024$, $L = \ceil{2\ln(N)}$, $p$ is chosen according to
(\ref{eqn:pdefined}), $q = 1/N$, and $w_{\sm S} = w_{\sm F} = w_c
=1$. In all experiments, we compute the eigenvalues $\{\lambda_k\}$ of
the matrix $\bo D^{-1/2}\bo W\bo{D}^{-1/2}$ associated with the fast
graph, the slow graph, and the fused graph.
\paragraph{Slow and fast subgraphs: two different dynamics revealed by
  the spectral decomposition} We first provide a back-of-the-envelope
computation of the spectrum of the slow and fast graphs. As we have
noticed before, the slow graph model is a ``fat'' path. We know that
the spectrum of a path without self-connections \cite{chung97} is
given by
\begin{equation*}
  \cos\left [\pi (k-1) / (N-1) \right], \quad k =  1,2,\ldots,N.
\end{equation*}
We expect therefore that the eigenvalues associated with the slow
graph will decay slowly away from one for small $k$. Figure
\ref{fig:modelspectra} (inset) displays the eigenvalues associated
with the slow graph model. As expected, the spectrum is flat around
$k=0$ and exhibits the slowest decay of all the graph models.
We use the similarity between the fast graph model and the
Erd\"os-Renyi graph to predict the spectrum of the fast graph.  Except
for $\lambda_1=1$, all the other eigenvalues of an Erd\"os-Renyi graph
asymptotically follow the Wigner semicircle distribution
\cite{Chung27052003}. Our numerical experiments confirm this
prediction: as shown in Figure \ref{fig:modelspectra}-right, the
eigenvalues of the fast graph appear to be distributed along a
semicircle. 

The decay of the spectrum has a direct influence on the dynamics of
the random walk. Specifically, the spectral gap controls the
{\em mixing rate}, which measures the expected number of time-steps
that are necessary to reduce the distance between the probability
distribution after $t$ steps ${\bo P}^{(t)}_{n,m}$ and the stationary
distribution $\pi_m$ by a certain factor \cite{Vempala05}. This
concept is justified by the fact that the convergence of ${\bo
  P}^{(t)}_{n,m}$ is exponential \cite{durrettgd}, and is given by
\begin{equation}
  \max_{n,m}\left|\frac{\bo
      P^{(t)}_{n,m}}{\pi_m}-1\right|\leq\frac{\lambda_{max}^{t}}{\pi_{min}},
  \quad t =1,\ldots
  \label{eqn:relativepwdistance}
\end{equation}
where $\lambda_{max} = \max\{\lambda_2,|\lambda_{N}|\}$ (which is
related to the spectral gap), and $\pi_{min}$ is the smallest entry of
the stationary distribution.  Since $\lambda_2$ is much larger in the
slow graph than in the fast graph, we expect that convergence to the
associated stationary distribution will take longer on the slow graph
than on the fast graph.
\paragraph{The dynamic of the fused graph is enslaved by the slow subgraph}
We now consider a random walk on the fused graph. If this random walk
begins at $\bx_n$ in the fast subgraph of the fused graph,
then after a small number of steps, $t_0$, the probability of finding
the random walker at any other vertex $\bx_m$ in the fast subgraph is
close to the stationary distribution, $\bo P^{t_0}_{n,m} \approx
\pi_m$. On the other hand, during the same amount of steps, a random
walk initialized in the slow subgraph will only explore a small
section of the slow subgraph, and consequently, the transition
probabilities will still be similar to its initial values $\bo
P^{(t_0)}_{n,m} \approx \bo P_{n,m}$. As a result, the restriction
imposed by the geometry of the slow subgraph is expected to decrease
the convergence rate of the transition probabilities on the fused
graph. We confirm this analysis with experimental results. Figure
\ref{fig:modelspectra} (inset) shows that for $k<23$ the eigenvalues
associated with the fused graph and the eigenvalues associated with
the slow graph exhibit slow decay away from one, thereby increasing
the convergence rate given in (\ref{eqn:relativepwdistance}). For $25
\leq k \leq 400$, the eigenvalues of the fused graph decay at a rate
similar to that of the fast graph. Finally, for $k \ge 400$ the
eigenvalues of the fused graph join those of the slow graph (see  also
the histogram in Figure \ref{fig:modelspectra}-right). We
have observed experimentally that these transitions in the behavior of
the spectrum of the fused graph are not affected by varying the parameters $N$, $L$, and
$q$. We%
\begin{figure}[H]
\centerline{
    \includegraphics[width=.35\textwidth]{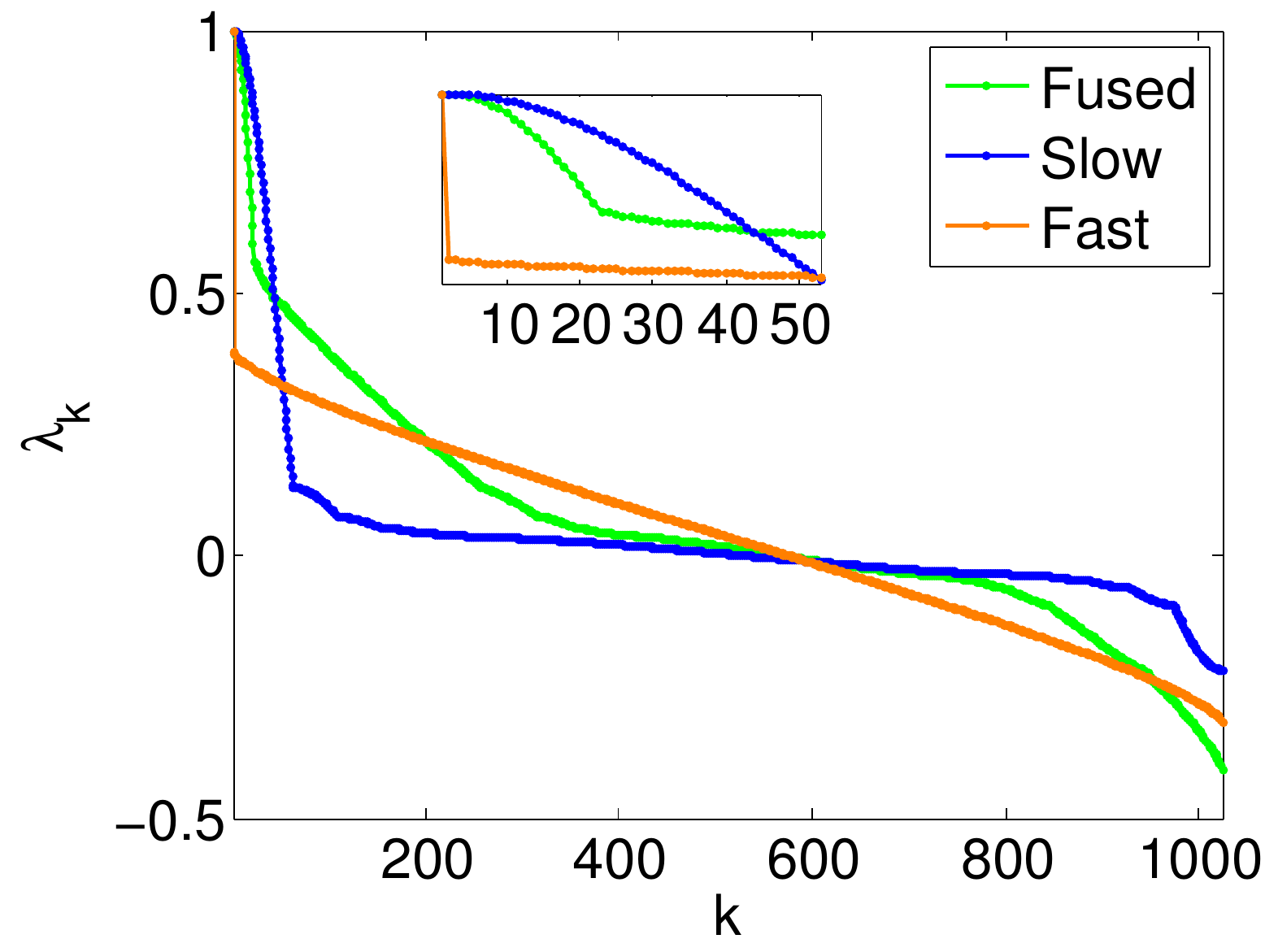} 
    \includegraphics[width=.35\textwidth]{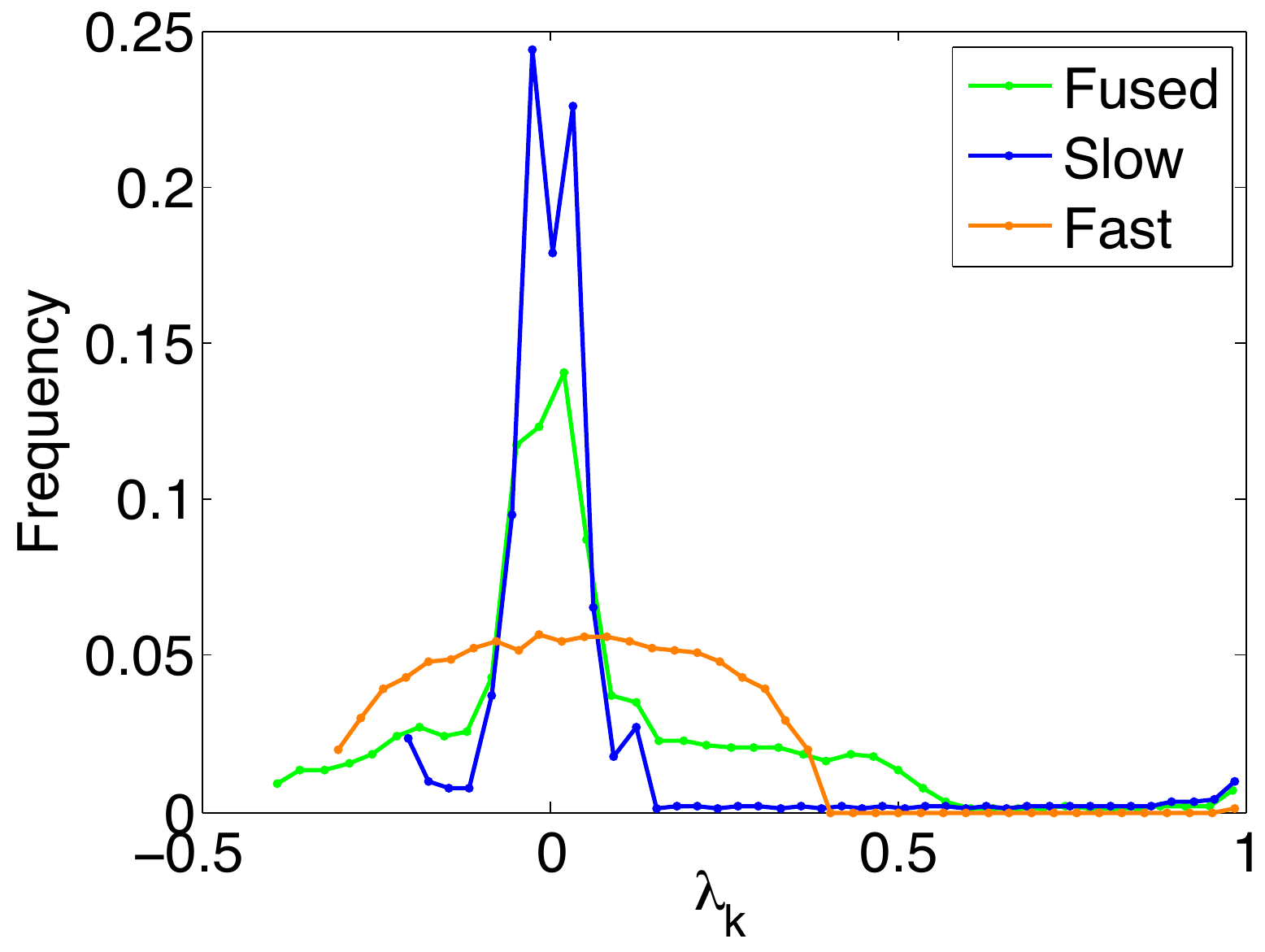} 
}
  \caption{The eigenvalues $\lambda_k$ of the matrix $\bo D^{-1/2}\bo
    W\bo D^{-1/2}$ associated with the fused (green), slow (blue), and
    fast (orange) graphs. Left: $\lambda_k$ as a function
    of $k$; right: histogram of the $\lambda_k$.\label{fig:modelspectra}}
\end{figure}
\begin{figure}[H]
\centerline{
    \includegraphics[width=.325\textwidth]{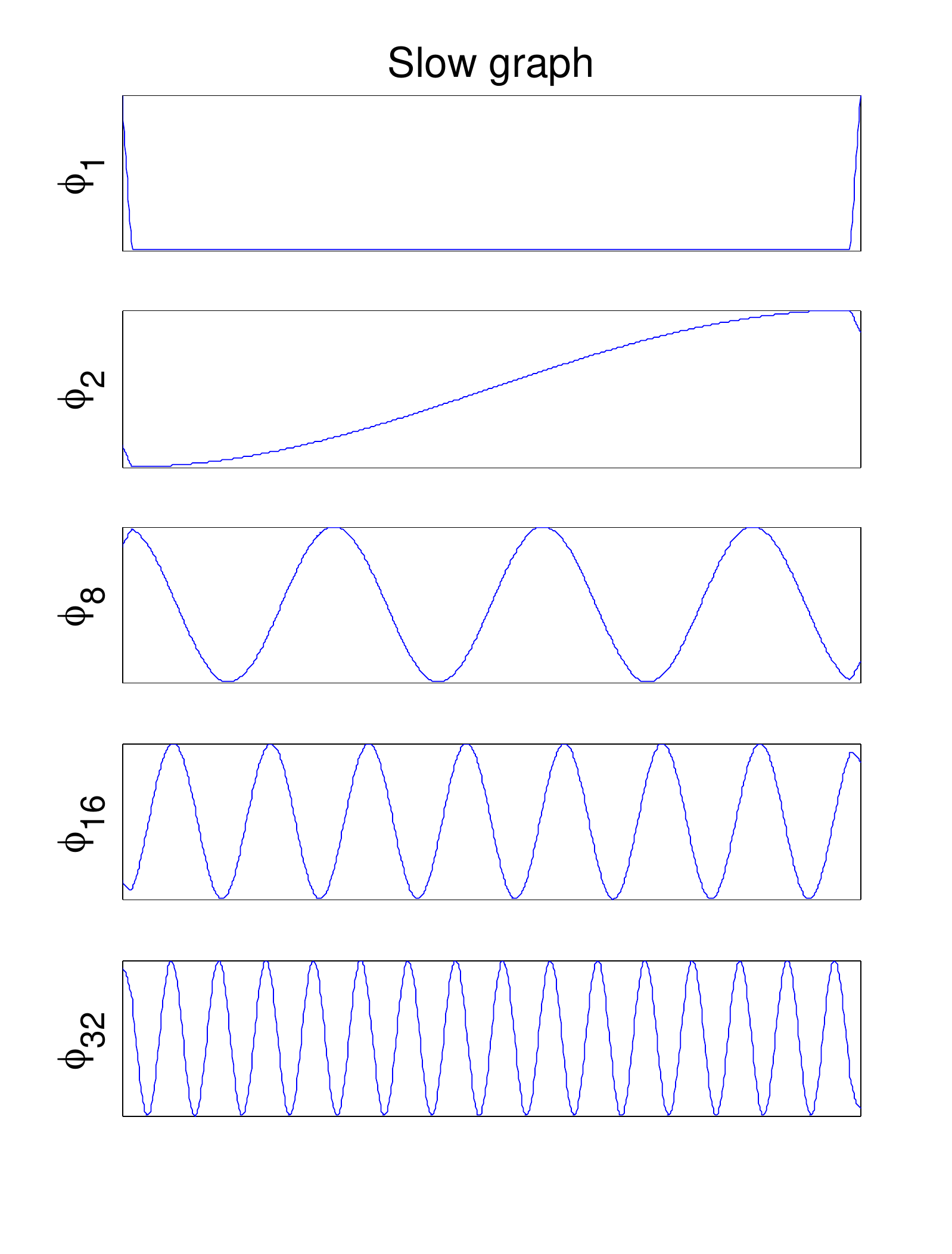} 
    \includegraphics[width=.325\textwidth]{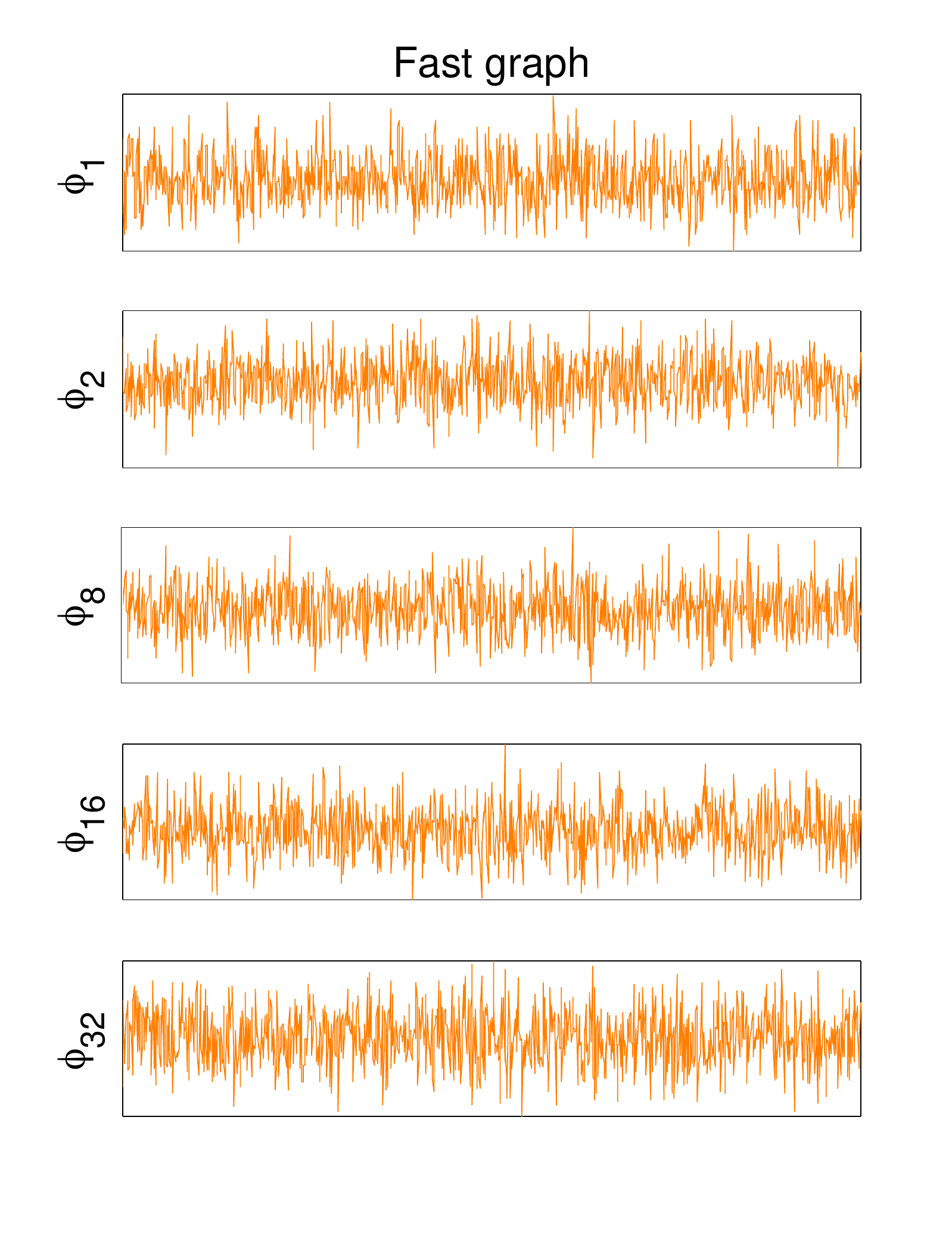} 
    \includegraphics[width=.325\textwidth]{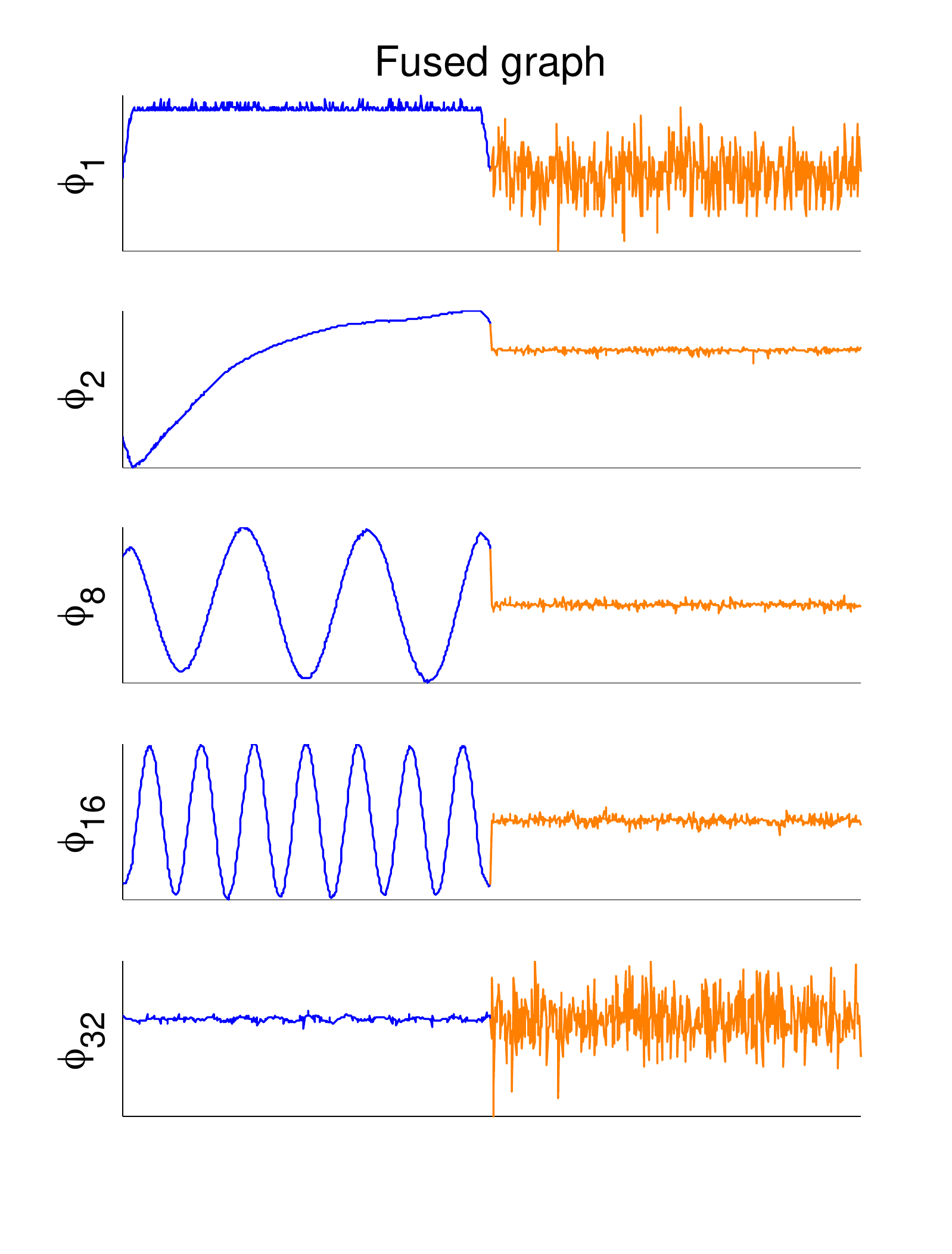} 
}
  \caption{The eigenvectors
    $\{\phi_1,\phi_2,\phi_8,\phi_{16},\phi_{32}\}$ associated with the
    slow (left), fast (center), and fused (right) graphs.
    Right: the large amplitude of the eigenvectors $\phi_k$ on the first half
    of vertices (blue) belonging to the slow subgraph leads to a larger separation between the fast and slow subgraphs
    when truncating the commute time expansion.
    \label{fig:modelvectors}}
\end{figure}
\noindent conclude that the slow subgraph has the largest influence
on the first few (small $k$) eigenvalues $\lambda_k$ of the fused graph. 

\paragraph{The eigenvectors of the fused graph and their impact
  on the commute time}
The transition exhibited in the spectrum of the fused graph can also
be detected in the corresponding eigenvectors $\phi_k$.  Figure
\ref{fig:modelvectors} shows the eigenvectors
$\{\phi_1,\phi_2,\phi_8,\phi_{16},\phi_{32}\}$ corresponding to the
three graph models.  The first eigenvector $\phi_1$ has entries equal
to the square root of the stationary distribution, $\phi_1(\bx_n) =
\sqrt{\pi_n}$, and is not used in the expansion
of the commute time (\ref{eqn:commutetime}). As expected, the random
walk spends most of its time inside the slow subgraph of the
fused graph, as indicated by the larger values of $\phi_1$ for the
first (blue) $N/2$ vertices (see Figure
\ref{fig:modelvectors}-right). The eigenvectors
$\{\phi_2,\phi_8,\phi_{16}\}$ of the fused graph exhibit large
amplitude oscillations over the vertices belonging to the slow
subgraph (first half -- shown in blue -- of the plots in Figure
\ref{fig:modelvectors}-right), which resemble those found in the
eigenvectors associated with the slow graph (Figure
\ref{fig:modelvectors}-left). As $k$ increases, the eigenvectors
$\phi_k$ of the fused graph become more and more similar to the eigenvectors of the
fast graph.

The impact of the eigenvectors $\phi_k$ on the commute time on the
fused graph can be analyzed by estimating the size of the terms
\begin{equation}
  \frac{1}{1-\lambda_k} \left(\frac{\phi_k(\bx_n)}{\sqrt{\pi_n}} -
    \frac{\phi_k(\bx_m)}{\sqrt{\pi_m}}\right)^2
  \label{single}
\end{equation}
in the spectral expansion (\ref{eqn:commutetime}) of the commute time
$\kappa$.  We claim that $\kappa (\bx_n,\bx_m)$ will be small if both
vertices $\bx_n$ and $\bx_m$ are in the fast subgraph, and that
$\kappa$ will be large if either vertex is in the slow subgraph.  

We can first estimate the size of $\phi_k(\bx_n)/\sqrt{\pi_n} -
\phi_k(\bx_m)/\sqrt{\pi_m}$. We observe that the eigenvectors $\phi_k$
for small values of $k$ have large amplitude oscillations on vertices
belonging to the slow subgraph, but are relatively constant on the
fast subgraph (see Figure \ref{fig:modelvectors}-right). Therefore,
for small values of $k$, each term (\ref{single}) will be small when
$\bx_n$ and $\bx_m$ both belong to the fast subgraph (we also have
$\pi_n \approx \pi_m$ when two vertices belong to the same subgraph).
Conversely, these terms will be large when either $\bx_n$ or $\bx_m$
belongs to the slow subgraph. While this analysis of the size of the
terms (\ref{single}) only holds for small values of $k$, it turns out
that these are the terms that have the largest influence in the
expansion of the commute time (\ref{eqn:commutetime}). Indeed, the
spectrum of the fused graph decays slowly, and therefore the first few
coefficients $(1-\lambda_k)^{-1}$ in the commute time expansion
(\ref{eqn:commutetime}) are much larger than the remainders, and
therefore the terms (\ref{single}) for small values of $k$ will
provide the largest contribution in the expansion of the commute time.

We conclude that $\kappa (\bx_n,\bx_m)$ is small when $\bx_n$ and
$\bx_m$ belong to the fast subgraph, and $\kappa (\bx_n,\bx_m)$ is
large when either vertex is in the slow subgraph.  Furthermore, we
expect that this difference will be further magnified if we replace
the exact expansion of $\kappa$ in (\ref{eqn:commutetime}) by an
approximation that only includes the first few values of $k$.
\paragraph{The truncated spectral expansion of the commute time
  increases the contrast between the slow and fast subgraphs} We finally
come to the heart of the section: the numerical computation of the
average approximate commute time defined by
\begin{equation}
  \kappa^\prime = \frac{2}{N(N-1)} \sum_{n< m} \|\Phi (\bx_n) - \Phi(\bx_m)\|^2.
  \label{approx-commute}
\end{equation}
Because of (\ref{isometry}), we expect that $\kappa^\prime$ will be
close to the true commute time $\kappa$. We compute $\kappa^\prime$
for the three graphs: slow, fast and fused.  We
generated 25 realizations of the fast and fused graphs, and we
estimated the expected commute time with the sample
mean, given by $\kappa^\prime$ in (\ref{approx-commute}).

Figure \ref{fig:kapparatios}-A displays $\kappa^\prime_{\sm
  F}/\kappa^\prime_{\sm S}$ as a function of the number of terms $d'$
used in the embedding (\ref{eqn:lowdparam}), for several values of the
number of vertices $N$, for the slow and fast graphs. Our theoretical
analysis of $\kappa_{\sm F}/\kappa_{\sm S}$, performed in Corollary
\ref{prop:prop21}, is only valid for large values of $N$.
Nevertheless, our numerical simulations indicate that for very low
values of $N$, $\kappa_{\sm F}$ is already smaller than $\kappa_{\sm
  F}$, since all ratios are below one (see Figure
\ref{fig:kapparatios}-A). Furthermore, we see that this ratio is even
smaller for smaller values of $d'$. We observe similar results when
the commute times $\kappa^\prime_{\sm S}$ and $\kappa^\prime_{\sm F}$
are computed within the slow the fast subgraphs of the fused graph
(see%
\begin{figure}[H]
  \centerline{
    \includegraphics[width=.37\textwidth]{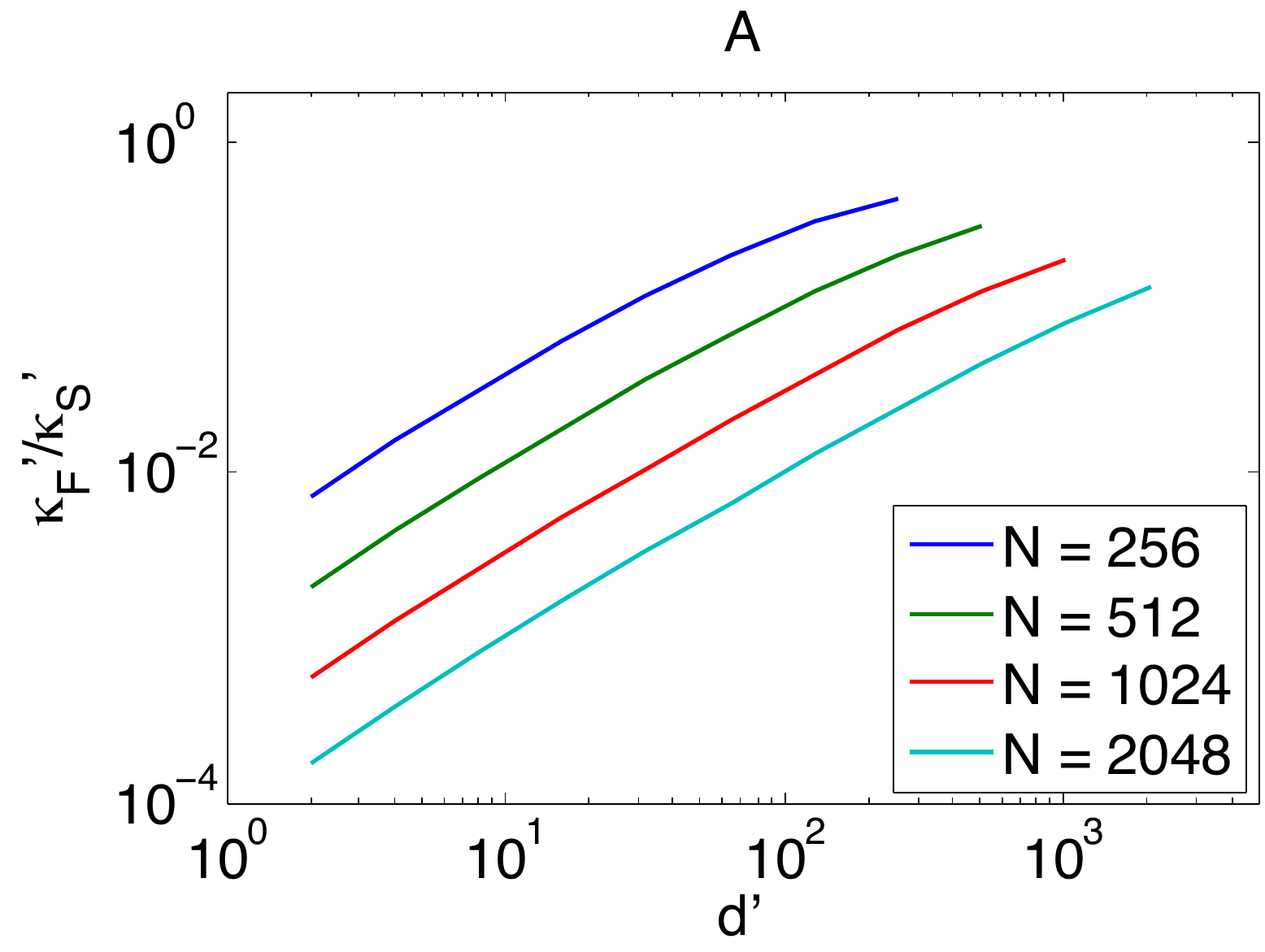} 
    \includegraphics[width=.37\textwidth]{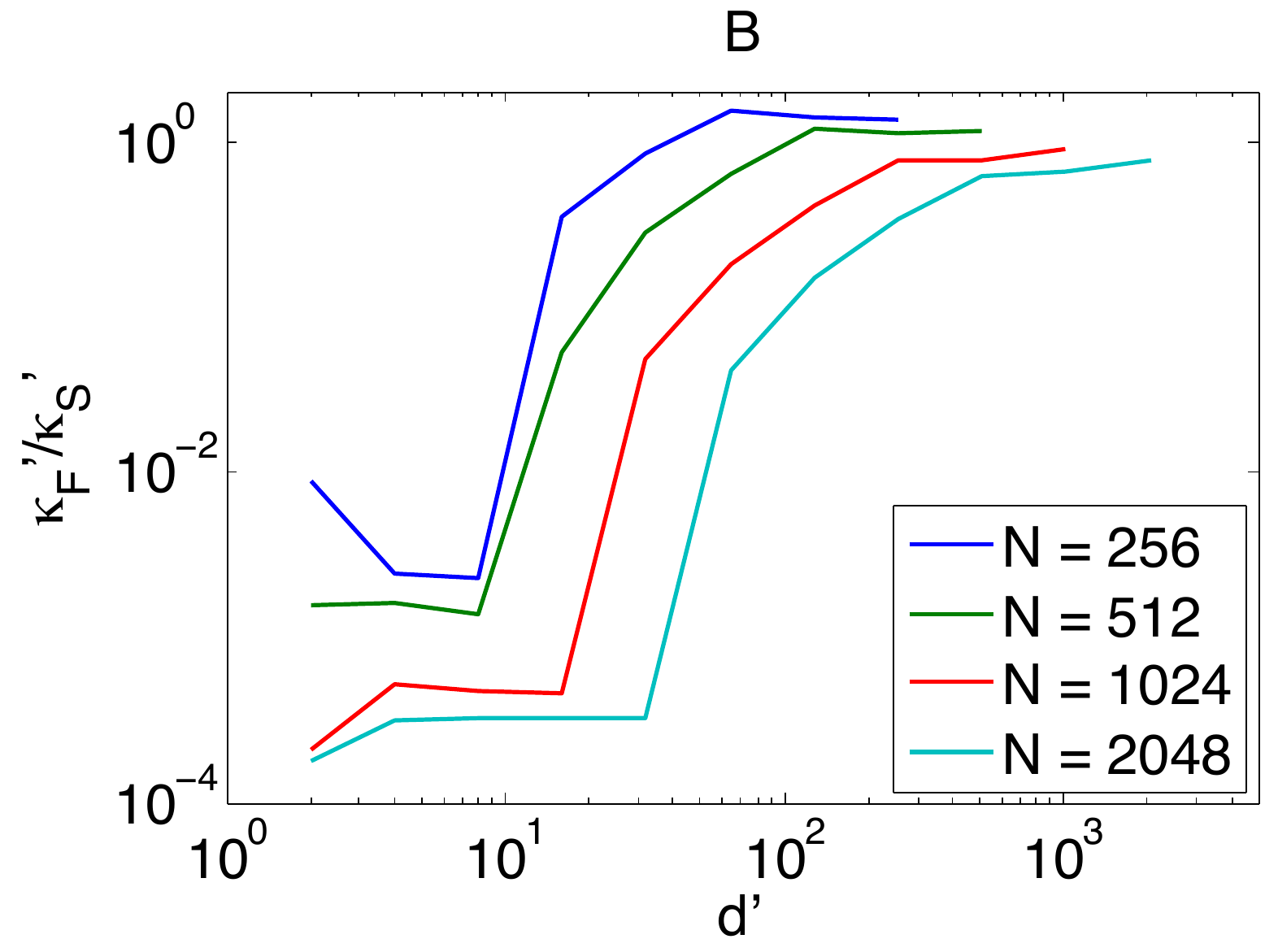} 
}
  \caption{$\kappa^\prime_{\sm F}/\kappa^\prime_{\sm S}$ as a function of the
    dimension $d'$ of the embedding $\Phi$, for several values of the
    number of vertices $N$. Left: slow  $\sm S$ and fast $\sm F$ graphs separately; right: slow and fast
    subgraphs in the fused graph $\Gamma^\ast$. 
    \label{fig:kapparatios}}
\end{figure}
\begin{figure}[H]
  \centerline{
    \includegraphics[width=.37\textwidth]{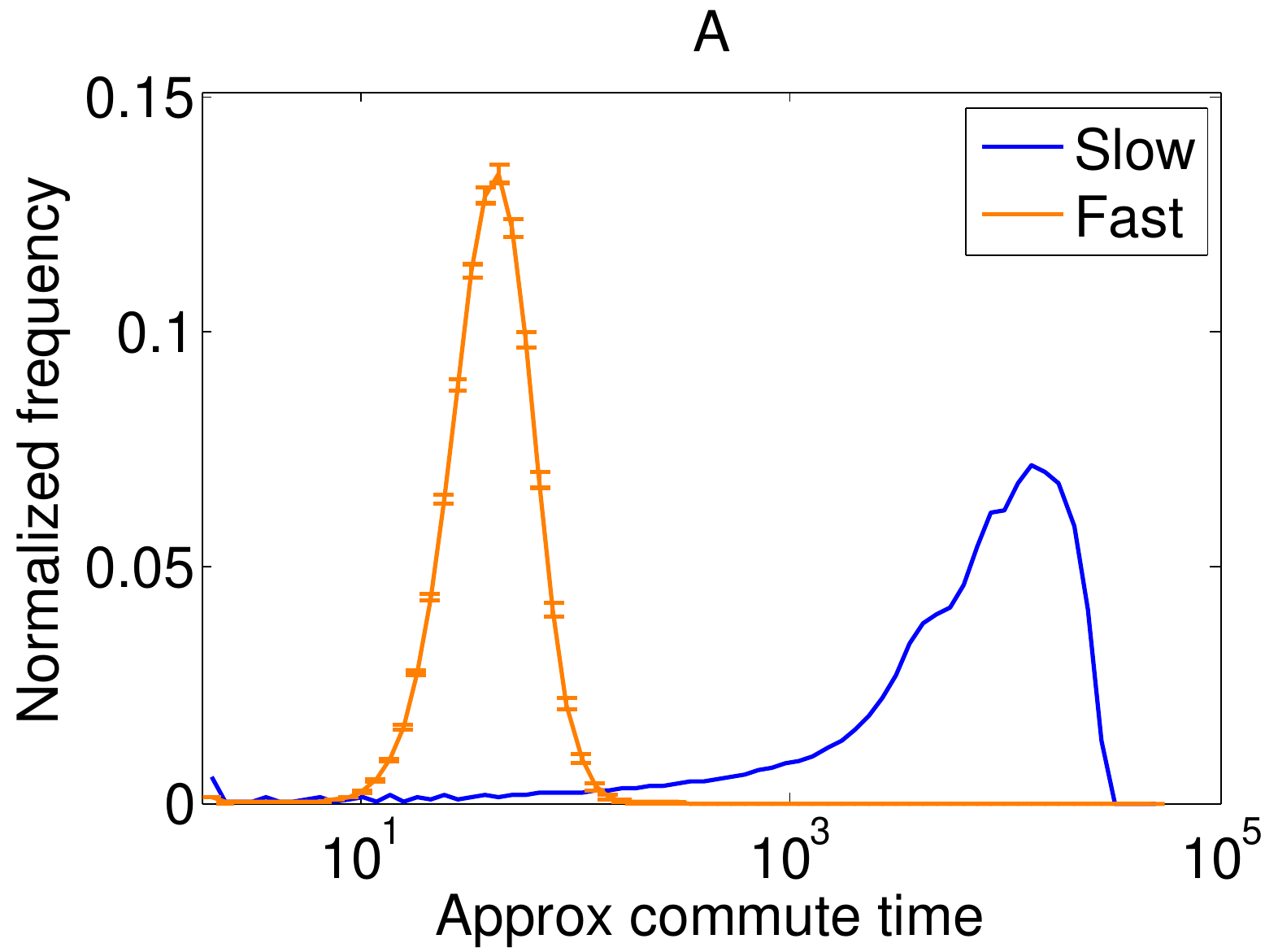} 
    \includegraphics[width=.37\textwidth]{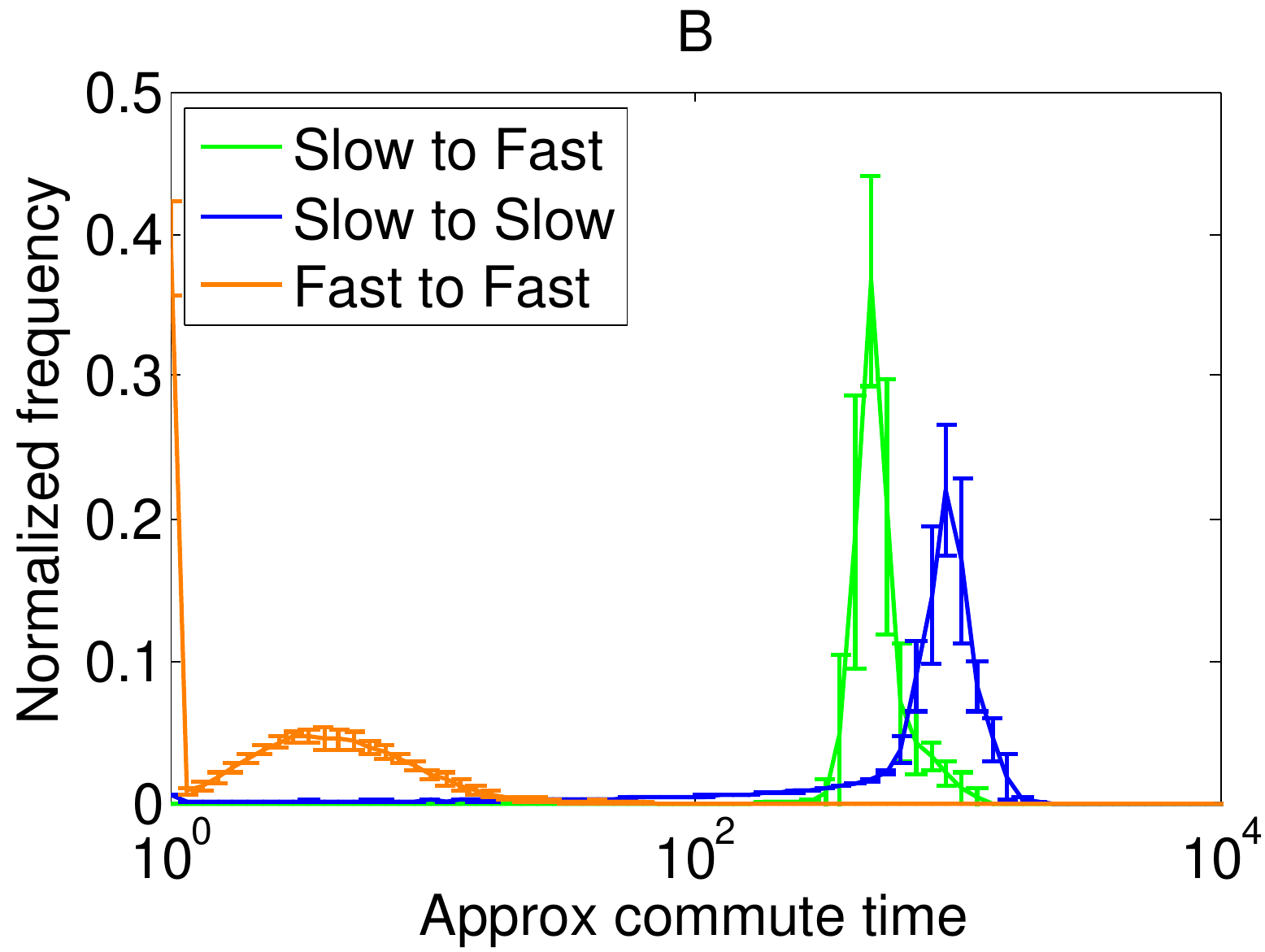} 
}
  \caption{Histogram of  $\kappa^\prime$. Left:
    slow graph $\sm S$ and fast graph $\sm F$. Right: $\kappa^\prime$
    for the three types of transition between the subgraphs of the  
    fused graph $\Gamma^\ast$. Note the logarithmic scale on the
    horizontal axes.
    \label{fig:kappadistributions}}
\end{figure}
\begin{figure}[H]
  \centerline{
    \includegraphics[width=.37\textwidth]{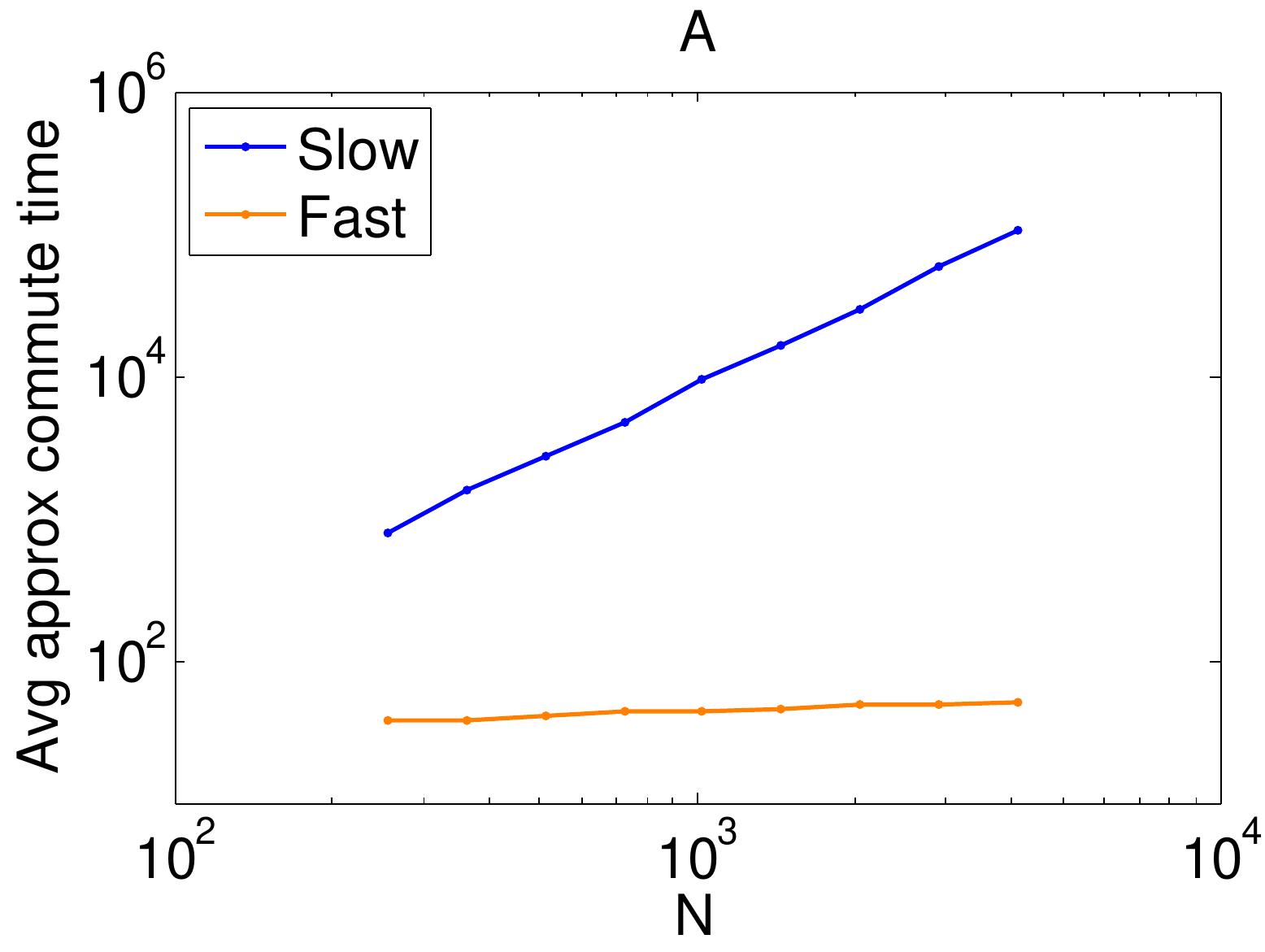} 
    \includegraphics[width=.37\textwidth]{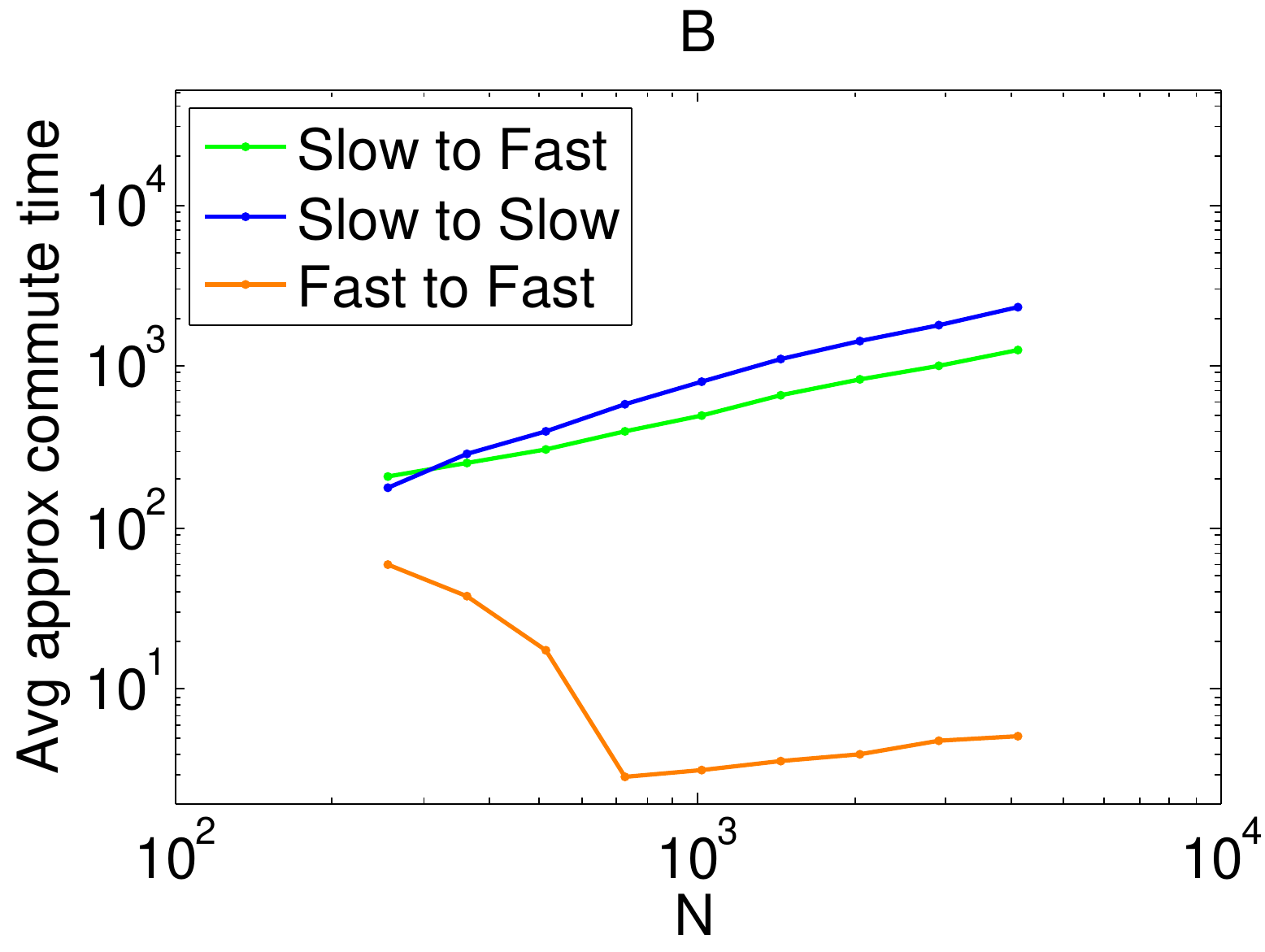} 
}
  \caption{ $\kappa^\prime$ as function of $N$. Left: slow graph $\sm
    S$ and fast graph $\sm F$. Right: $\kappa^\prime$
    for the three types of transition between the subgraphs of the  
    fused graph $\Gamma^\ast$.
    \label{fig:kappaboundsexperimental}}
\end{figure}
\noindent Figure \ref{fig:kapparatios}-B). These results confirm that the
embedding $\Phi$ will further concentrate the vertices of the fast
graph if $d'$ is chosen to be much smaller than $N$. We have observed
experimentally that choosing $d' \approx \ln(N)$ leads to the smallest
ratio of averages, not only on the graph models, but also on the
general patch-graphs studied in section \ref{sec:experiments}.

The enslaving of the fused graph by the slow graph is clearly shown in
Figure \ref{fig:kappadistributions}-B, where the normalized
histogram of $\kappa^\prime$ is shown for the three types of
transition between the subgraphs of the fused graph $\Gamma^\ast$:
slow $\rightarrow$ slow, fast $\rightarrow$ fast, and slow
$\rightarrow$ fast. The histogram of the slow $\rightarrow$ fast
transition is very similar to the histogram of the slow $\rightarrow$
slow transition, clearly indicating that once the random walk is
trapped in the slow subgraph, the presence of the fast subgraph does
not help the random walk escape from the slow graph. We also notice
that the average of $\kappa^\prime$ for the fast $\rightarrow$ fast
transition is roughly two orders of magnitude smaller than the average
of $\kappa^\prime$ for the slow $\rightarrow$ slow, or slow
$\rightarrow$ fast transitions. In addition, the variance of each
distribution is small enough to limit the overlap between the
distributions.

Figure \ref{fig:kappaboundsexperimental} displays $\kappa^\prime$ as a
function of the number of vertices $N$, where $d' = \ln N$. Again,
this result confirms that the asymptotic analysis of the ratio
$\kappa_{\sm F}/\kappa_{\sm S}$, performed in Corollary
\ref{prop:prop21}, actually holds for very small values of
$N$. Indeed, whether the slow and fast graphs are considered
separately (Figure \ref{fig:kappaboundsexperimental}-A), or are the
components of the fused graph (Figure
\ref{fig:kappaboundsexperimental}-B), the ratio $\kappa_{\sm
  F}/\kappa_{\sm S} \rightarrow 0$ (note the logarithmic scale). It is
important to bear in mind that when analyzing images, $N$ is typically
of the order of $10^6$ and therefore our theoretical analysis will
hold without any difficulty. Lastly, we again note in Figure
\ref{fig:kappaboundsexperimental}-B that the transitions slow
$\rightarrow$ fast in the fused graph have the same dynamics as the
transition slow $\rightarrow$ slow.

\subsection{Summary of the experiments}
\label{ss:theoryconclusion}
We have confirmed experimentally that embedding the fused graph using
$\Phi$ shrinks the mutual distance between vertices of the fast
subgraph, effectively concentrating these vertices closer to one
another. As a result, the embedding helps divide the fused graph into
the slow and the fast subgraphs by concentrating the vertices of the
fast subgraph away from the vertices of the slow subgraph.  Our
analysis of the embedding is based on the fact that $\Phi$
approximately preserves the commute time measured on the fused
graph. Furthermore, we have demonstrated that a truncated version of
the commute time, $\kappa^\prime$, is even more conducive to
identifying vertices of the fast subgraph of the fused graph.

The implication of these results is that the embedding of the true
patch-graph $\Gamma$ using $\Phi$ will concentrate the ``anomalous''
patches, which contain rapid changes in the signal, away from the
baseline patches. This concentration of the fast anomalous patches
happens for values of the embedding dimension $d'$ that are of the
order of $\ln (N)$: this choice of $d'$ results in a low-dimensional
embedding of the patch-graph. Because the fast patches are more
clustered after embedding, their detection -- for the purpose of
detection of anomalies, classification, or segmentation -- will become
much easier.  Finally, we note that our theoretical analysis can be
extended to a more general context where patches are replaced by a
vector of local features extracted from elements of a large
dataset. The only requirement is that the graph of features exhibit a
geometry similar to the fused graph $\Gamma^\ast$.
\section{Numerical experiments with synthetic signals}
\label{sec:experiments}
In this section we validate our theoretical results using synthetic
signals. Each signal is the realization of a stochastic process with a
prescribed autocorrelation function. We study two types of stochastic
processes: one that generates signals that transition from low to high
local frequency, and a second one that yields signals with varying local smoothness. We argue that
these signals embody the types of local changes that are of
fundamental importance in many areas of image processing. For both
classes of signals, we embed the patch-sets using $\Phi$ in
(\ref{eqn:lowdparam}). We study the property of the embedding by
quantifying the average commute time $\kappa^\prime$, defined in
(\ref{approx-commute}) between fast%
\begin{figure}[H]
  \centerline{
    \includegraphics[width=.6\textwidth]{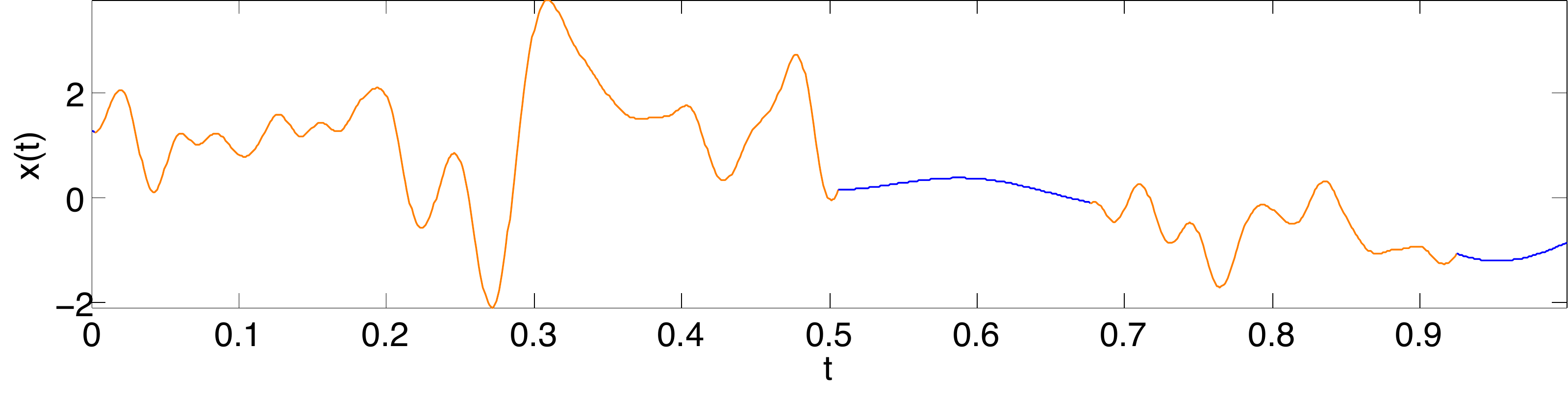} 
  }
  \caption{A realization of the time-frequency model. The low
    frequency portion ($\beta_{\sm S} = 8$) is shown in blue; the high
    frequency portion ($\beta_{\sm F} =256$) is shown in orange. There
    are four subintervals ($\mu =3$).
    \label{fig:autocorrelations1}}
\end{figure}
\begin{figure}[H]
  \centerline{
    \includegraphics[width=.6\textwidth]{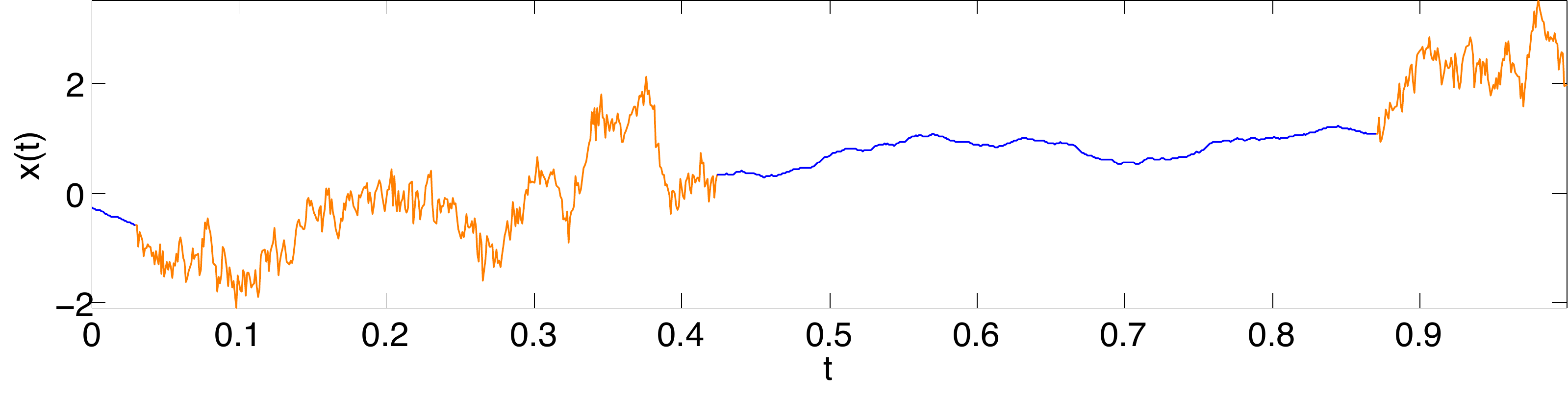} 
  }
  \caption{A realization of the local regularity model. The smooth
    portion ($H_\sm S = 0.9$) is shown in blue; the irregular portion 
    ($H_{\sm F}=0.3$) is shown in orange. There are four subintervals ($\mu =3$).
    \label{fig:autocorrelations11}}
\end{figure}
\noindent and slow patches, and we compare
the numerical results with the theoretical predictions given in
section \ref{ssec:theory}.
\subsection{The signals} \label{ssec:syntheticsignalmodels} We
consider two types of models: a time-frequency signal model and a
local regularity signal model. Each model is characterized by an
autocorrelation function.  The autocorrelation function can be
modified using a parameter that controls the local frequency, or the
local regularity of the signal. We partition the interval $[0,1]$ into
subintervals over which the autocorrelation parameter is kept
constant. The  parameter alternates between two different
values creating subintervals of alternating local frequency, or
alternating local regularity. The number of alternations is chosen
randomly according to a homogeneous Poisson process with intensity
$\mu$: there are on the average $\mu +1$ subintervals. A simpler
version of this model has been used in \cite{Cohen:1997} to mimic the
presence of edges in images. Unlike the model used in
\cite{Cohen:1997}, we adjust the signal defined on each subinterval so
that the result is continuous on $[0,1]$.  In all experiments that we
report here we use $\mu = 3$. The autocorrelation function associated
with the time-frequency signal model is given by
\begin{equation}
  \E(x(t)\overline{x(t+\tau)}= 2\left(\frac{1+\cos \left(2\pi
        \tau\right)}{2}\right)^{\beta}-1, 
  \label{eqn:autocorrelation}
\end{equation}
where $\tau\in [0,1)$, $\beta\ge 0$.  As the autocorrelation parameter
$\beta$ increases, the range of frequencies present in the signal also
increases. Figure \ref{fig:autocorrelations1} displays a realization
of this model where the signal's covariance parameter alternates four
times between $\beta_{\sm S} = 8$, and $\beta_{\sm F} = 256$.  See
appendix \ref{ssec:generatingzk} for more on generating a signal from
the time-frequency signal model.  The autocorrelation function associated with the
local regularity signal model is equal to that of fractional Brownian
motion, given by
\begin{equation}
  \E(x(\tau_1)\overline{x(\tau_2)}) =
  \frac{1}{2}\big(|\tau_1|^{2H}+|\tau_2|^{2H}-|\tau_2-\tau_1|^{2H}\big), 
\end{equation}
where $H$ is the Hurst parameter. As $H$ decreases, the local
regularity decreases.  A realization of this model is shown in Figure
\ref{fig:autocorrelations11} where the signal's covariance parameter alternates four times
between $H_\sm{S} = 0.9$ and $H_\sm{F}=0.3$.  We use the method
described in \cite{abry96} to generate the fractional Brownian motion.
\subsection{Embedding the patch-graph}
For each realization of a specific signal model, we construct a
patch-set of $N = 1024$ maximally overlapping patches. The patch size
is given by $d =32$ for the time-frequency model, and $d=16$ for the
local regularity model. We compute the embedding $\Phi$
(\ref{eqn:lowdparam}) and keep $d'$ eigenvectors $\phi_k$.  Figure
\ref{fig:autocorrelations2} shows the patch-set associated with the
realization of the time-frequency signal displayed in Figure
\ref{fig:autocorrelations1} before (left) and after (right)
embedding. The scatterplot before embedding is computed using the
first three principal components. Figure \ref{fig:autocorrelations22}
shows the patch-set associated with the realization of the local
regularity signal displayed in Figure \ref{fig:autocorrelations11}
before (left) and after (right) embedding. The fast patches of the
time-frequency signal are the orange patches extracted from the high
frequency segments. The slow patches are the blue patches extracted
from the low frequency sections. Similarly, the fast patches of the
local regularity signal are the orange patches extracted from the
irregular segments, and the slow patches are the blue patches
extracted from the smooth sections. For both signals, the fast patches
are scattered across the space before embedding. After embedding, the
fast patches are aligned along smooth curves. This visual impression
is confirmed by computing the mutual distance between patches after
embedding, $\|\Phi(\bo x_n)-\Phi(\bo x_m)\|$. In principle, we should
report the value of the Lipschitz ratio%
\begin{figure}[H]
\centerline{
    \includegraphics[width=.3\textwidth]{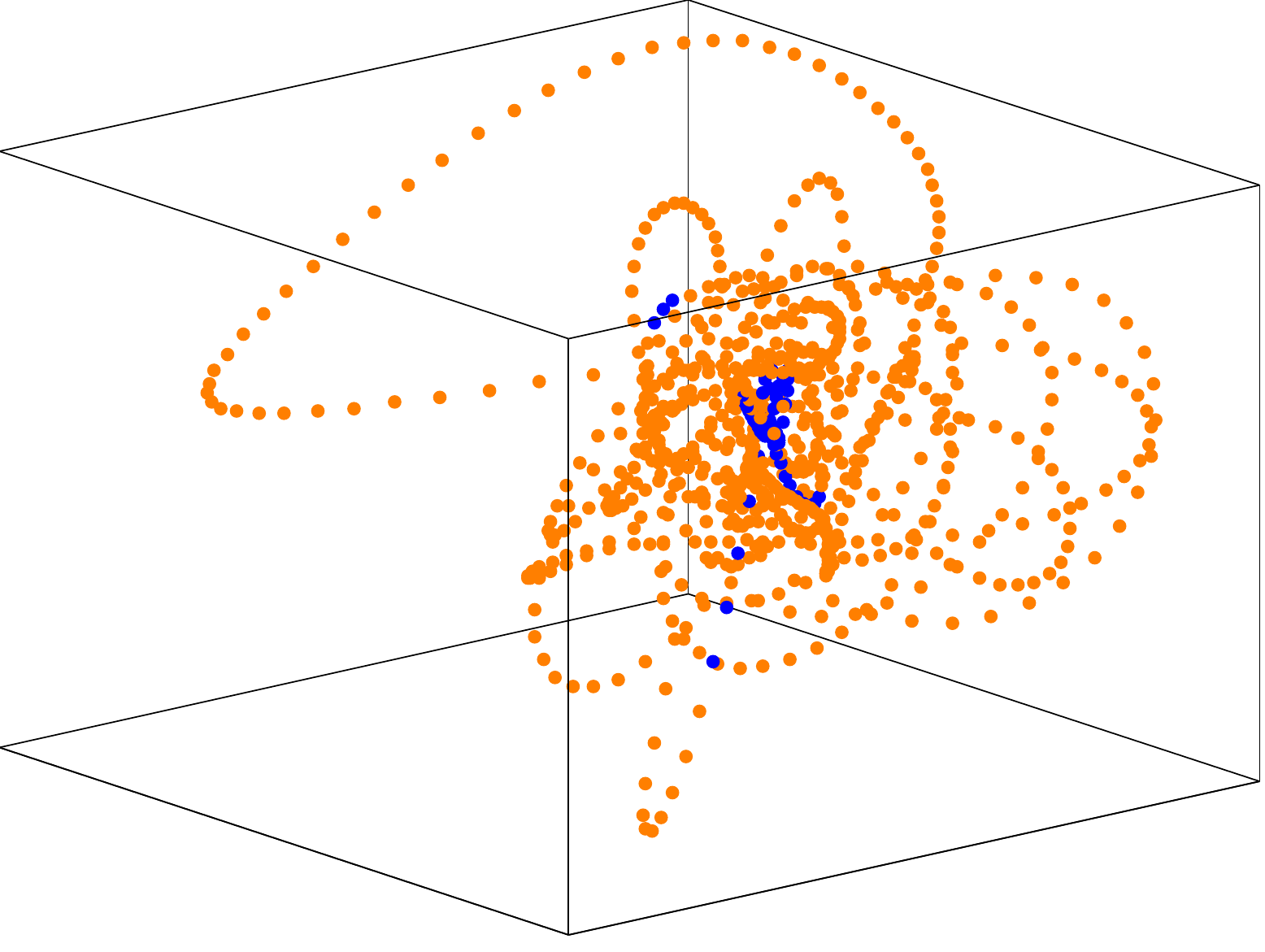} 
    \includegraphics[width=.3\textwidth]{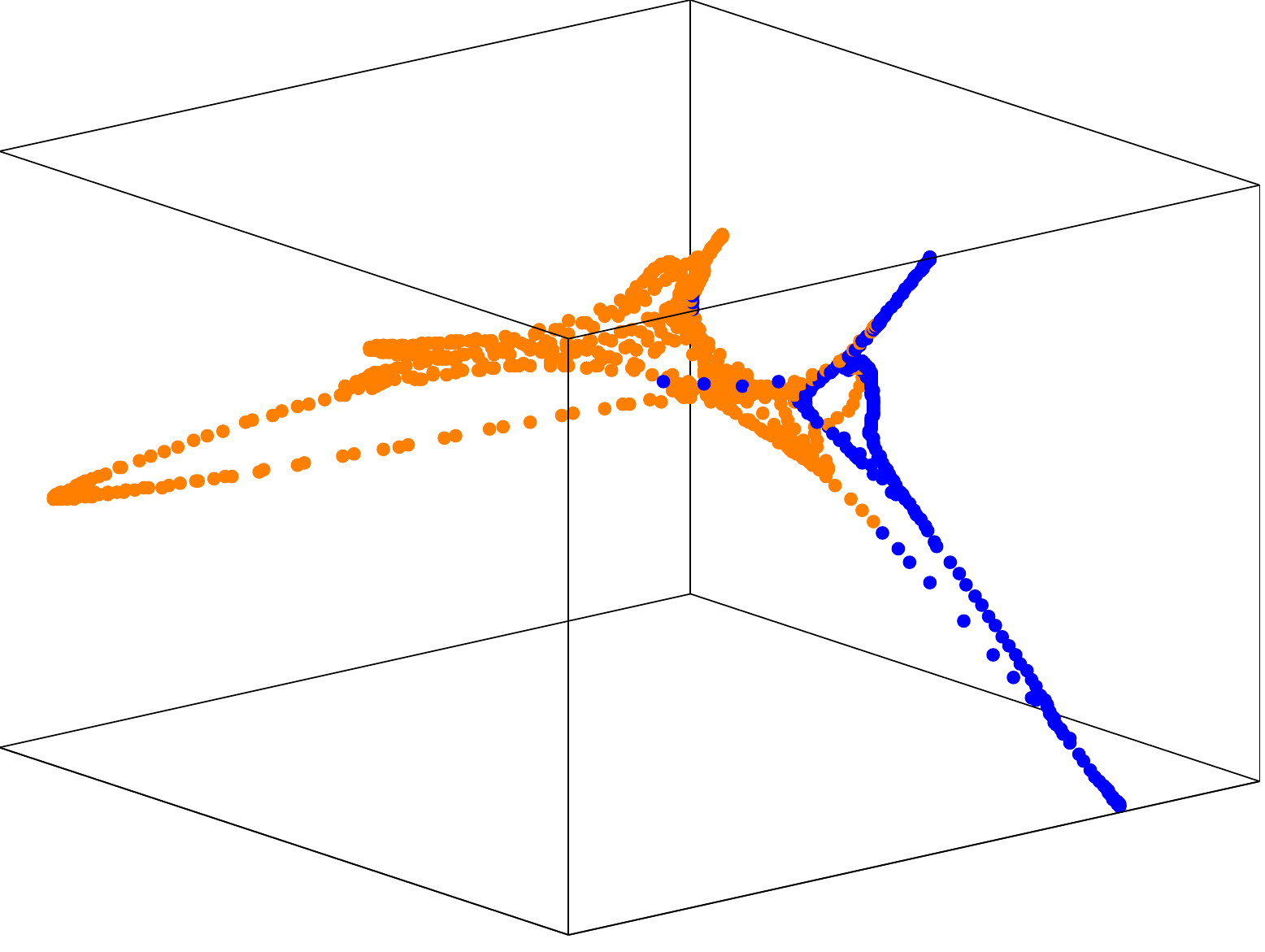} 
}
  \caption{Patch-set of the time-frequency signal (see Figure
    \ref{fig:autocorrelations1}) before (left) and after (right)
    embedding.  The color-code matches the color used in the plot of the
    signal: blue = low frequency, orange = high
    frequency. \label{fig:autocorrelations2} }
\end{figure}
\begin{figure}[H]
  \centerline{
    \includegraphics[width=.3\textwidth]{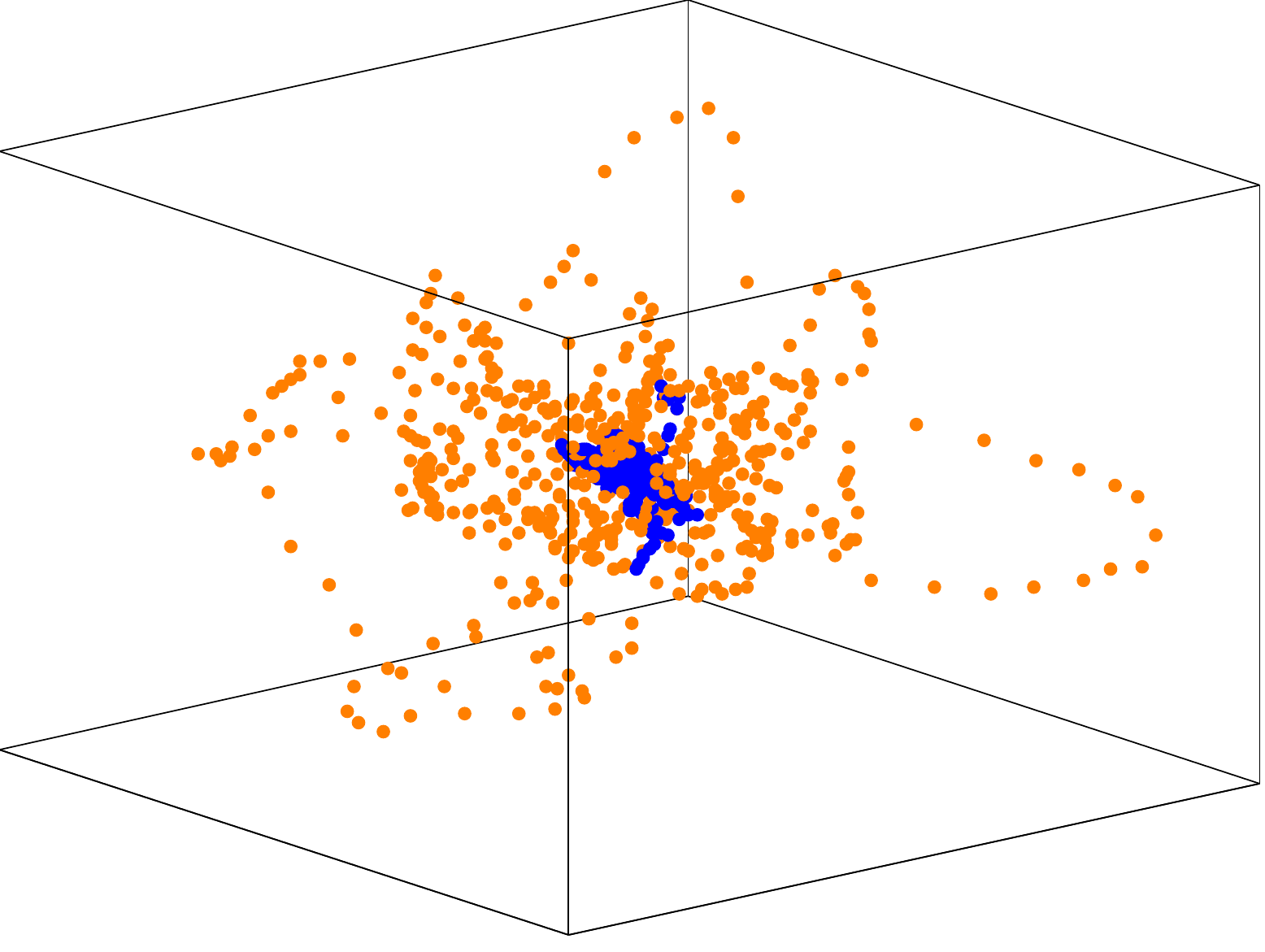} 
    \includegraphics[width=.3\textwidth]{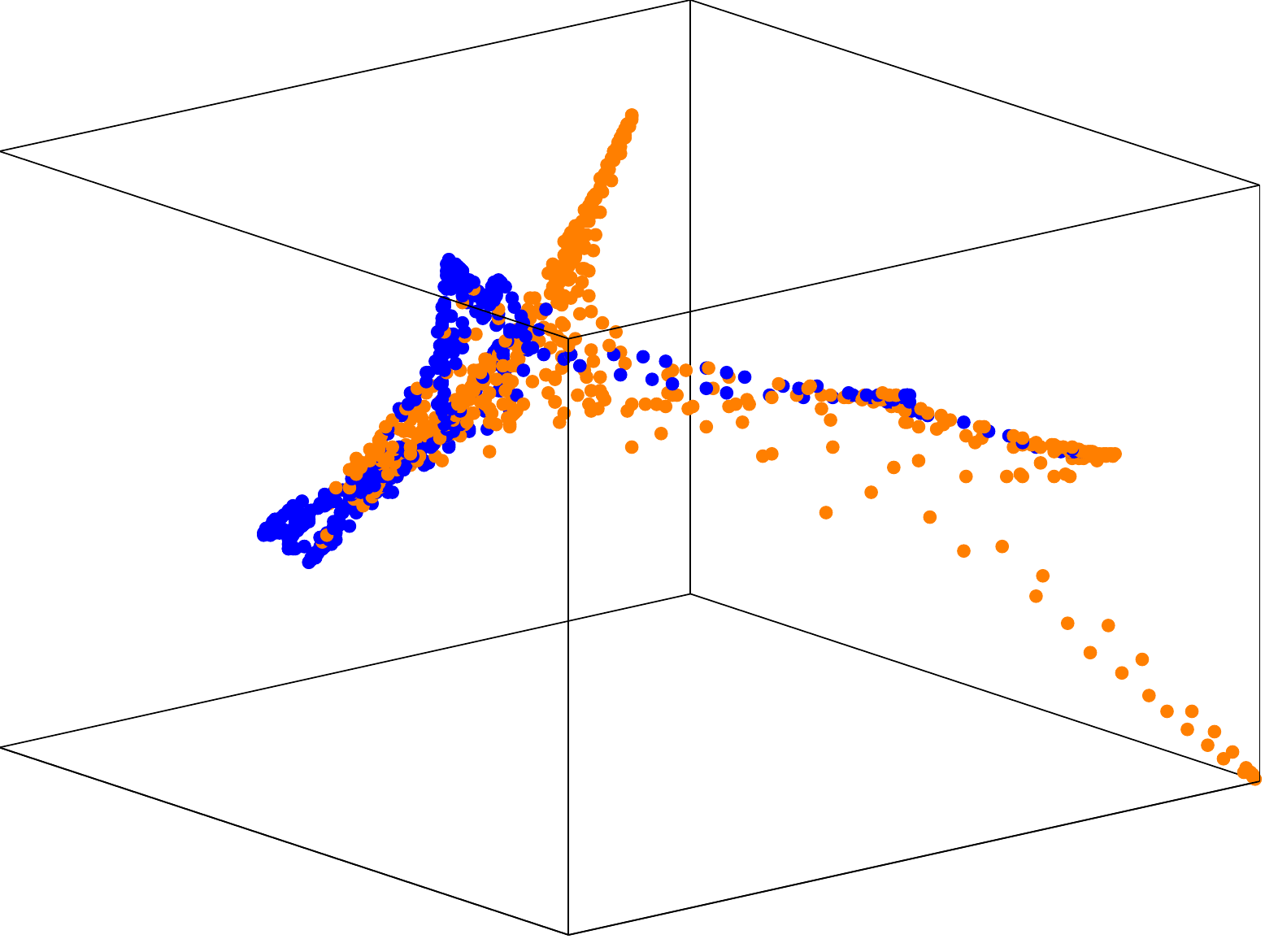} 
}
  \caption{Patch-set of the local regularity signal (see Figure
    \ref{fig:autocorrelations11}) before (left) and after (right)
    embedding.  The color-code matches the color used in the plot of the
    signal: blue = smooth, orange = irregular. \label{fig:autocorrelations22} }
\end{figure}
\noindent 

\begin{figure}[H]
  \centerline{
    \includegraphics[width=.4\textwidth]{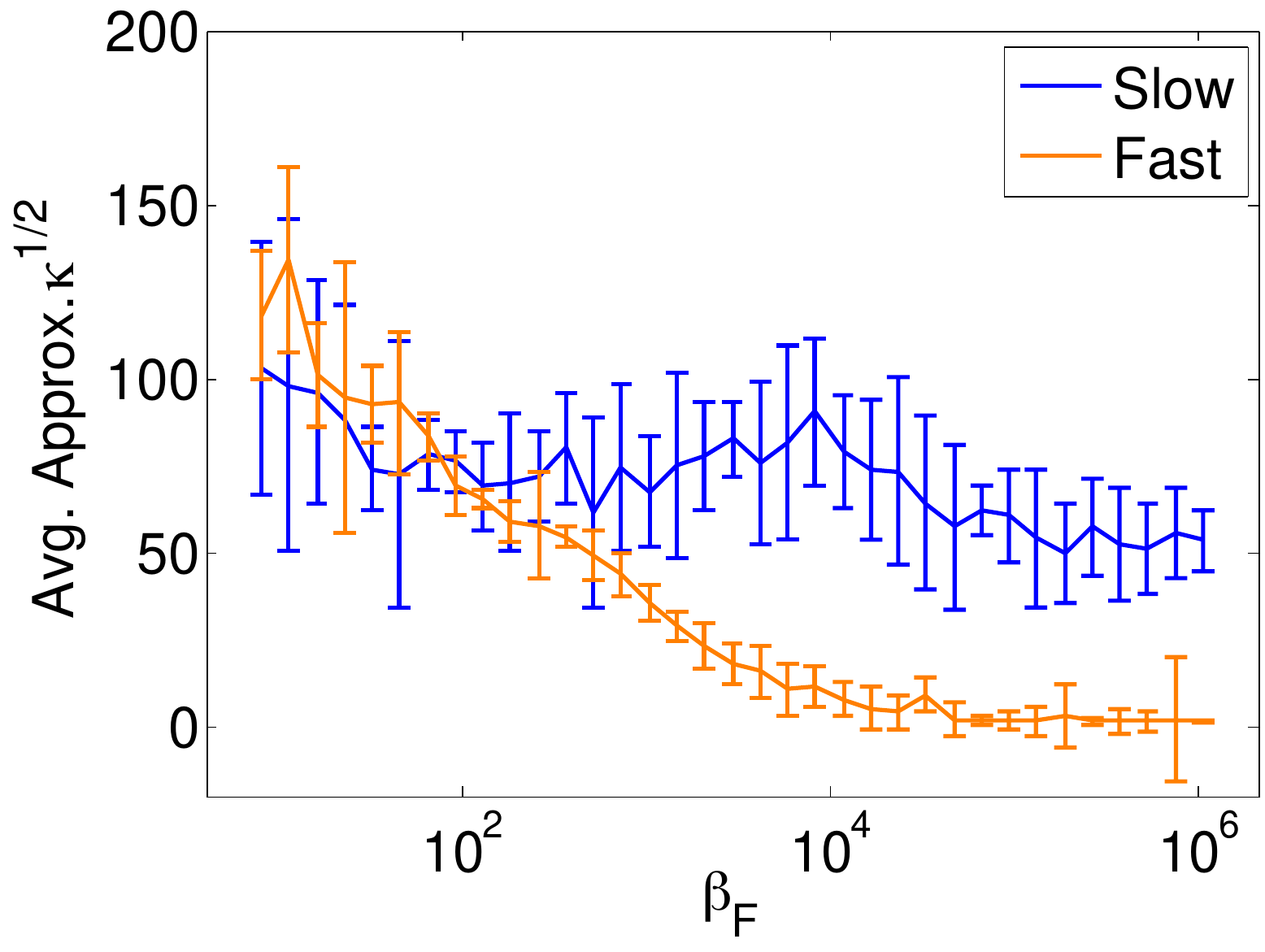} 
    \includegraphics[width=.4\textwidth]{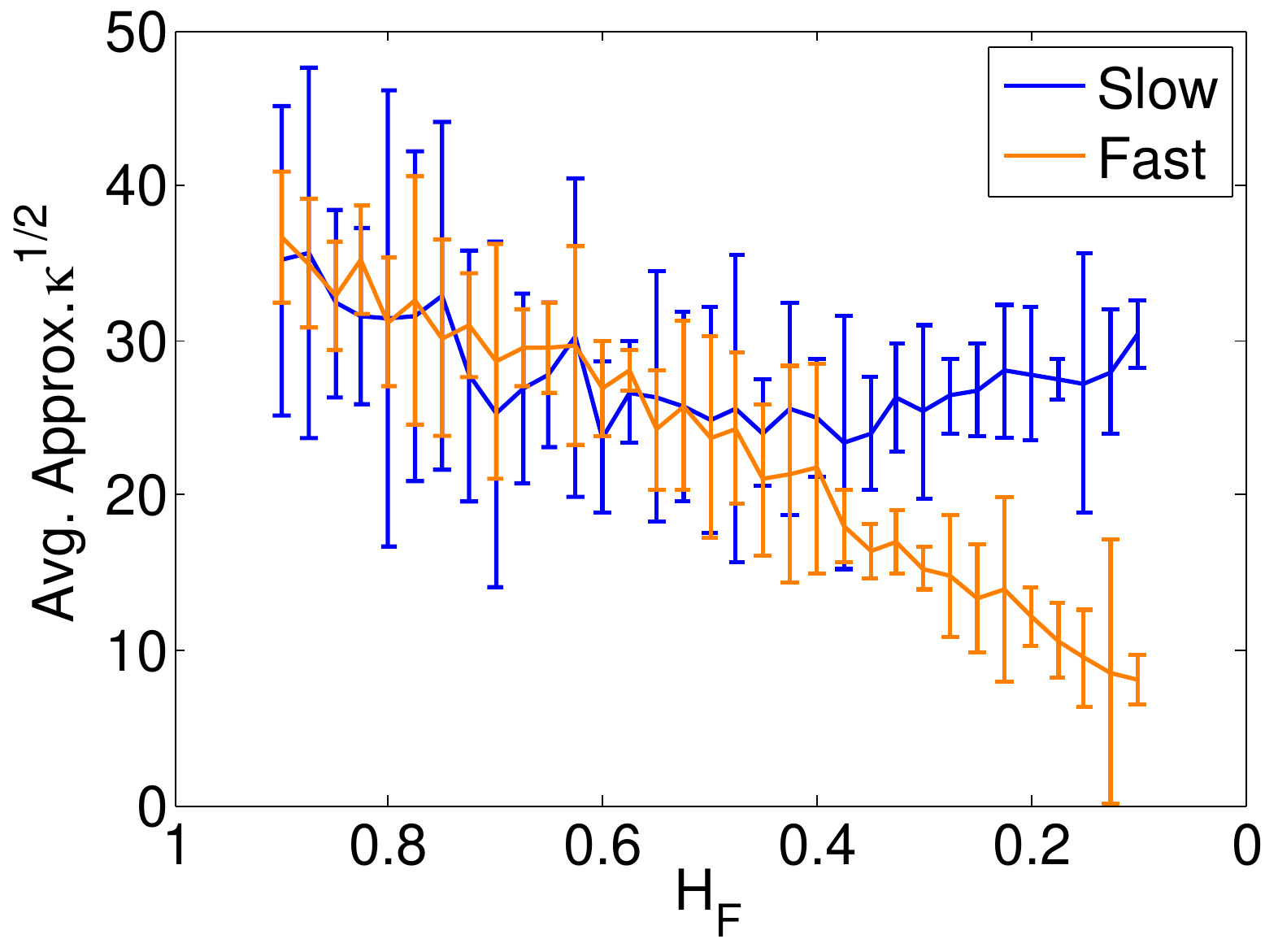} 
}
  \caption{$\sqrt{\kappa^\prime}$ for slow (blue) and fast patches (orange)
    for the time-frequency model (left) and the local regularity model
    (right) as a function of the ``roughness'' of the fast
    patches. The slow patches were generated using $\beta_{\sm S} = 8$
    (left) and $H_{\sm S} = 0.9$ (right).
    \label{fig:autocorrelations3}}
\end{figure}
\noindent $\|\Phi(\bo x_n)-\Phi(\bo
x_m)\|/\|\bx_n-\bx_m\|$ to quantify the contraction experienced
through the mapping $\Phi$.  However, we have noticed that because the
mutual distances $\|\bx_n-\bx_m\|$ between fast patches is always
large (as explained in section \ref{ssec:firstlookpatchgraph}), the
Lipschitz ratio ends up being always small for fast patches.
Therefore studying the size of the Lipschitz ratio associated with $\Phi$
does not reveal whether the map concentrates the fast patches or not,
but only indicates that the sampling of the fast patches (in the
patch-set) is coarse. For this reason we prefer to study how
$\|\Phi(\bo x_n)-\Phi(\bo x_m)\|$ varies for pairs of slow and fast
patches.  Based on our theoretical analysis, we expect that after the
embedding the mutual distance between fast patches will becomes much
shorter than the mutual distance between slow patches.

We point out that the eigenvectors $\phi_k$ used in the embedding
$\Phi$ (\ref{eqn:lowdparam}) are designed to have, on average, small
gradients (as measured along edges of the graph). Indeed, these
eigenvectors are also the eigenvectors of the graph Laplacian
\cite{chung97}, and therefore minimize a Rayleigh ratio that
quantifies the average norm of the gradient of $\phi_k$. Thus, if
we further restricted our computation of the commute times inside each
subset of fast and slow patches to only those patches that were
connected by an edge in the graph, we would expect to see smaller
values and little dependence on whether or not the patch was fast or
slow. However, since our theoretical analysis of section
\ref{ssec:theory} is based on the average commute time between all
vertices belonging to the fast or slow graph models, we choose to
compute the commute times between all patches, not just between
patches that are connected with an edge.

For each signal model, we compute the square root of the average approximate
commute time 
\begin{equation}
  \sqrt{\kappa^\prime} = \sqrt{\frac{2}{N(N-1)} \sum_{n< m} \|\Phi (\bx_n) - \Phi(\bx_m)\|^2}.
  \label{approx-commute2}
\end{equation}
for pair of patches $\bx_n,\bx_m$ that are either both fast,
or both slow patches. We study how $\kappa^\prime$ varies as a
function of the autocorrelation parameter that controls how irregular
the fast patches are. $\kappa^\prime$ was computed using ten
realizations of each signal model. The slow patches were generated
using $\beta_{\sm S} = 8$ and $H_{\sm S} = 0.9$. As before, we used $N
= 1024$ and $d = 32$ for the time-frequency model and $d = 16$ for the
local regularity model. We observed that the overall shapes of the
curves in (\ref{fig:autocorrelations3}) is invariant under variation
of the parameters (as along as the ratio of the patch length to the
average subinterval length remains less than 10\%). The dimension $d'$
of the embedding used to compute $\kappa^\prime$ was chosen so that
$(1-\lambda_{k})^{-1}<0.1(1-\lambda_{2})^{-1}$, for all
$k>d'+1$. Figure \ref{fig:autocorrelations3} shows
$\sqrt{\kappa^\prime}$ as a function of the frequency parameter
(left), and smoothness parameter (right). We note that as the signal
exhibits more rapid, local changes (increasing $\beta_{\sm F}$, or
decreasing $H_\sm{F}$), the associated fast patches are increasingly
concentrated (smaller $\|\Phi (\bx_n) - \Phi(\bx_m)\|$) through the
parametrization.  These experiments confirm that the theoretical
analysis can be applied to the true patch-set constructed from
realistic signals.%
\section{Discussion}
\label{sec:discuss}
Using realistic graph models, probabilistic arguments, and the
connection between the commute time of random walks on graphs and the
embedding (\ref{eqn:lowdparam}), we provided a theoretical explanation
for the success of the methods that analyze and process images based
on graphs of patches. Our results establish that the embedding of the
patch-graph of an image based on the commute time between vertices of
the graph reveals the presence of patches containing rapid changes in
the underlying signal or image by concentrating these patches close to
one another while leaving the patches extracted from the slowly
changing portions of the signal organized along low-dimensional
structures.
\subsection{Parameter selection}
\label{ssec:paramselect}
\subsubsection{Choosing the patch size}
In this work we are interested in the \textit{local} behavior of the
image, and therefore $d$ should remain of the order of what we
consider to be the local scale. We also note that as $d$ becomes
large, the number of available patches ($N/d$) becomes smaller, making
the estimation of the geometry of the patch-set more difficult, since
patches now live in high-dimension. Another consequence of the ``curse
of dimensionality'' is that the distance between patches becomes less
informative for large values of $d$. If the original signal is
oversampled with respect to the true physical processes at stake, then
one can coarsen the sampling of the patch-set in the image 
domain. In practice, it would be more advisable to coarsen the
underlying continuous patch-set, which is a nontrivial question.

\subsubsection{Choosing edge weights}
In general, two principles guide the choice of edge weights in the
patch-graph. On the one hand, patches that are very close should be
connected with a large weight (short distance), while patches that are
faraway should have a very small weight along their mutual edge. This
principle is equivalent to the idea of only trusting local distances
in $\real^d$. Such a requirement is intuitively reasonable if we
assume that the patch-set represents a discretization of a nonlinear
manifold in $\real^d$. In this situation, we know that when the points
on the manifold are very close to another, the \textit{geodesic
  distance} is well approximated by the Euclidean
distance. Conversely, because of the presence of curvature, the
Euclidean distance is a poor approximation to the geodesic distance on
the manifold when points are far apart. Because the only information
available to us is the Euclidean distance between patches, we should
not trust large Euclidean distances.

On the other hand, as observed in Section
\ref{ssec:firstlookpatchgraph}, the fast patches, which contain rapid
changes, are all very far apart (large $\rho^2(\bx_n,\bx_m)$). Therefore
the probability that the random walk escapes the fast patch $\bx_n$
and jumps to a different patch $\bx_m$, which is given by 
\begin{equation*}
\frac{w_{n,m}}{\sum_l w_{n,l}} = \frac{e^{\displaystyle -\rho^2(\bx_n,\bx_m)/\sigma^2}}{\sum_l w_{n,l}},
\end{equation*}
is always much smaller than the probability of staying at $\bx_n$,
which is given by
\begin{equation*}
\frac{1}{\sum_l w_{n,l}}.
\end{equation*}
In order to avoid that the random walk be trapped at each node $\bx_n$, we
``saturate'' the distance function by choosing $\sigma$ to be very
large. In this case, for all the nearest neighbors $\bx_m$ of $\bx_n$,
we have $w_{n,m} \approx 1$, and the transition probability is the
same for all the neighbors, $\bo P_{n,m} \approx 1/\nu$.  This choice of
$\sigma$ promotes a very fast diffusion of the random walk locally.
We note that choosing a large $\sigma$ may be avoided if
self-connections are not enforced (i.e. $w_{n,n} = 0$). However,
self-connections are a necessary technical requirement to prove that
the Markov process is aperiodic, which is required to prove the
equality (\ref{eqn:diffusiondist}) \cite{Coifman06a}.

We note that choosing $\sigma$ to be very large does not entirely
obliterate the information provided by the mutual distance between
patches, measured when patches are projected on the sphere with
$\rho(\bx_n,\bx_m)$, (\ref{rho}). Indeed, $\rho(\bx_n,\bx_m)$ is used
to select the nearest neighbors of each patch, and therefore allows us
to define a notion of a local neighborhood around each patch. Choosing
$\sigma$ to be very large forces a very fast diffusion within this
neighborhood, irrespective of the actual distances
$\rho(\bx_n,\bx_m)$. Alternatively, we could consider choosing
$\sigma$ to vary adaptively from one neighborhood to another. The
parameter $\sigma$ could be small when patches are extremely close to
one another, while $\sigma$ could be large when the patches are at a
large mutual distance of one another.  This notion is the foundation
of the self-tuning weight matrix, which adjusts its weights based on a
point's local neighborhood \cite{zelnik04tuning}.

\subsection{Extensions and generalizations}
\label{ssec:extension}
In general, the patch-set of an image consists of more than two
homogeneous subsets. For example, one could partition an image
patch-set into uniform patches, edge patches, and texture patches. Our
experience \cite{Taylor11} with a generalization of the time-frequency
signal model (section \ref{sec:experiments}) indicates that we can
still separate the patches when the signal is composed of up to four
different local behaviors (specified by four different values of the
parameter in the autocorrelation function). Another extension of this
work involves the embedding of a patch-set constructed from a library
of images. Recent studies \cite{Lee03} indicate that high-contrast
patches extracted from optical images organize themselves around
2-dimensional smooth sub-manifold (\cite{Carlsson08}). This idea has
also been exploited to construct dictionaries that lead to very sparse
representations of images (e.g. \cite{elad10}, and references
therein). Finally, we note that our results about the embedding of the
slow (\ref{slowgraph}), fast (\ref{fastgraph-def}), and the fused
graph (\ref{fused}) are very general and can be applied to datasets
\cite{cazals10} where the corresponding graph exhibits a similar
structure. For instance, one could imagine using this idea to study
social networks, where the concept of cliques would correspond to fast
subgraphs.

\subsection{Related work}
\label{ssec:relatedwork}

The concept of patches has proven extremely useful in many areas of
image analysis: texture analysis/synthesis \cite{CGF:CGF1407}, image
completion \cite{Mobahi:2009,zhou2009}, super-resolution
\cite{4694003}, and denoising
\cite{bougleux09,Buades05,Gilboa08,Katkovnik10,peyre08,SingerBoaz,Szlam08,zhou2009}. While
these references do not explicitly construct a patch-graph, these
works all compute distances between patches, and use the nearest
neighbors of each patch to analyze and process patches. Recent works
on the analysis of time-series also use patches and construct networks
of patches \cite{borges07,lac08,PhysRevLett.96.238701}. All these references provide experimental
evidence for the success of working on image (or signal) patches.

In this work, we provide a theoretical justification for this
experimental success. We study the effect of the embedding $\Phi$
(\ref{eqn:lowdparam}) on the organization of the patch-set. Our
analysis assumes that there exists a natural partition of the patch-set
into two classes: patches extracted from the smooth baseline and
patches that contain sudden local changes of the image intensity or
signal value. It is interesting to compare and contrast our work to
the work of Singer, Shkolnisky, and Nadler \cite{SingerBoaz} who
provide a different theoretical explanation for the success of
patch-based denoising algorithms. The authors in \cite{SingerBoaz}
treat the matrix $\bo P$ as a filter, which acts on an $N$-dimensional
column-vector-representation of the signal or the image. Each
multiplication of the probability distribution by $\bo P$ is
interpreted as the evolution of the diffusion process on the
patch-graph over a time-step of duration $\sigma$. The results in
\cite{SingerBoaz} rely on the convergence of a properly normalized
version of $\bo P$ toward the backward Fokker-Planck operator. The
authors can compute the eigenfunctions of the operator when the signal
is either a one-dimensional constant function perturbed by Gaussian
noise, or a one-dimensional step function also contaminated by
Gaussian noise.

In contrast, our analysis is based on the analysis of the commute time
on graphs that epitomize the patch-graph constructed from two classes
of patches. In addition, we need not assume that the image is piece-wise
constant. In fact, our experiments demonstrate that our analysis can
be applied to detect many different types of anomalies: changes in the
local frequency content, changes in local regularity,
etc. Furthermore, our theoretical analysis holds for finite values of
the number of patches $N$. It is interesting to note that Singer et
al. study the \textit{mean first-passage time} between patches
extracted from the noisy step function. The mean first-passage time is
derived from the hitting time, which is used to define the commute
time. The authors in \cite{SingerBoaz} use an energy argument to
explain the existence of a large mean first-passage time between
patches extracted from either side of the step function's
discontinuity. They argue that a high density of
patches is associated with a lower potential energy, and consequently
it will take longer for a random process to exit the well with such a
low potential. Finally, our results are not limited to patches of size
$d = 1$, as are the results in \cite{SingerBoaz}.

The energy argument in \cite{SingerBoaz} adds an interesting
interpretation to our analysis. Following this perspective, the
slow patches can be interpreted as points sampled from a probability
density function $P$ defined on $\real^d$ with a support that is
defined along a low-dimensional manifold. This localization leads to a
potential $U = -\log P$ with a deep and narrow well, from which the random
walk cannot escape. This argument agrees with our findings that the
average commute time between slow patches is very large, and thus, the
random walk spends considerably more time in the slow subgraph
before being able to reach a patch that is temporally faraway.

From a more general perspective, this work presents an investigation
into the diffusion process on the graphs models presented in Section
\ref{ssec:modelgraphs}. Our work is thus related to a large body of
work on the analysis of complex and random networks using
first-passage time (e.g. \cite{Condamin07} and references
therein). This area if usually motivated by physical problems such as
transport in disordered media, neuron firing, or energy flow on
power-grids instead of applications in signal processing.
\subsection{Open questions}
\label{ssec:unresolved}
While we obtained estimates for the average commute time on the fast
and slow graph models considered separately, it would be desirable to
obtain similar estimates on the fused graph. At the moment, our
analysis of the fused graph relies on numerical simulations.  We are
also aware of a small discrepancy in the upper bound on
$\kappa_{\sm{F}}$: this bound is increasing with $L$. In fact we
expect that the commute time on $\sm F(N,p)$ should decrease as $p$,
and therefore $L$, increases. The reason for this apparent
inconsistency is that the proof of (\ref{eqn:kappadc}) relies on a
loose upper bound for the effective resistance between two vertices,
which is provided by the geodesic distance on the graph
\cite{Chandra89theelectrical}. This is not a tight inequality on a
graph such as $\sm F(N,p)$. A more effective inequality, which could
improve the upper bound (\ref{eqn:kappadc}), relies on the computation of
the distribution of the number of paths $s$ of length at most $l$
between the two vertices. We could then use the fact that the commute
time is bounded from above by a constant times the ratio
$l/s$ \cite{Chandra89theelectrical}, which would decrease the
upper bound in (\ref{eqn:kappadc}).
\appendix
\section{The connectedness of the fast graph}
\label{sssec:connectedR}
It is necessary that the fast graph $\sm F(N,p)$ be connected to be
able to apply the spectral decomposition of the commute time.  To
ensure that the probability of $\sm F(N,p)$ being disconnected will
vanish as $N$ gets large, we must choose $Np>\log N$
\cite{durrettgd}. Since $p$ is defined as a function of $L$ in
(\ref{eqn:pdefined}), any requirement on $p$ ultimately constrains
$L$. First, because the maximum degree of a vertex in $\sm S(N,L)$ is
$2L+1$, according to (\ref{eqn:flatgraph}), we require

\begin{equation*}
  2L+1\leq N.
\end{equation*}
Manipulation of this inequality leads to 
\begin{equation*}
  \frac{L+1}{N}\leq \frac{1}{2}+\frac{1}{2N}.
\end{equation*}
We assume that $N\geq 2$, so that
\begin{equation*}
  \frac{L+1}{N}\leq \frac{3}{4}.
\end{equation*}
It follows that
\begin{equation*}
  \left(2-\frac{L+1}{N}\right)\geq \frac{5}{4}>1.
\end{equation*}
Therefore, rewriting (\ref{eqn:pdefined}) and using the last inequality we have 
\begin{equation*}
  p = \frac{L}{N-1}\left(2-\frac{L+1}{N}\right) >
  \frac{L}{N}\left(2-\frac{L+1}{N}\right) >\frac{L}{N}.
\end{equation*}
Therefore, choosing $L = c\log N$ for some $c>1$ ensures that $Np>\log
N$, and consequently, the probability of $\sm F(N,p)$ being
disconnected approaches zero as $N$ approaches infinity.
\section{Bounding the commute times in the graph models}
\label{sec:ctproofs}
\subsection{Proof of the lower bound on the average commute time in the slow graph}
\label{sssec:slowgraphproof}
In order to compute a lower bound on the average commute time, we
consider a fixed pair of vertices in the slow graph, $\bx_{n_0}$ and
$\bx_{m_0}$, and compute a lower bound on the commute time $\kappa
(\bx_{n_0}, \bx_{m_0})$. We can then compute the average of this lower
bound over all the pairs of vertices. To obtain the lower bound on
$\kappa (\bx_{n_0},\bx_{m_0})$ we use a standard tool to obtain lower
bounds on commute time: the Nash-Williams inequality
\cite{lyons11}. The Nash-Williams inequality is usually formulated in
terms of electrical networks. We prefer to present an equivalent
formulation that is directly adapted to our
problem.  We first introduce the concept of {\em edge-cutset}.\\

\begin{definition}
Let $V_1$ and $V_2$ be two disjoint sets of vertices. A set of edges
$E$ is an edge-cutset separating $V_1$ and $V_2$ if every path that
connects a vertex in $V_1$ with a vertex in $V_2$ includes an edge in $E$.
\end{definition}

Given a weighted graph, which may contain loops, we define a random
walk with the probability transition matrix $\bo P_{n,m} =
\bo W_{n,m}/\bo D_{n,n}$. Let $\bx_{m_0}$ and $\bx_{n_0}$ be two vertices. The commute
time between vertices $\bx_{m_0}$ and $\bx_{n_0}$, $\kappa(\bx_{m_0},\bx_{n_0})$ satisfies the
following lower bound.\\

\begin{lemma}[Nash-Williams]
If $\bx_{m_0}$ and $\bx_{n_0}$ are distinct vertices in a graph that are separated by
disjoint edge-cutsets $E_k, k =1,\ldots$, then
\begin{equation}
V \;\sum_k \left [\sum_{\{\bx_{n},\bx_{m}\} \in E_k} w_{n,m} \right]^{-1}
\leq \kappa(\bx_{m_0},\bx_{n_0}) \qquad \text{where $\{\bx_{m},\bx_{n}\}$ is an edge in the
  cutset $E_k$},
\end{equation}
\end{lemma}
and where the volume of the graph is defined by $V=\sum_{i=1}^N
\sum_{j=1}^N w_{i,j}$.  

We now exhibit a sequence of edge-cutsets in the slow graph. We refer
to Figure \ref{fig:matrixtile} for the construction of the cutsets. 
We define the first cutset $E_1$. If $m_0 < L$, then $E_1$ needs a little more
attention and is defined as the set of $L$ edges $\{\bx_{i},\bx_{j}\}$, where $i$ and $j$ are defined by
\begin{equation}
\begin{cases}
i = 1, \ldots, m_0,\\
j = m_0 + 1, \ldots, L +i.
\end{cases}
\end{equation}
The edge-cutset $E_1$ is shown in the Figure \ref{fig:matrixtile} for
$m_0 =1$ (left) and $m_0=2$ (center), for $L=3$. The removal of this
set of edges prevents $\bx_{m_0}$ from being connected to
$\bx_{n_0}$. Indeed, the self loop on the diagonal (green entry) does
not allow the random walk to move toward $\bx_{n_0}$. This can be also
be visualized in Figure \ref{fig:bluegraph}, where $E_1$ is the
leftmost set of edges that connect $\bx_{m_0}$ to that part of the
graph that is connected to $\bx_{n_0}$. The sum of edge weights in
$E_1$ is at most $L(L+1)w_{\sm{S}}/2$.  If $m_0 \ge L$, then $E_1$, is
defined as the other generic edge-cutsets.

We now define the generic edge-cutsets $E_k$ as the set of $L(L+1)/2$
edges $\{\bx_{i},\bx_{j}\}$ such that
\begin{equation}
\begin{cases}
i = m_0 + 1 + (k-2)L, \ldots, m_0 + (k-1)L,\\
j = m_0 + 1 + (k-1)L, \ldots, L + i.
\end{cases}
\end{equation}
As seen in Figure \ref{fig:matrixtile}-right for $k=3$, setting the
entries of $E_3$ to zero disconnects the upper and lower part of the
submatrix $\bo W(m_0:n_0,m_0:n_0)$, thereby isolating $\bx_{m_0}$ and
$\bx_{n_0}$. Alternatively, we also see in Figure \ref{fig:bluegraph}
that any path from $\bx_{m_0}$ to $\bx_{n_0}$ needs to go through
$E_3$.  Each edge-cutset $E_k, k \ge 2$ is a triangle with a height of
size $L$. Therefore, after creating $E_1$, we can fit
$\floor{\frac{n_0-(m_0+1) +1}{L}}$ such cutsets between $\bx_{m_0+1}$
and $\bx_{n_0}$. The sum of the weights along the edges of each cutset
$E_k, k =2,\ldots$ is given by $L(L+1)w_s/2$.  In addition, 
the sum of edge weights in the first cutset $E_1$ is at most
$L(L+1)w_s/2$.  Putting everything together, the computation of the
lower bound using the Nash-Williams Lemma yields
\begin{equation*}
\begin{split}
V  \sum_k & \left [\sum_{\{n,m\} \in E_k} w_{n,m} \right]^{-1}\\  
&\geq \left[ 
  N (2L +1 ) - L (L+1)
\right] w_s 
\left (
  \floor{\frac{(n_0 - m_0)}{L}} \frac{2}{L(L+1)w_s} +
  \frac{2}{L(L+1)w_s}
\right) \\
& \ge \frac{\left[N (2L +1 ) - L (L+1) \right]}{L(L+1)}
\left (
2\left(\frac{n_0 - m_0}{L} - 1\right) + 2
\right)\\
& \ge \frac{\left[N (2L +1 ) - L (L+1) \right]}{L(L+1)}
\left (
2\frac{n_0 - m_0}{L}
\right)
\end{split}
\end{equation*}
We can summarize this result in the following lemma.\\
\begin{lemma} 
 The commute time between vertices $\bx_{n_0}$ and $\bx_{m_0}$ inside
  $\sm S(N,L)$ satisfies 
  \begin{equation} 
    \kappa(\bx_{m_0},\bo
    x_{n_0})\geq
    \frac{2\left[N(2L+1)-L(L+1)\right]}{L(L+1)}\left(\frac{n_0-m_0}{L}\right).
    \label{eqn:kappagenerallb} 
  \end{equation} 
  \label{lemma:lemma82}
\end{lemma}
Finally, we bound the average commute time in the slow graph.  Observe
that the slow graph model $\sm S(N,L)$ has $N-j$ pairs of vertices
such that $|m-n| = j$, for $j = 1,\ldots,N-1$. Therefore, using the
lower bound given in Lemma
\ref{lemma:lemma82} it follows that
\begin{equation*}
  \sum_{1\leq m<n\leq N}\kappa(\bx_m,\bx_n)
  \geq \frac{2\left[N(2L+1)-L(L+1)\right]}{L^2(L+1)} 
  \sum_{j = 1}^{N-1}(N-j)j
\end{equation*}
But
\begin{align*}
  \sum_{j = 1}^{N-1}(N-j)j
& =  \left(N \sum_{1}^{N-1} j - \sum_{1}^{N-1} j^2\right)
 = \left (\frac{N^2(N-1)}{2}
  -\frac{N(N-1)}{2}\frac{2N-1}{3}\right)\\
& = \frac{N(N-1)}{2}\frac{N+1}{3} 
\end{align*}
Dividing both sides by $N(N-1)/2$ and simplifying yields (\ref{eqn:kappafc}).
\begin{figure}[H]
    \begin{center}
      \includegraphics[width=.33\textwidth]{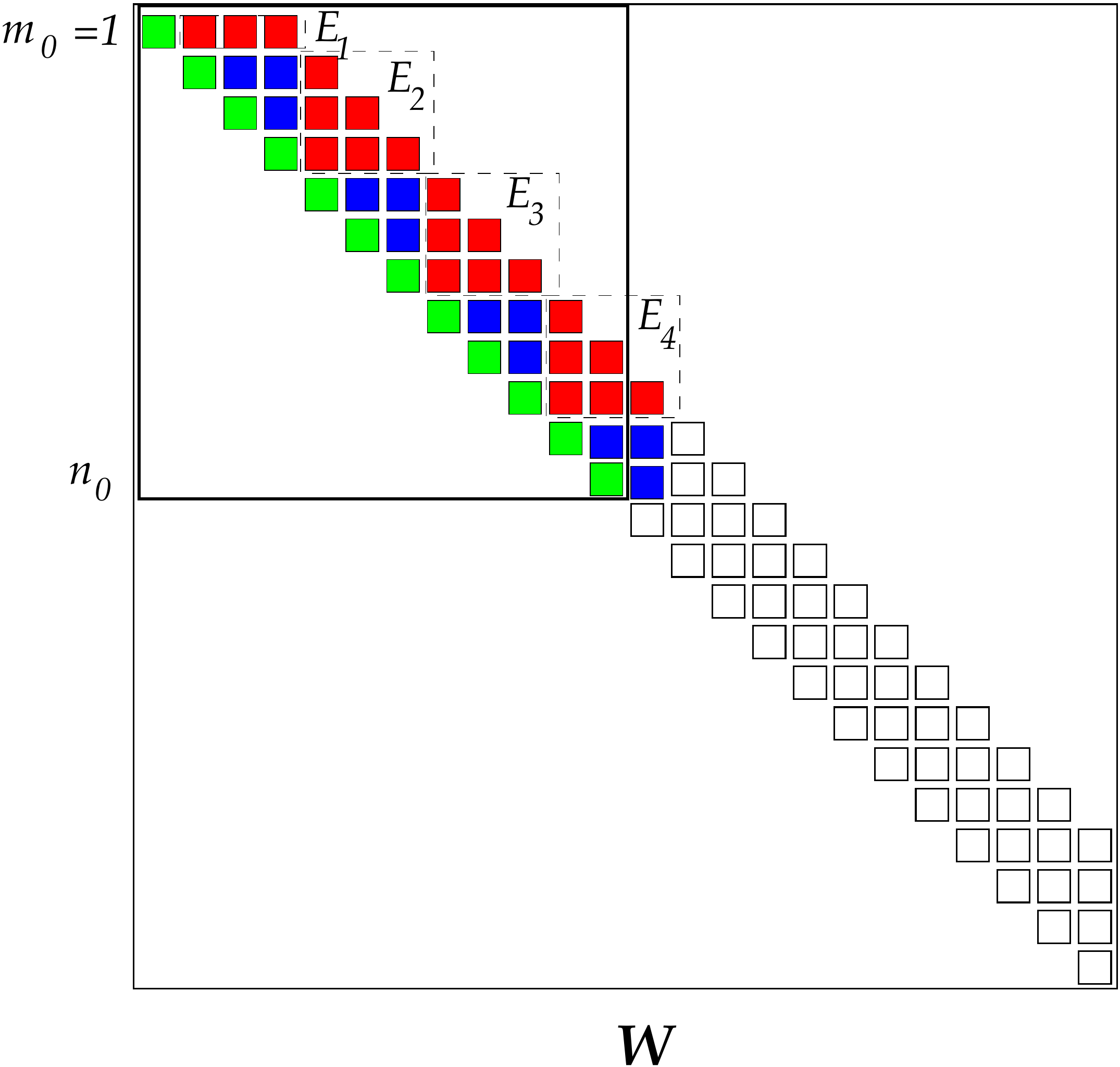} 
      \raisebox{-2mm}{\includegraphics[width=.33\textwidth]{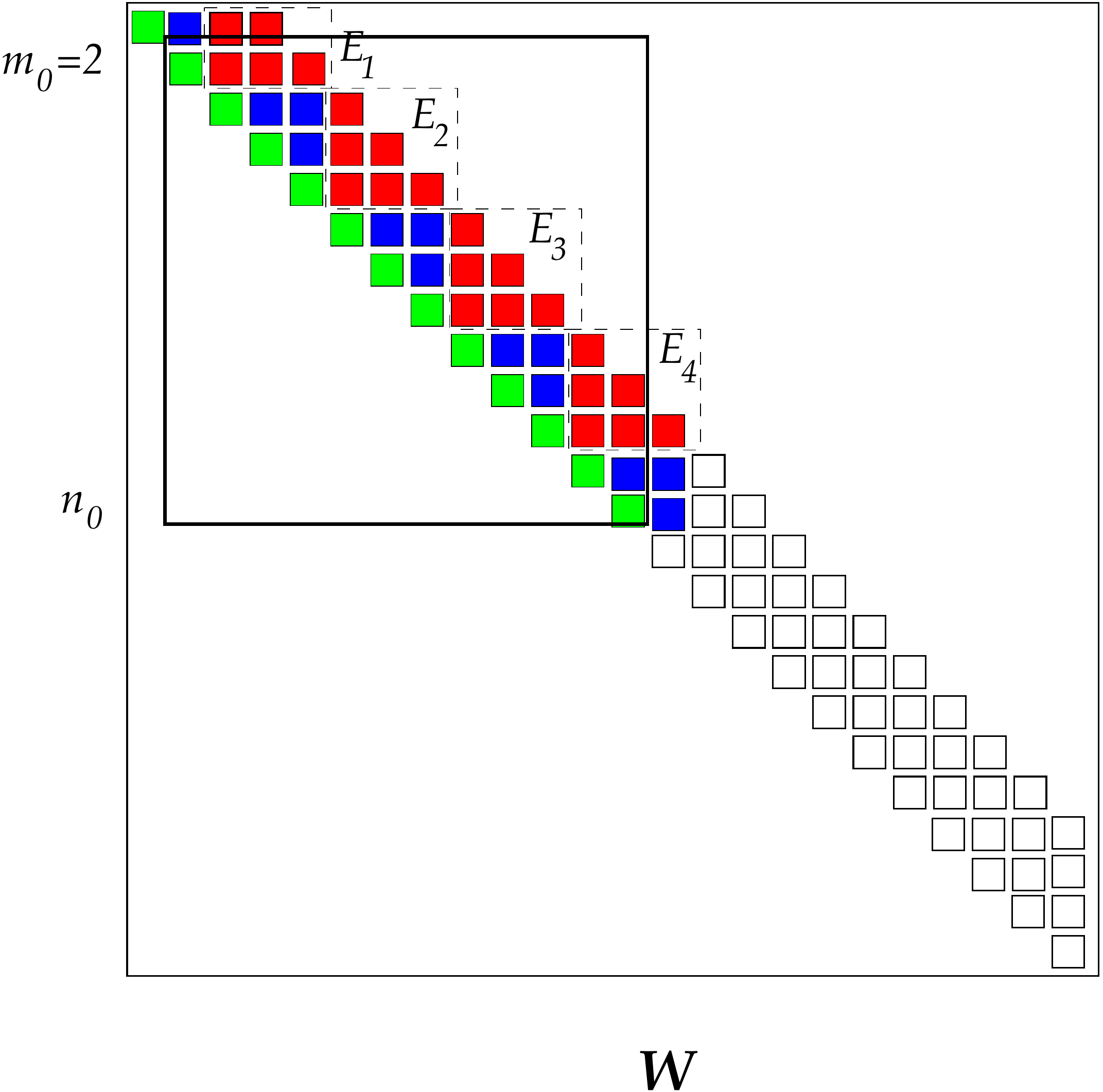}}
     \includegraphics[width=.292\textwidth]{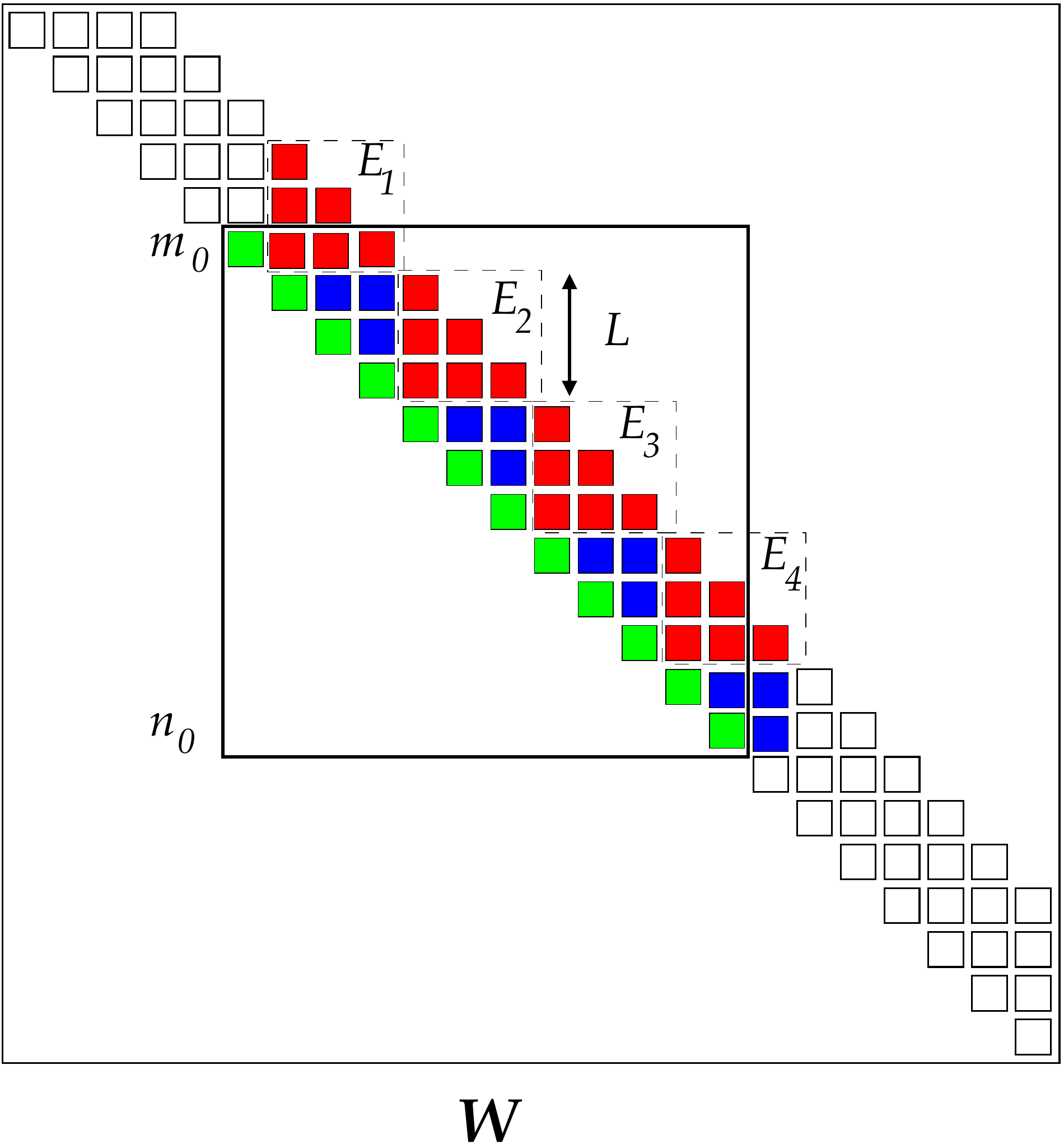} 
  \end{center}
  \caption{The small squares represent the nonzero entries in the
    upper triangular portion of the weight matrix $\bo W$ of $\sm
    S(N,L)$. The green entries on the diagonal are the self-loops. The
    edge-cutsets $E_k$ are shown in red for $m_0=1$ (left), $m_0=2$
    (center), and for $m_0 \ge L$ (right). The submatrix $\bo
    W(m_0:n_0,m_0:n_0)$ is also shown.\label{fig:matrixtile}}
  \end{figure}
 \begin{figure}[H]
    \begin{center}
      \includegraphics[width=.6\textwidth]{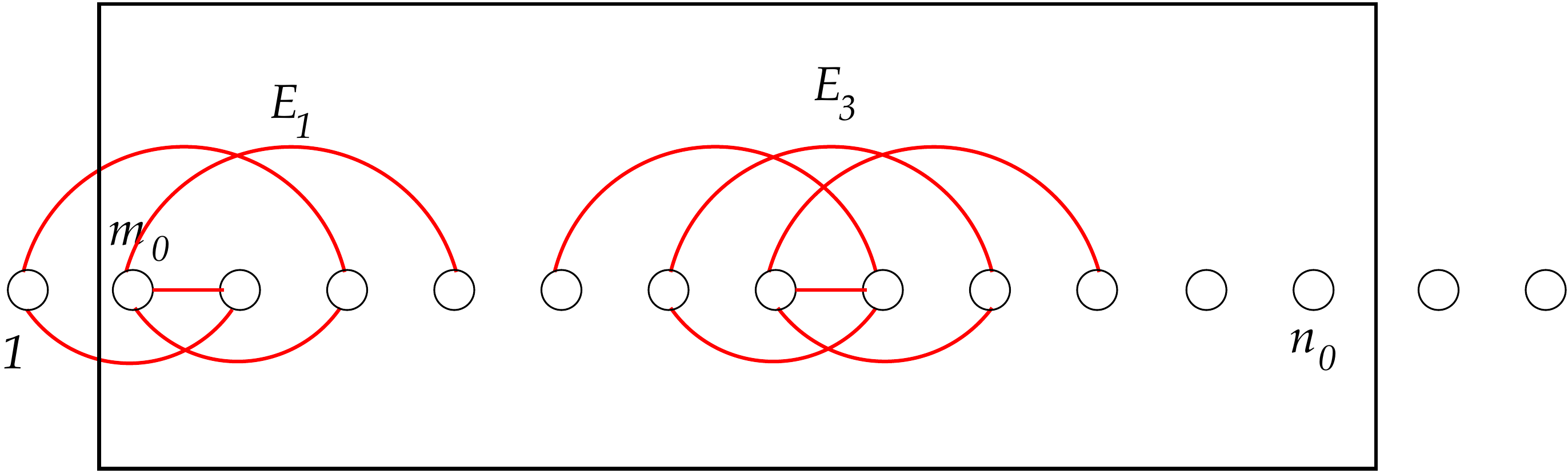} 
      \includegraphics[width=.6\textwidth]{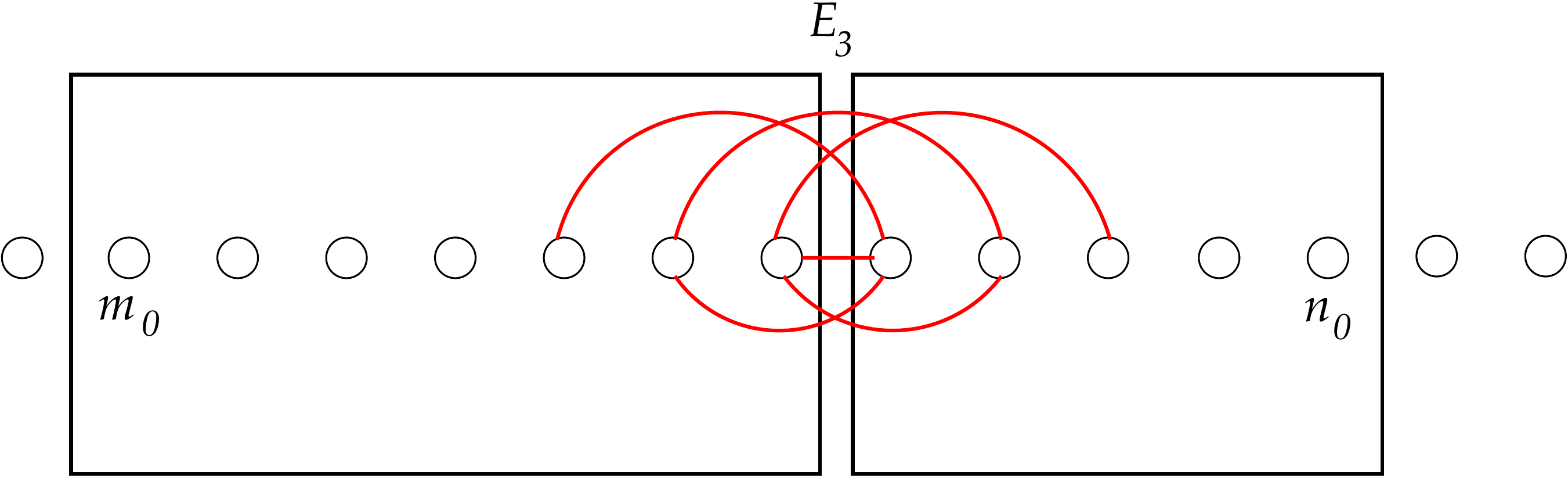} 
    \end{center}
    \caption{Top: edge-cutsets $E_1$ and $E_3$. Bottom: any path from
      $m_0$ to $n_0$ needs to use an edge of the
      edge-cutset $E_3$. \label{fig:bluegraph}}
  \end{figure}
\noindent 
\subsection{Proof of upper bound on the average commute time in the fast graph}
\label{sssec:fastgraphctproof}
Our approach relies on the relationship between electrical networks
and random walks on graphs \cite{2000math......1057D}.  We begin by
introducing the property of interest --- the \textit{effective
  resistance} --- and its relationship to the commute time.
\paragraph{The electrical network perspective} For each pair of
vertices $\bx_{n}$ and $\bx_{m}$ with a non zero weight $w_{n,m}$, we
assign the resistance
\begin{equation}
  r_{n,m} = \frac{1}{w_{n,m}}
  \label{eqn:resistanceandweight}
\end{equation}
to the edge $\{\bx_{n},\bx_{m}\}$. We note that if $w_{n,m}=0$, then there is no
connection between $\bx_n$ and $\bx_m$, and no resistance to consider.
Now, consider applying a potential difference, or voltage, across the
vertices $\bx_{m_0}$ and $\bx_{n_0}$. As a result, some current flows
across the resistors (edges) in the electrical network (graph). We may
replace the set of resistors across which some current flows by an
equivalent, {\em effective resistance}, $R_{m_0,n_0}$ that is
connected between $\bx_{m_0}$ and $\bx_{n_0}$. The effective
resistance $R_{m_0,n_0}$ is defined by the voltage necessary to
maintain a one-unit current between $\bx_{m_0}$ and $\bx_{n_0}$.  The
main result in \cite{Chandra89theelectrical}, is that the commute time
between vertices $\bx_{m_0}$ and $\bx_{n_0}$ can be expressed as
\begin{equation}
  \kappa(\bx_{m_0},\bx_{n_0}) = V R_{m_0,n_0}.
  \label{eqn:commuteitoR}
\end{equation}

Taking expectations of both sides of
Equation (\ref{eqn:commuteitoR}) with respect to the process of
generating edges and choosing terminals in a fast graph, we obtain

\begin{align}
  \kappa_{\sm{F}} = \E(V)\,\E(R)+\cov(V,R).
  \label{eqn:kappaFappend}
\end{align}
Notice that every edge in the fast graph has weight $w_\sm{F}$.  Therefore, $V$ can be expressed as
\begin{equation}
  V = \sum_{n=1}^Nw_{nn} + 2\sum_{1\leq m<l\leq N}w_{ml}= w_\sm{F}N +
  2w_\sm{F} \tilde{N},
\label{eqn:volume2}
\end{equation}
where $\tilde{N}$ is a binomial random variable representing the
number of edges connecting distinct vertices in the fast graph. We now
rewrite (\ref{eqn:kappaFappend}), using (\ref{eqn:volume2}) and the
assumption that $\cov(V,R)\leq 0$, to obtain

\begin{equation*}
  \kappa_{\sm{F}} \leq w_\sm{F}\left[N + 2\E(\tilde{N})\right]\E(R).
\end{equation*}

Recall that $\tilde{N}$ is distributed as a binomial random variable
with parameters $(N(N-1)/2,p)$. Also, the effective resistance between
two nodes of a network is at most the geodesic distance between them,
$\delta$, scaled by $1/w_{\sm F}$ \cite{Chandra89theelectrical}. It
follows that

\begin{equation*}
  \kappa_{\sm{F}} \leq \left[N(N-1)p+N\right]\E(\delta).
\end{equation*}

The authors \cite{PhysRevE.70.056110} give a closed form expression
for $\E(\delta)$ on Erd\"os-Renyi graphs, which we can utilize since
the fast graph's self-connections do not change the geodesic distance.
This yields

\begin{equation*}
  \kappa_{\sm{F}} \leq \left[N(N-1)p+N\right]\left[\frac{\log N - \gamma_e}{\log ((N-1)p+1)} +\frac{1}{2}\right],
\end{equation*}
where $\gamma_e \approx 0.5772$ is Euler's constant. Simplification using (\ref{eqn:pdefined}) gives the desired result.

\paragraph{Remark} Although $\cov(V,R)\leq 0$ is an assumption, we
conjecture that it is always satisfied due to the fact that increasing
the number of resistors $M$ in an electrical network with a fixed
number of nodes is effectively like adding resistors-in-parallel, and,
according to Rayleigh's Monotonicity Law, adding edges (increasing $M$)
can only decrease the effective resistance \cite{2000math......1057D}.

\section{Generating a random trigonometric polynomial with a specified autocorrelation}
\label{ssec:generatingzk}
Let $z(t)$ represent a random trigonometric polynomial on $[0,1)$ with an autocorrelation function given by 

\begin{equation}
  C(\tau) = 2(\cos(\pi \tau))^{2\beta}-1\hspace{.2in}\text{ for } \tau\in \left[-\frac{1}{2},\frac{1}{2}\right),
  \label{eqn:autocorrelationapdx}
\end{equation}
for some nonnegative integer $\beta$.  It follows that we can do a Fourier expansion of $C(\tau)$ to obtain
\begin{equation}
  C(\tau) = \sum_{j\in\mathbb{Z}}\hat{C}_je^{2\pi i j t},
  \label{eqn:chat}
\end{equation}
where $i = \sqrt{-1}$ and 
\begin{align*}
  \hat{C}_j &= \int_0^1C(\tau)e^{-2\pi i j \tau} d\tau\\
  &= \int_0^1\left(2^{(1-2\beta)}\left(\sum_{k =
        0}^{2\beta}\begin{pmatrix}2\beta \\ k\end{pmatrix}e^{2\pi
        i(\beta-k)\tau}\right)-1\right)e^{-2\pi i j \tau} d\tau\\ 
  &= 2^{(1-2\beta)}\sum_{k=0}^{2\beta}\begin{pmatrix}2\beta \\
    k\end{pmatrix}\int_0^1e^{2\pi i
    (\beta-k-j)\tau}\,d\tau-\int_0^1e^{-2\pi i j \tau}\,d\tau\\  
  &= \left\{\begin{array}{ll} 2^{(1-2\beta)}\begin{pmatrix}2\beta \\
        \beta\end{pmatrix}-1&\text{if } j = 0,\vspace{.05in}\\ 
      2^{(1-2\beta)}\begin{pmatrix}2\beta \\ \beta-j\end{pmatrix}& \text{if
      }|j|\leq \beta,\\ 0 & \text{if }j>\beta,\end{array}\right. 
\end{align*}

where the second equality follows after expressing cosine with
complex exponentials, and applying the binomial theorem.

It is clear that $2\beta$ is the frequency of the fastest sinusoid
making up the random signal $z(t)$, and that most of the energy is on
average at frequency $\beta$. Let $A_j$ and $B_j$ be independent and
identically distributed Normal random variables with zero mean and
unit variance. Define

\begin{equation*}
  \hat{z}_j = \sqrt{\frac{\hat{C}_j}{2}}\left(A_j+iB_j\right).
\end{equation*}
Finally, the signal $z(t)$ is defined as
\begin{equation*}
  z(t) = \sum_{j\in\mathbb{Z}}\hat{z}_je^{2\pi i j t}.
\end{equation*}
To check that the signal $z(t)$ defined above has the correct
autocorrelation, observe that linearity of the expectation,
independence and zero mean of the random variables, and the fact that
$\hat C_{j} = \hat C_{-j}$ together imply that

\begin{align*}
  \E(z(t)& \overline{z(t+\tau)}) =
  \sum_{|j|\leq2\beta}\sum_{|k|\leq2\beta}\E\left(\hat{z}_j\overline{\hat{z}_k}\right)e^{-2\pi
    i k \tau}e^{2\pi i(j-k) t}\\ 
  &= \sum_{|j|\leq2\beta}\sum_{|k|\leq2\beta}\frac{\sqrt{\hat C_j \hat
      C_k}}{2}\left[\E(A_jA_k)-i\E(A_jB_k)+i\E(A_kB_j)+\E(B_jB_k)\right]e^{-2\pi
    i k \tau}e^{2\pi i(j-k) t}\\ 
  &=\sum_{|j|\leq2\beta}\frac{\hat
    C_j}{2}\left[\E(A_j^2)+\E(B_j)^2\right]e^{-2\pi i j \tau} =\sum_{|j|\leq2\beta}\hat C_je^{2\pi i j \tau}.
\end{align*}
Therefore, referencing (\ref{eqn:chat}), it follows that $\E(z(t)\overline{z(t+\tau)}) = C(\tau)$.

\end{document}